\documentclass[12pt]{article}
\usepackage[margin=1.15in]{geometry}
\usepackage[utf8]{inputenc}
\usepackage{amsmath}
\usepackage{amsfonts}
\usepackage{amssymb}
\usepackage{graphicx}
\usepackage{bbm}
\usepackage{multicol}
\usepackage{sectsty}
\usepackage{enumerate}
\usepackage{epstopdf}
\usepackage{amsthm}
\usepackage[titletoc,title]{appendix}
\usepackage{tabularx,ragged2e,booktabs,caption}
\usepackage{mathtools}
\usepackage{xcolor, color}
\usepackage{chngcntr}
\usepackage{natbib}
\usepackage{mathrsfs}
\usepackage{hhline}
\usepackage{booktabs}
\usepackage{pdflscape}
\usepackage{multirow}
\usepackage{makecell}
\usepackage{authblk}
\usepackage{setspace}
\usepackage[final]{changes}
\definecolor{mygray}{gray}{0.9}

\numberwithin{equation}{section}
\counterwithout{equation}{section}

\newcommand{\stkout}[1]{\ifmmode\text{\sout{\ensuremath{#1}}}\else\sout{#1}\fi}
\setdeletedmarkup{\stkout{#1}}
\theoremstyle{plain}
\newtheorem{thm}{Theorem}[section]
\newtheorem{cor}{Corollary}[section]
\newtheorem{lemma}{Lemma}[section]

\theoremstyle{remark}
\newtheorem{remark}{Remark}[section]
\newtheorem{exmp}{Example}[section]

\makeatletter
\@namedef{Changes@AuthorColor}{red}
\colorlet{Changes@Color}{red}
\makeatother

\newcommand{\R}{\mathbb{R}}

\newcommand{\boldB}{\boldsymbol{B}}
\newcommand{\boldD}{\boldsymbol{D}}

\newcommand{\boldI}{\boldsymbol{I}}

\newcommand{\boldP}{\boldsymbol{P}}

\newcommand{\boldX}{\boldsymbol{X}}
\newcommand{\boldY}{\boldsymbol{Y}}

\newcommand{\boldw}{\boldsymbol{w}}
\newcommand{\boldx}{\boldsymbol{x}}
\newcommand{\boldy}{\boldsymbol{y}}

\newcommand{\indicator}{\mathbbm{1}}

\newcommand{\bbeta}{\boldsymbol{\beta}}
\newcommand{\bdelta}{\boldsymbol{\delta}}

\newcommand{\bepsilon}{\boldsymbol{\varepsilon}}
\newcommand{\bnu}{\boldsymbol{\nu}}

\newcommand{\btheta}{\boldsymbol{\theta}}
\newcommand{\bmu}{\boldsymbol{\mu}}
\newcommand{\bSigma}{\boldsymbol{\Sigma}}
\newcommand{\boldzero}{\mathbf{0}}
\newcommand{\boldone}{\mathbf{1}}

\newcommand{\prob}{\mathbb{P}}

\DeclareMathOperator*{\argmin}{arg\,min}
\DeclareMathOperator*{\argmax}{arg\,max}

\definecolor{MSU}           {RGB}{24, 69, 59}

\newcommand\hl[1]{%
  \bgroup
  \hskip0pt\color{red!80!black}%
  #1%
  \egroup
}

\title{Ultra high dimensional generalised additive model: Unified Theory and Methods}
\author{}
\date{\today}
\begin{document}
\begin{titlepage}
\maketitle 
\begin{center}
(Running title: High dimensional GAM)
\end{center}
 
        \vfill
\begin{center}
\begin{multicols}{2}

        Kaixu Yang\\
        Department of Statistics and Probability\\
        Michigan State University\\
        USA\\
        
\columnbreak
  
        Tapabrata Maiti\\
        Department of Statistics and Probability\\
        Michigan State University\\
        USA
\end{multicols}
\end{center}
\end{titlepage}

\begin{abstract}
Generalised additive model is a powerful statistical learning and predictive modeling tool that has been applied in a wide range of applications. The need of high-dimensional additive modeling is eminent in the context of dealing with high through-put data such as \replaced{genetic data analysis. }{}In this article, we studied a two step selection and estimation method for ultra high dimensional generalised additive models. The first step applies group lasso on the expanded bases of the functions. With high probability this selects all nonzero functions without having too much over selection. The second step uses adaptive group lasso with any initial estimators, including the group lasso estimator, that satisfies some regular conditions. The adaptive group lasso estimator is shown to be selection consistent with improved convergence rates. Tuning parameter selection is also discussed and shown to select the true model consistently under GIC procedure. The theoretical properties are supported by extensive numerical study.
\end{abstract}


\section{Introduction}
The main objective of this work is to establish theory driven high dimensional generalised additive modeling method with nonlinear links. The methodology includes convergence rate, variable selection consistency and tuning parameter selection consistency. Additive models play important roles in nonparametric statistical modeling and machine learning. Although this important statistical learning tool has been used in many important applications and there are free software available for implementing these models along with their variations, to our surprise, there is no literature that has studied the high-dimensional GAM with non-identity link systematically with theoretical foundation. Generalised additive modeling allows nonlinear relationship between a response variable and a set of predictor variables. This general set up includes the special case, namely, the generalised linear models, by letting each additive component be a linear function. In general, let $(y_i,\boldX_i),i=1,...,n$ be independent observations, where $y_i$'s are response variables whose corresponding $p\text{-dimensional}$ predictor vectors are $\boldX_i$'s. A generalised additive model \added{\citep{hastie1986}} is defined as
\begin{equation}
\label{GAM}
\mu_i=E(y_i|\boldX_i)=g^{-1}\left(\sum_{j=1}^{p_n}f_j(X_{ij})\right),
\end{equation}
where $g(\cdot)$ is a link function, $f_j$'s are unspecified smooth functions and $X_{ij}$ is the $j$th component of vector $\boldX_i$. One of the functions could be a constant, which is the intercept term, but this is not necessary. The number of additive components is written as $p_n$, since it sometimes (usually in high dimensional set up) increases as $n$ increases. A simple case that many people have studied is $p_n=p$, where the number of additive components is fixed and usually less than the sample size $n$. The choice of link function is as simple as in generalised linear models, where people prefer to choose link functions that make the distribution of the response variables belong to the popular exponential family. A widely used generalised additive model has the identity link function $g(\mu)=\mu$, which gives the classical additive model
\begin{equation}
\label{AM}
y_i=\sum_{j=1}^{p_n}f_j(X_{ij})+\epsilon_i,
\end{equation}
where $\epsilon_i$'s are i.i.d random variables with mean $0$ and finite variance $\sigma^2$.

On the other hand, high dimensional data analysis has become a part of many modern days scientific applications. Often the number of predictors $p_n$ is much larger than the number of observations $n$, which is usually written as $p_n\gg n$. One of the most interesting scale is $p_n$ increases exponentially as $n$ increases, i.e. $\log p_n=O(n^{\rho})$ for some constant $\rho>0$. \cite{fan2011nonconcave} called this as non-polynomial dimensionality or ultra high-dimensionality.

In this paper, we consider the generalised additive model in a high-dimensional set up. To avoid identification problems, the functions are assumed to be sparse, i.e. only a small proportion of the functions are non-zero and all others are exactly zero. A more generalised set up is that the number of nonzero functions, denoted $s_n$, also diverges as $n$ increases. This case is also considered in this paper.

Many others have worked on generalised additive models. Common approaches use basis expansion to deal with the nonparametric functions, and perform variable selection and estimation methods on the bases. \cite{meier2009high} considered a simpler case (\ref{AM}), with a new sparsity-smoothness penalty and proved it's oracle property. They also performed a simulation study under logit link with their new penalty, however, no theoretical support was provided. \cite{fan2011nonparametric} proposed the nonparametric independence screening (NIS) method in screening the model (\ref{AM}).  However, the selection consistency and the generalised link functions were not discussed. \cite{marra2011practical} discussed the practical variable selection in additive models, \replaced{but not in the high-dimensional set up}{}. \cite{liu2013oracally} considered a two-step oracally efficient approach in generalised additive models in the low dimensional set up, but no variable selection in the high dimensional set up was done. \cite{huang2010variable} focused on the variable selection of (\ref{AM}) with fixed number of nonzero functions and identity link function using a two step approach: \replaced{f}{}irst group lasso \citep{bakin1999adaptive,yuan2006model} on the bases to select the nonzero predictors and then use adaptive group lasso to estimate the bases coefficients. They then established the selection consistency and provided the rate of convergence of the estimation. \added{\cite{amato2016additive} reviewed several existing algorithms highlighting the connections between them, including the non-negative garrote, COSSO and adaptive shrinkage, and presented some computationally efficient algorithms for fitting the additive models.} \cite{nandy2017additive} extended the consistency and rate of convergence of \cite{huang2010variable} to spatial additive models. \added{\cite{fan2018nonparametric} studied the GAM with identity link under the endogeneity setting. It worth mentioning that alternative methods to penalization have also been studied, for example, \cite{tutz2006generalized} studied fitting GAM and perform variable seleciton implicitly through likelihood based boosting.}

However, though widely used, no systematic theory about selection and estimation consistency and rate of convergence has been established for generalised additive models with non-identity link functions in the high-dimensional set up.

In this paper, we establish the theory part for generalised additive models with non-identity link functions in high dimensional set up. We develop a two-step selection approach, where in the first step we use group lasso to perform a screening, which, under mild assumptions, is able to select all nonzero functions and not over-select too much. In the second step, the adaptive group lasso procedure is used and is proved to select the true predictors consistently.

Another important practical issue in variable selection and penalised optimization problems is tuning parameter selection. Various cross validation (CV) techniques have been used in practice for a long time. Information criteria such as Akaike information criterion (AIC), AICc, Bayesian information criterion (BIC), Mallow's $C_p$ and etc. have been used to select `the best' model as well. Many equivalences among the tuning parameter selection methods have been shown in the Gaussian linear regression case. However, the consistency of these selection methods were not established. Later some variations of the information criteria such as modified BIC \citep{zhang2007modified,wang2009shrinkage} extended BIC \citep{chen2008extended} and generalised information criterion (GIC) \citep{fan2013tuning} were proposed and shown to have good asymptotic properties in penalised linear models and penalised likelihoods. However, the results are not useful for grouped variables in additive models, for which basis expansion technique is usually used and thus brings grouped selection.

In this paper, we generalise the result of generalised information criterion (GIC) by \cite{fan2013tuning} to group-penalised likelihood problems and show that under some common conditions and with a good choice of the parameter in GIC, we are able to select the tuning parameter that corresponds to the true model.

In section \ref{Sec:Model}, the model is specified and basic approach is discussed. Notations and basic assumptions are also introduced in this section. Section \ref{Sec:Methodology} gives the main results of the two steps selection and estimation procedure. Section \ref{Sec:Tuning} develops the tuning parameter selection. Extensive simulation study and  real data example are presented in section \ref{Sec:Numerical} followed by a short discussion in section \ref{Sec:Discussion}. The proofs of all theorems are deferred to supplementary materials.

\section{Model}
\label{Sec:Model}
We consider the generalised additive model (\ref{GAM}) with the link function corresponding to an exponential family distribution of the response. For each of the $n$ independent observations, the density function is given as
\begin{equation}
\label{expfamilydensity}
f_{y_i}(y)=c(y)\exp\left[\frac{y\theta_i-b(\theta_i)}{\phi}\right],\ 1\leq i\leq n,\ \theta_i\in\mathbb{R}.
\end{equation}
Without loss of generality, we assume that the dispersion parameter $0<\phi<\infty$ is assumed to be a known constant. Specifically we assume  $\phi=1$. \added{We consider a fixed-design throughout this paper, i.e., the design matrix $\boldX$ is assumed to be fixed. However, we have shown in appendix \ref{appendix1} that the same theory works for a random design under simple assumptions on the distribution of $\boldX$.} The additive relationship assumes that the densities of $y_i$'s depend on $\boldX_i$'s through the additive structure $\theta_i=\sum_{j=1}^{p_n}f_j(X_{ij})$. This is the canonical link. If we use other link functions, for example, $A(\cdot)$, the theory also works as long as the functions $A(\cdot)$ satisfies the Lipschitz conditions for some order.
 Let $b^{(k)}(\cdot)$ be the $k$-th derivative of $b(\cdot)$, then by property of the exponential family, the expectation and variance matrix of $\boldy=(y_1,...,y_n)^T$, under mild assumptions of $b(\cdot)$, is given by $\bmu(\btheta)$ and $\phi\Sigma(\btheta)$, where
\begin{equation}
\bmu(\btheta)=(b^{(1)}(\theta_1),...,b^{(1)}(\theta_n))^T\ \ \text{and}\ \ \Sigma(\btheta)=\text{diag}\{b^{(2)}(\theta_1),...,b^{(2)}(\theta_n)\}.
\end{equation}

The log-likelihood (ignoring the term $c(y)$ which is not interesting to us in parameter estimation) can be written as
\begin{align}
l&=\sum_{i=1}^n\left[y_i\left(\sum_{j=1}^{p_n}f_j(X_{ij})\right)-b\left(\sum_{j=1}^{p_n}f_j(X_{ij})\right)\right].
\end{align}

\replaced{Assume that the additive components belong to the Sobolev space $W_2^d([a,b])$. According to \cite{schumaker1981spline}, see pages 268-270, there exists B-spline approximation
\begin{equation}\label{basis expansion}
f_{nj}(x)=\sum_{k=1}^{m_n}\beta_{jk}\phi_k(x),\text{\ \ \ \ \ }1\leq j\leq p.
\end{equation}
with $m_n=K_n+l$, where $K_n$ is the number of internal knots and $l\geq d$ is the degree of the splines. Generally, it is recommended that $d=2$ and $l=4$, i.e., cubic splines. }{}

Using the approximation above, \cite{huang2010variable} proved that $f_{nj}$ well approximates $f_j$ in the sense of rate of convergence that
\begin{equation}
\label{approxerr}
\|f_j-f_{nj}\|_2^2=\int_a^b(f_j(x)-f_{nj}(x))^2dx=O(m_n^{-2d}).
\end{equation}

Therefore, using the basis approximation, the log-likelihood (ignoring the term $c(y)$ which is not related to the parameters) can be written as
\begin{align}\label{loglikelihood basis}
l&=\sum_{i=1}^n\left[y_i\left(\sum_{j=1}^{p_n} \sum_{k=1}^{m_n}\beta^0_{jk}\Phi_k(x_{ij})\right)-b\left(\sum_{j=1}^{p_n} \sum_{k=1}^{m_n}\beta^0_{jk}\Phi_k(x_{ij})\right)\right]=\sum_{i=1}^n\left[y_i\left({{\bbeta^0}}^T\Phi_i\right)-b\left({{\bbeta^0}}^T\Phi_i\right)\right],
\end{align}
where ${\bbeta^0}$ and $\Phi_i$ are the vector basis coefficients and bases defined  below.

It's also worth noting that the number of bases $m_n$ increases as $n$ increases. This is necessary since \cite{schumaker1981spline} mentioned that one need to have sufficient partitions to well approximate $f_j$ by $f_{nj}$. If we fix $m_n$, i.e. let $m_n=m_0$, though in the later part we will show the approach to estimate the basis coefficients can have better rate of convergence, the approximation error between the additive components and the spline functions $\|f_j(x)-f_{nj}(x)\|_2=[\int_a^b(f_j(x)-f_{nj}(x))^2dx]^{1/2}=O(1)$ will increase and lead to inconsistent estimations. Therefore, $m_n$, or more precisely, $K_n$, need to increase with $n$.

Our selection and estimation approach will be based on the bases approximated log likelihood (\ref{loglikelihood basis}). Before starting the methodology, we list the notations and state the assumptions we need in this paper.
\subsection*{\textit{Notations}}
\label{Sec:Notations}

The design matrix is $\boldX_{(n\times {p_n})}=(\boldx_1, ..., \boldx_n)^T$. The basis matrix is $\Phi_{(n\times m_n{p_n})}=(\Phi_1, ..., \Phi_n)^T$, where
$\Phi_i=(\phi_1(x_{i1}), ..., \phi_{m_n}(x_{i1}), ..., \phi_1(x_{i{p_n}}), ..., \phi_{m_n}(x_{i{p_n}}))^T.$

The true basis parameters are ${\bbeta^0}=(\beta^0_{11}, ..., \beta^0_{1m_n}, ..., \beta^0_{{p_n}1}, ..., \beta^0_{{p_n}m_n})^T\in \mathbb{R}^{m_n{p_n}}$

We assume the functions $f_1, ..., f_{p_n}$ are sparse, then ${\bbeta^0}$ is block-wise sparse, i.e. the blocks ${\bbeta^0}_{1}=(\beta^0_{11}, ..., \beta^0_{1m_n})^T,..., {\bbeta^0}_{{p_n}}=(\beta^0_{{p_n}1}, ..., \beta^0_{{p_n}m_n})^T$ are sparse.

Let $\bmu_y$ be the expectation of $\boldy$ based on the true basis parameters and $\bepsilon=\boldy-\bmu_y$.

Define the relationship $a_n\preceq b_n$ as there exists a finite constant $c$ such that $a_n\leq cb_n$.

For any function $f$ define $\|f\|_2=[\int_a^bf^2(x)dx]^{1/2}$, whenever the integral exists.

For any two collections of indices $S,\tilde{S}\subseteq\{1,...,p_n\}$, the difference set is denoted $S-\tilde{S}$. The cardinality of $S$ is denoted card$(S)$.
For any $\bdelta\in\mathbb{R}^{m_np_n}$, define $\bdelta_{1},...,\bdelta_{p_n}$ as its sub-blocks, where $\bdelta_{i}\in\mathbb{R}^{m_n}$, and define the block-wise support
$$\text{supp}_B(\bdelta)=\{j\in\{1,...,p_n\};\bdelta_j\neq \boldzero\}.$$

Define the block-wise cardinality
$\text{card}_B(\bdelta)=\text{card}(\text{supp}_B(\bdelta)).$

For $S=\{s_1,...,s_q\}\subseteq \{1,...,p_n\}$, define sub-block vector
$\bdelta_{S}=(\bdelta_{s_1}^T,...,\bdelta_{s_q}^T)^T$.

The number of additive components is denoted $p_n$, which is possible to grow faster than the sample size $n$. Let $T=\text{supp}_B({\bbeta^0})$ and $T^c$ be the compliment set. Let $\text{card}(T)=s_n$, where $s_n$ is allowed to diverge slower than $n$.

For each $U\subseteq\{1,...,p_n\}$ with $\text{card}(U-T)\leq m$ for some $m$, define
$$\mathcal{B}(U)=\{\bdelta\in\mathbb{R}^{m_np_n};\text{supp}_B(\bdelta)\subseteq U\},$$
$$\mathcal{B}(m)=\{\mathcal{B}(U);\text{for\ any}\ U\subseteq \{1,...,p_n\};\text{Card}(U-T)\leq m\}.$$

Let $\replaced{q\ }{}$ be an integer such that $\replaced{q\ }{}>s_n$ and $\replaced{q\ }{}=o(n)$. Define
$$\mathcal{B}_1=\{\bbeta\in\mathcal{B}:\text{card}_B(\bbeta)\leq \replaced{q\ }{}\},$$
\added{where $\mathcal{B}$ is a sufficiently large, convex and compact set in $\R^d$.}

\subsection*{\textit{Assumptions}}
\textbf{Assumption 1} (On design matrix)
\begin{itemize}
\item[]
Using the normalised B-spline bases, the basis matrix $\Phi$ has each covariate vector $\Phi_j,j=1,...,p_n$ bounded, i.e., $\exists\ c_{\Phi}$ such that $\|\Phi_j\|_2\leq \sqrt{n}c_{\Phi},\forall j=1,...,m_n\times p_n$.
\end{itemize}
\textbf{Assumption 2} (Restricted Eigenvalues \textbf{RE})
\begin{itemize}
\item[]
For a given sequence $N_n$, there exist $\gamma_0$ and $\gamma_1$ such that
\begin{equation}
\gamma_0\gamma_2^{2s_n}m_n^{-1}\leq\frac{\bdelta^T\Phi^T\Phi\bdelta}{n\|\bdelta\|_2^2}\leq\gamma_1m_n^{-1},
\end{equation}
where $\gamma_2$ is a positive constant such that $0<\gamma_2<0.5$, for all
$\bdelta\in\mathcal{C}$, where $\bdelta^T=(\bdelta_1^T,...,\bdelta_{p_n}^T)$ and
\begin{equation}
\mathcal{C}=\left\{\bdelta\in\R^{p_nm_n}:\|\bdelta\|_2\neq 0,\ \|\bdelta\|_2\leq N_n\ \text{and}\ \text{card}_B(\bdelta)=o(s_n)\right\}.
\end{equation}
\end{itemize}
\textbf{Assumption 3} (On the exponential family distribution)
\begin{itemize}
\item[]
The function $b(\theta)$ is three times differentiable with $c_1\leq b''(\theta)\leq c_1^{-1}$ and $|b'''(\theta)|\leq c_1^{-1}$ in its domain for some constant $c_1>0$. For unbounded and non-Gaussian distributed $Y_i$, there exists a diverging sequence $M_n=o(\sqrt{n})$ such that
\begin{equation}
\label{assbddmean}
\sup_{\bbeta\in\mathcal{B}_1}\max_{1\leq i\leq n}\left|b'\left(\left|\Phi_i^T\bbeta\right|\right)\right|\leq M_n.
\end{equation}
Additionally the error term $\epsilon_i=y_i-\mu_{y_i}$'s follow the uniform sub-Gaussian distribution, i.e., there exist constants $c_2>0$ such that uniformly for all $i=1,..,n$, we have
\begin{equation}
\label{assbddtailprob}
P(|\epsilon_i|\geq t)\leq 2\exp(-c_2t^2)\ \text{for\ any\ }t>0.
\end{equation}
\end{itemize}
\textbf{Assumption 4} (On nonzero function coefficients)
\begin{itemize}
\item[]\replaced{There exist a sequence $c_{f,n}$ that may tend to zero as $n\rightarrow\infty$ such that for all $j\in T$, the true nonzero functions ${f_j}$ satisfy $$\min_{j\in T}\|{f_j}\|_2\geq c_{f, n}.$$}{}
\end{itemize}
We note that Assumption 1 is a standard assumption in high dimensional models, where the design matrix needs to be bounded from above. Assumption 2 is a well-known condition in high-dimension set up on the empirical Gram matrix \citep{bickel2009simultaneous}. It is different than the regular eigenvalue condition, since when $n<p$, the $p\times p$ Gram matrix has rank less than $p$, thus it must have zero eigenvalues. Therefore, it is not realistic to bound the eigenvalues away from zero for all $\bnu\in\R^{p_nm_n}$, but we need to restrict to some space $\mathcal{C}$. In our set up, $\mathcal{C}$ is the restricted sub-block eigenvalue condition on sub-blocks of the Gram matrix studied by \cite{belloni2013least}. Though the lower bound and upper bound are imposed on the fixed design matrix, we gave a derivation in supplementary materials that this condition holds when $\boldX$ is drawn from a continuously differentiable density function which is bounded away from 0 and infinity on the domain of $\boldX$. This result is similar to the results in \cite{huang2010variable}.

Assumption 3 is a standard assumption to generalised models. (\ref{assbddmean}) and (\ref{assbddtailprob}) together controls the tail behavior of the responses, and as mentioned by \cite{fan2013tuning}, ensure a general and broad applicability of the method. Analogous assumptions to (\ref{assbddmean}) can also be seen in \cite{fan2010sure} and \cite{Bhlmann:2011:SHD:2031491}. \added{Specifically, for example, we have $b(\btheta)=\log(1+\exp(\btheta))$. It's easy to verify that both its second and third derivatives have their absolute values all bounded from above by 1. For equation (\ref{assbddmean}), observe that the first derivative is the mean of Bernoulli distribution, and thus it is also bounded. The error term is also bounded by 1, therefore, taking $c_2=\log(2)$ will make equation (\ref{assbddtailprob}) satisfy all logistic regression cases. Moreover, bounded second moment in logistic regression ensure that there exist $\epsilon$ such that the probability $p_i$ of each observation satisfies $\epsilon<p<1-\epsilon$.}

Assumption 4 appears often in variable selection methodologies, because intuitively a nonzero function or covariate has to contribute enough to the response in order to be considered nonzero.

\begin{remark}
\label{REremark}
In assumption 2, $\added{\bdelta}=\bbeta-\bbeta_0$ is the difference vector between a $\bbeta$ and the true coefficients $\bbeta_0$, thus we can view $\mathcal{C}$ as a restricted neighborhood of $\bbeta_0$, i.e.,
$$\mathcal{N}^{RE}_{\bbeta_0}=\left\{\bbeta:\|\bbeta-\bbeta_0\|_2\leq N_n,\ m_n\times\text{card}_B(\bdelta)\leq n^*=o(n)\right\}.$$
If $\bbeta\in\mathcal{N}^{RE}_{\bbeta_0}$, then by assumption 2 we have
$$\frac{(\bbeta-\bbeta_0)^T\Phi^T\Phi(\bbeta-\bbeta_0)}{n\|\bbeta-\bbeta_0\|_2^2}\geq\gamma_0\gamma_2^{2s_n}m_n^{-1}.$$
This, together with the bounded variance assumption in assumption 3, ensures the restricted strong convexity of the target function, i.e., for a $\bbeta^*\in\mathcal{N}^{RE}_{\bbeta_0}$, we have
\begin{equation}
\label{restrictedstrongconvexity}
\frac{(\bbeta^*-\bbeta_0)^T\Phi^T\bSigma(\bbeta)\Phi(\bbeta^*-\bbeta_0)}{n\|\bbeta^*-\bbeta_0\|_2^2}\geq\gamma_0c_1\gamma_2^{2s_n}m_n^{-1},\ \forall\ \bbeta\in\mathcal{N}^{RE}_{\bbeta_0}.
\end{equation}
\end{remark}

\section{Methodology \& Theoretical Properties}
\label{Sec:Methodology}
We propose a two step procedure for selecting high dimensional additive models with generalised link that has improved convergence rates compared to single stage selection.
\subsection{\textit{First step: model screening}}
The objective of this step is to recover the true support $T$ of the additive components. Let $\hat{T}$ be a random support given by a model selection procedure and $|\hat{T}|$ be the number of variables selected. A good model selection procedure should satisfy the common screening consistency conditions
\begin{equation}
\label{firstrequirement}
T\subset \hat{T},\ |\hat{T}|=O(s_n),\ \text{w.p.\ converging\ to\ }1.
\end{equation}

There have been many variable selection penalization \citep{fan2004nonconcave,van2008high,fan2010sure,fan2011nonconcave} in generalised linear models and \citep{huang2010variable} in linear additive models where this condition holds. Specifically, \cite{fan2010sure} satisfies the requirements in (\ref{firstrequirement}) in generalised linear models and \cite{huang2010variable} also satisfies (\ref{firstrequirement}) in additive models with identity link function. In this paper, we show that under mild conditions, by maximizing the log-likelihood with group lasso-like penalization, we can select a model that satisfies (\ref{firstrequirement}). We also provide a rate of convergence of this first step selection.

Define the objective function to be
\begin{equation}
\label{grouplasso}
L(\bbeta;\lambda_{n1})=-\frac{1}{n}\sum_{i=1}^n\left[y_i\left(\bbeta^T\Phi_i\right)-b\left(\bbeta^T\Phi_i\right)\right]+\lambda_{n1}\sum_{j=1}^{p_n}\|\bbeta_{j}\|_2.
\end{equation}
Let $\hat{\bbeta}$ be the optimiser for (\ref{grouplasso}), i.e.
$$\hat{\bbeta}=\argmin_{\bbeta\in\mathbb{R}^{p_nm_n}}L(\bbeta;\lambda_{n1}).$$
Let $\hat{T}=\text{supp}_B(\hat{\bbeta})$.

The objective function is the negative log-likelihood plus the group lasso penalization term, and the parameters are estimated as the minimisers of the objective function. Here the negative log likelihood function is averaged among the n observations to ensure that it is under the same scale as the penalization function.

With this group lasso type penalised log-likelihood, the selected model has the following properties.

\begin{thm}
\label{screening1}
Consider the model $\hat{T}$ obtained by minimizing (\ref{grouplasso}). Under Assumptions 1-4, for some constant $C$ and any diverging sequence $\gamma_n>0$, choose the regularization parameter
$$\lambda_{n1}^b=C\sqrt{m_n}\sqrt{\frac{\gamma_n+\log(p_nm_n)}{n}}$$
for bounded response (i.e., $|y_i|<c$), and the regularization parameter
$$\lambda_{n1}^{ub}=\sqrt{m_n}\gamma_n\sqrt{\frac{\log(p_nm_n)}{n}}$$
for unbounded sub-Gaussian response, as the sample size increases,
\begin{enumerate}
\item[(i)]
With probability tending to 1,
$$|\hat{T}|=O(s_n)$$
\item[(ii)]
With probability tending to 1,
$$\sum_{j=1}^{p_n}\left\|\bbeta^0_j-\hat{\bbeta}_j\right\|_2^2=O_P\left(s_n\gamma_2^{-2s_n}\frac{m_n^2\log(p_nm_n)}{n}\right)+O({\lambda^b_{n1}}^2m_n^2s_n\gamma_2^{-2s_n})+O(s_n^2m_n^{1-2d}\gamma_2^{-2s_n})$$
for the bounded response and
$$\sum_{j=1}^{p_n}\left\|\bbeta^0_j-\hat{\bbeta}_j\right\|_2^2=O_P\left(s_n\gamma_2^{-2s_n}\gamma_n\frac{m_n^2\log(p_nm_n)}{n}\right)+O({\lambda^{ub}_{n1}}^2m_n^2s_n\gamma_2^{-2s_n})+O(s_n^2m_n^{1-2d}\gamma_2^{-2s_n})$$
for any diverging sequence $\gamma_n$ and unbounded sub-Gaussian response.
\item[(iii)]
If $s_n\gamma_2^{-2s_n}m_n^2\log(p_nm_n)/n\replaced{<<c_{f,n}}{}$ ($s_n\gamma_2^{-2s_n}m_n^2\gamma_n\log(p_nm_n)/n\replaced{<<c_{f,n}}{}$ in the unbounded case), $\gamma_2^{-2s_n}m_n^2\lambda_{n1}^2s_n/m_n\replaced{<<c_{f,n}}{}$ and $s_n^2m_n^{1-2d}\gamma_2^{-2s_n}\replaced{<<c_{f,n}}{}$, with probability tending to 1, all nonzero coefficients are selected.
\end{enumerate}
\end{thm}
The proof of this theorem is given in supplementary materials.

\begin{remark}
\label{numberofselected}
To avoid estimability issues, here the constants $C$ are selected to be large enough such that the number of parameters to be estimated, i.e., the number of selected nonzero functions $|\hat{T}|$ multiplied by the number of basis function $m_n$ should be less than or equal to $n$. Moreover, considering the multicollinearity in the design matrix, the constants are chosen such that $m_n\times|\hat{T}|=o(n)$.
\end{remark}

\begin{remark}
The additional term $\gamma_n$ in the convergence rate is due to unboundedness nature of the response variable rather than due to non-linear link function.
\end{remark}

\begin{remark}
For the special case, linear (Gaussian) additive model, our results coincide with \cite{huang2010variable}. The difference is \replaced{that we study a fixed design with assumptions on the eigenvalues of the design matrix and they studied a random design with assumption on the distribution of the design matrix. We have put further assumption on the eigenvalue due to the divergence of $s_n$, the number of nonzero variables. In the special case that $s_n$ is fixed, our assumptions coincides with the assumptions in \mbox{\cite{huang2010variable}}.}{} Another difference is that we include a diverging term $\gamma_n$ that establishes the rate of convergence with probability converging to one.
\end{remark}

There are three terms in the convergence rate: the first term comes from the regression itself, the second term comes from shrinkage, and the third term comes from the spline approximation error.

\begin{remark}
\label{remark:screeningfunction}
\added{Let $\hat{f}_{nj}(x)=\sum_{k=1}^{m_n}\hat{\beta}_{jk}\phi_k(x)$. We can also state the results of the first selection step in terms of functions, which is a direct consequence of theorem \ref{screening1}. First, we have (i) $|\hat{T}|=O(s_n)$ with probability tending to 1, and (ii) if $s_nm_n\gamma_2^{-2s_n}\log(p_nm_n)/n<<c_{f,n}$ ($s_nm_n\gamma_2^{-2s_n}\gamma_n\log(p_nm_n)/n<<c_{f,n}$ in the unbounded case), $\lambda_{n1}^2s_nm_n\gamma_2^{-2s_n}<<c_{f,n}$ and $s_n^2m_n^{-2d}\gamma_2^{-2s_n}<<c_{f,n}$, with probability tending to 1, all nonzero coefficients are selected.

Moreover, by the properties of spline in \mbox{\cite{de2001practical}}, see for example \mbox{\cite{stone1986dimensionality}} and \mbox{\cite{huang2010variable}}, there exist positive constants $c_1$ and $c_2$ such that
\begin{equation}
\label{eq:splinetofunctionrate}
c_1m_n^{-1}\|\hat{\bbeta}_{nj}-\bbeta_{nj}\|_2^2\leq\|\hat{f}_{nj}-f_{nj}\|_2^2\leq c_2m_n^{-1}\|\hat{\bbeta}_{nj}-\bbeta_{nj}\|_2^2,
\end{equation}
we have
$$\sum_{j=1}^{p_n}\left\|f_j-\hat{f}_{nj}\right\|_2^2=O_P\left(s_n\gamma_2^{-2s_n}\frac{m_n\log(p_nm_n)}{n}\right)+O({\lambda^b_{n1}}^2m_ns_n\gamma_2^{-2s_n})+O(s_n^2m_n^{-2d}\gamma_2^{-2s_n})$$
for the bounded response case and
$$\sum_{j=1}^{p_n}\left\|f_j-\hat{f}_{nj}\right\|_2^2=O_P\left(s_n\gamma_2^{-2s_n}\gamma_n\frac{m_n\log(p_nm_n)}{n}\right)+O({\lambda^{ub}_{n1}}^2m_ns_n\gamma_2^{-2s_n})+O(s_n^2m_n^{-2d}\gamma_2^{-2s_n})$$
for the unbounded case, for any diverging sequence $\gamma_n$.}
\end{remark}

\begin{remark}
\label{satisfactory}
The theorem and its remark together tell us under Assumptions 1-4, by choosing proper $\gamma_n$, the functions selected by minimizing the first target function satisfy
$$T\subset\hat{T}\ \ \text{and}\ \ |\hat{T}|=O(s_n)$$
with probability converging to 1, i.e. we obtained screening consistency.
\end{remark}

\subsection{\textit{Second step: Post selection}}
After we have a ``good" initial estimator, we use the adaptive group lasso to recover the true model \citep{huang2010variable} and we are able to achieve selection consistency in probability under some mild assumptions. The adaptive group lasso idea is similar to adaptive lasso \citep{zou2006adaptive} which enjoys better theoretical properties than simple lasso. \cite{chatterjee2013rates} and \cite{das2017perturbation} studied rate of convergence and other asymptotic properties of the adaptive lasso estimator.
Define the objective function to be
\begin{equation}
\label{adagrouplasso}
L_a({\bbeta};\lambda_{n2})=-\frac{1}{n}\sum_{i=1}^n\left[y_i\left({\bbeta}^T\Phi_i\right)-b\left(\bbeta^T\Phi_i\right)\right]+\lambda_{n2}\sum_{j=1}^{p_n}w_{nj}\|\bbeta_{j}\|_2,
\end{equation}
where the weights  depend on the screening stage group lasso estimator
\begin{equation}
w_{nj}=+
\begin{cases}
\|\hat{\bbeta}_{j}\|_2^{-1},\ \ \ \ \text{if}\ \|\hat{\bbeta}_{j}\|_2>0\\
\infty,\ \ \ \ \ \ \ \ \ \ \text{if}\ \|\hat{\bbeta}_{j}\|_2=0
\end{cases}.
\end{equation}
Let $\hat{\bbeta}_{AGL}$ be the optimiser for (\ref{adagrouplasso}), i.e.
$$\hat{\bbeta}_{AGL}=\argmin_{\bbeta\in\mathbb{R}^{m_np}}L_a(\bbeta;\lambda_{n2}).$$
For the choice of weights, the first stage estimators need not to be necessarily  the solution of group lasso, but could be more general estimators that satisfy following assumptions.

\noindent
\textbf{Assumption 5}
\begin{itemize}
\item[]
The initial estimator $\hat{\bbeta}$ is $r_n$ consistent at zero, i.e.,
\begin{equation}
\label{rnconsistent}
r_n\max_{j\in T^c}\|\hat{\bbeta}_j-\bbeta^0_j\|_2=O_P(1),
\end{equation}
and there exists a constant $c_3$ such that
\begin{equation}
\label{initialtrue}
\prob\left(\min_{j\in T}\|\hat{\bbeta}\|_2\geq c_3b_{n1}\right)\rightarrow 1,
\end{equation}
where $b_{n1}=\min_{j\in T}\|\bbeta_j^0\|_2$.
\end{itemize}
\textbf{Assumption 6}
\begin{itemize}
\item[]
Let $s_n^*=p_n-s_n$ be the number of zero components. The tuning parameter $\lambda_{n2}$ satisfies
\begin{equation}
\label{lambdan2}
\begin{split}
&\frac{\sqrt{\log(s_n^*m_n)}}{n^{1/2}\lambda_{n2}r_n}+\frac{s_n}{\lambda_{n2}r_nm_n^{d+1/2}}+\frac{\lambda_{n2}r_n}{\gamma_n\sqrt{s_n/n}}=o(1)
\end{split}
\end{equation}
for any diverging sequence $\gamma_n$.
\end{itemize}
Assumption 5 gives the restrictions on the initial estimator. We don't require our initial estimator to be the group lasso estimator. Any initial estimator satisfying assumption 5 will be able to make the adaptive group lasso estimator consistently selects and estimates the true nonzero components. However, the rate of convergence of the adaptive group lasso estimator  depends on the rate of convergence of the initial estimator, which is assumed to be $r_n$ in assumption 5. Moreover, the initial estimator mustn't have a 0 estimation for the nonzero components, otherwise it will mislead the results in the proceeding step. Assumption 6 put restrictions on the tuning parameter $\lambda_{n2}$ in the adaptive group lasso step. The first two terms gives the upper bound for $\lambda_{n2}$ and the third term gives the lower bound. Only with ``appropriate" choice of $\lambda_{n2}$  we can have the selection consistency and estimation consistency.

\added{It worth noting that if we take the group lasso estimator as our initial estimator, assumptions 5 and 6 are automatically satisfied. Specifically, a trivial choice of $r_n$ would be
$$r_n=O_P^{-1}\left(\sqrt{s_n}\gamma_2^{-s_n}\frac{m_n\sqrt{\log(p_nm_n)}}{\sqrt{n}}\right)+O^{-1}({\lambda^b_{n1}}m_n\sqrt{s_n}\gamma_2^{-s_n})+O^{-1}(s_nm_n^{0.5-d}\gamma_2^{-s_n})$$
for the bounded response and
$$r_n=O_P^{-1}\left(\sqrt{s_n}\gamma_2^{-s_n}\sqrt{\gamma_n}\frac{m_n\sqrt{\log(p_nm_n)}}{\sqrt{n}}\right)+O^{-1}({\lambda^{ub}_{n1}}m_n\sqrt{s_n}\gamma_2^{-s_n})+O^{-1}(s_nm_n^{0.5-d}\gamma_2^{-s_n})$$
for the unbounded case and any diverging sequence $\gamma_n$, since we observe that for $j\in T^c$, $\hat{\bbeta}_j$ is either estimated as zero, or has a rate of convergence to $\bbeta_j$ bounded by the rate of convergence in theorem (\ref{screening1}). For equation (\ref{initialtrue}), observe that the rate of convergence of the group lasso estimator is higher order infinitesimal of the minimal signal strength of nonzero coefficients, thus taking $c_3=0.5$ is sufficient. In assumption 6, with our trivial choice of $r_n$, we are able to find a range of tuning parameters that satisfy equation (\ref{lambdan2}). Therefore, it's reasonable to take the group lasso estimator as an initial estimator for the adaptive group lasso.
}

Let the notation $\hat{\bbeta}_n\stackrel{0}{=}\bbeta^0$ denote that the sign of each $\hat{\bbeta}_j$ and $\bbeta^0_j$ are either both zero or both nonzero. Then we have the following asymptotic properties for the adaptive group lasso estimator.
\begin{thm}
\label{selection1}
Assume assumptions 1-6 hold, consider the estimator $\hat{\bbeta}_{AGL}$ by minimizing (\ref{adagrouplasso}), we have
\begin{enumerate}
\item[(i)]
\added{If $f_{c,n} >> \sqrt{s_n/n}$}, the adaptive group lasso consistently selects the true active predictors with probability converging to 1, i.e.,
\begin{equation}
\prob\left(\hat{\bbeta}_{AGL}\overset{0}{=}{\bbeta^0}\right)\rightarrow 1.
\end{equation}
\item[(ii)]
The rate of convergence of the adaptive group lasso estimator is given by
\begin{equation*}
\sum_{j\in T}\|\hat{\bbeta}_{AGLj}-\bbeta^0_{j}\|_2^2=O_p\left(s_n\gamma_2^{-2s_n}m_n^2\frac{\log(s_nm_n)}{n}\right)+O(s_n^2\gamma_2^{-2s_n}m_n^{1-2d})+O(\lambda_{n2}^2m_n^2s_n\gamma_2^{-2s_n})
\end{equation*}
for the bounded response case and
\begin{equation*}
\sum_{j\in T}\|\hat{\bbeta}_{AGLj}-\bbeta^0_{j}\|_2^2=O_p\left(\gamma_ns_n\gamma_2^{-2s_n}m_n^2\frac{\log(s_nm_n)}{n}\right)+O(s_n^2\gamma_2^{-2s_n}m_n^{1-2d})+O(\lambda_{n2}^2m_n^2s_n\gamma_2^{-2s_n})
\end{equation*}
for the unbounded response case, where $\gamma_n$ is any diverging sequence.
\end{enumerate}
\end{thm}
The proof of this theorem is given in supplementary materials. \added{It's interesting to compare the adaptive group lasso results with \cite{wang2019adaptive}, who studied the asymptotic properties of the adaptive group lasso for generalized linear models. It worth noting that we considered a more general case by allowing the group size to diverge with $n$, and the eigenvalue to be bounded by sequences that depending on $n$ on a broader domain. In the special case that corresponds to their assumptions, our results (Theorem \ref{selection1}) coincides with their results.}

Similar to the group lasso estimator, we also derive the results for the non-parametric function estimations, stated in the following remark.

\begin{remark}
\added{Let $\hat{f}_{AGLj}(x)=\Phi_{j}(x)\hat{\bbeta}_{AGLj}$. We can also state the results of the first selection step in terms of functions, which is a direct consequence of theorem \ref{screening1}. First, we have the true nonzero subset is recovered with probability tending to 1. Moreover, by the same properties of spline as in Remark \ref{remark:screeningfunction}, we have
$$\sum_{j\in T}\|\hat{f}_{AGLj}-f_j\|_2^2=O_p\left(s_n\gamma_2^{-2s_n}m_n\frac{\log(s_nm_n)}{n}\right)+O(s_n^2\gamma_2^{-2s_n}m_n^{-2d})+O(\lambda_{n2}^2m_ns_n\gamma_2^{-2s_n})$$
for the bounded response case and
$$\sum_{j\in T}\|\hat{f}_{AGLj}-f_j\|_2^2=O_p\left(\gamma_ns_n\gamma_2^{-2s_n}m_n\frac{\log(s_nm_n)}{n}\right)+O(s_n^2\gamma_2^{-2s_n}m_n^{-2d})+O(\lambda_{n2}^2m_ns_n\gamma_2^{-2s_n})$$
for the unbounded case, for any diverging sequence $\gamma_n$.}
\end{remark}

\added{The convergence rate for the group lasso estimator is
$$\sum_{j=1}^{p_n}\left\|f_j-\hat{f}_{nj}\right\|_2^2=O_P\left(s_n\gamma_2^{-2s_n}\frac{m_n\log(p_nm_n)}{n}\right)+O({\lambda^b_{n1}}^2m_ns_n\gamma_2^{-2s_n})+O(s_n^2m_n^{-2d}\gamma_2^{-2s_n}),$$
while for the adaptive group lasso estimator is
$$\sum_{j\in T}\|\hat{f}_{AGLj}-f_j\|_2^2=O_p\left(s_n\gamma_2^{-2s_n}m_n\frac{\log(s_nm_n)}{n}\right)+O(\lambda_{n2}^2m_ns_n\gamma_2^{-2s_n})+O(s_n^2\gamma_2^{-2s_n}m_n^{-2d})$$
The regression term differs by the size of candidate set. The price we pay by not knowing the true set is $\log(pm_n)$ in the group lasso step, and becomes $\log(s_nm_n)$ in the adaptive group lasso step, since the initial estimator have recovered a super set of the true set with cardinality $O(s_n)$. The penalty term's difference appears on the tuning parameter, where $\lambda_{n2}$ is of a smaller order than $\lambda_{n1}$ with a multiplier of $r_n^{-1}$. According to our choice of $\lambda_{n2}$, it has a trivial upper bound which is of order $O(\lambda_{n1}^2)$. Therefore, the tuning parameter part in the penalty convergence rate term becomes quadratic. The approximation error term is not affected by the adaptive group lasso step.
}

The adaptive group lasso is important in two reasons: first, with probability tending to 1, this is enable to select the true nonzero components accurately, which is not always the case in group lasso; second, the rate of convergence of the adaptive group lasso estimator is faster than the rate of convergence of the group lasso estimator. The difference in the leading terms are in the order of $r_n^{-1}$.  This makes the adaptive group lasso estimator to achieve a better error with the same sample size, or the same error with a smaller sample size.

The theorem and remark in this section ensure that under mild assumptions, we are able to recover the true model with probability tending to 1 and achieve a rate of convergence better than the initial estimator. Particularly, if the restrictions of $n, p_n, m_n$ and $s_n$ in the previous section satisfy, the group lasso estimator is actually a good initial estimator. Therefore, this two step procedure actually is a complete procedure that gives us a way to do this model selection and estimation on any high-dimensional generalised additive model. However, the procedure is not practically complete without proper selection of the tuning parameter $\lambda$. Therefor, we propose a theoretically validated tuning parameter selection in the next section.

\section{Tuning parameter selection}
\label{Sec:Tuning}
One important issue in penalised methods is choosing a proper tuning parameter. It is known that the selection results are sensitive to the choice of tuning parameters. The theoretical results only provide the order of the tuning parameter, which is not very useful in practice. The reason is that the order of a sequence describes the limit properties when $n$ goes to infinity. In reality, our $n$ is a fixed number, so we must have a practical instruction on selecting the tuning parameter.

Despite its importance, there isn't much development for tuning parameter selection in the high dimensional literature. The conventional tuning parameter selection criteria tend to select too many predictors, thus is hard to reach selection consistency. Another reason, especially in group lasso problems, is that the solution path of group lasso is piecewise nonlinear, which makes the testing procedure even harder. Here, we propose the generalised information criterion (GIC) \citep{zhang2010regularization,fan2013tuning} that supports consistent model selection.

Let $\hat{\bbeta}^{\lambda}$ be the adaptive group lasso solution with tuning parameter $\lambda$. The generalised information criterion is defined as
\begin{equation}
GIC(\lambda)=\frac{1}{n}\{D(\hat{\mu}_{\lambda};\boldY)+a_n|\hat{T}\added{_{\lambda}}|\},
\end{equation}
where $D(\hat{\mu}_{\lambda};\boldY)=2\{l(\boldY;\boldY)-l(\hat{\mu}_{\lambda};\boldY)\}$\replaced{. Here the $l(\bmu;\boldY)$ is the log-likelihood function in equation (\ref{expfamilydensity}) expressed as a function of the expectation $\bmu$ and $\boldY$. $l(\boldY; \boldY)$ represents the saturated model with $\bmu=\boldY$, and $\hat{\mu}_{\lambda}=b'(\sum_{i=1}^{p_n}\hat{f}^{\lambda}_j(x_{ij}))=b'(\phi\hat{\bbeta}^{\lambda})$ is our estimated expectation when the tuning parameter is $\lambda$.
}{} The hyperparameter $a_n$ is to penalise the size of the model. Using GIC, under proper choice of $a_n$, we are able to select all active predictors consistently.

 The importance of  the following consistency theorem is that the result in the previous section guarantees that with probability converging to 1, there exists a $\lambda_{n0}$ that will be able to identify the true model. Therefore, a good choice of $a_n$ will be able to identify the true model with probability converging to 1. \added{For a support $A\subset\{1,...,p\}$ such that $|A|\leq q$, where $q\geq s_n$ and $q=o(n)$, }let
\begin{equation}
I(\bbeta(A))=E\left[\log(f^*/g_{A})\right]=\sum_{i=1}^n\left[b'(\Phi_i\bbeta^0)\Phi_i^T(\bbeta^0-\bbeta(A))-b(\Phi_i^T\bbeta^0)+b(\Phi_i^T\bbeta(A))\right]
\end{equation}
be the Kullback-Leibler (KL) divergence between the true model and the selected model, where $f^*$ is the density of the true model, and $g_{A}$ is the density of the model with population parameter $\bbeta(A)$. Let $\bbeta^*(A)$ be the model with the smallest KL divergence over all models with support $A$, and let
$$\delta_n=\inf_{\substack{A\not\supset T\\\replaced{|A|\leq q}{}}}\frac{1}{n}I(\bbeta^*(A)).$$
\added{Here we note that if $T\subset A$, the minimizer is automatically $\bbeta^0$ and thus the KL-divergence is zero. For an underfitted models $T\not\subset A$, $\delta_n$ describes how easily one can distinguish the models from the true model by measuring the minimum distance from the true model to the ``best estimated models''. Later in the theorems we will need to assume lower bounds on $\delta_n$ so that we will be able to reach our consistency results. The following theorem proves that GIC works under mild conditions.}
\begin{thm}
\label{GICconsistency}
Under assumptions 1-6, suppose that $\delta_nq^{-1}R_n^{-1}\rightarrow\infty$, $n\delta_ns_n^{-1}a_n^{-1}\rightarrow\infty$ and $a_n\psi^{-1}\rightarrow\infty$, where $R_n$ and $\psi_n$ are defined in lemma \ref{gicunderfitted} and lemma \ref{gicoverfitted}, we have, as $n\rightarrow \infty$,
\begin{equation}
\prob\{\inf_{\lambda\in\Omega_-\cup\Omega_+}GIC_{a_n}(\lambda)>GIC_{a_n}(\lambda_{n0})\}\rightarrow 1,
\end{equation}
where
$$\Omega_-=\{\lambda\in[\lambda_{min},\lambda_{max}]:T_{\lambda}\not\supset T\},$$
$$\Omega_+=\{\lambda\in[\lambda_{min},\lambda_{max}]:T_{\lambda}\supset T\ \text{and}\ T_{\lambda}\neq T\},$$
where
$T_{\lambda}$ is the set of predictors selected by tuning parameter $\lambda$. \replaced{$\lambda_{min}$ can be chosen as the smallest $\lambda$ such that the selected model has size $q$ that satisfies the theorem assumption, and $\lambda_{max}$ simply corresponds to a model with no variables.}{}
\end{thm}
The proof of this theorem is given in supplementary materials. In practice, a choice of $a_n$ is proposed to be $m_n\log(\log(n))\log(p_n)$. We have
\begin{cor}
Under assumptions 1-6, with choice of $a_n=m_n\log(\log(n))\log(p_n)$, we have
$$\prob\{\inf_{\lambda\in\Omega_-\cup\Omega_+}GIC_{a_n}(\lambda)>GIC_{a_n}(\lambda_{n0})\}\rightarrow 1.$$
\end{cor}

In our two step procedure, there are two tuning parameters to be selected: $\lambda_{n1}$ in the group lasso step and $\lambda_{n2}$ in the adaptive group lasso step. The choice of $\lambda_{n2}$ is of more importance, since $\lambda_{n1}$ only serve as the parameter in screening. As long as we have a screening step that satisfies (\ref{firstrequirement}), we are ready for the adaptive group lasso step. To be simple, we propose to use GIC for selecting both $\lambda_{n1}$ and $\lambda_{n2}$. As a result of the previous theorem, we are able to reach selection consistency.

\section{Numerical Properties}
\label{Sec:Numerical}
In this section we conduct various empirical exercises to illustrate our theoretically guided method in practice. \added{To optimize the group lasso problems, we apply the algorithm named groupwise-majorization-descent (GMD) by \cite{yang2015fast}, which approximates the convex log-likelihood part with second order Taylor expansion and solve it with a quadratic function's closed form solution, wrapped in a block coordinate descent algorithm. We made the algorithm in GAM available as a python class, which is accessible at https://github.com/KaixuYang/PenalizedGAM.

As smoothness is a  concern in practical GAM computations, we bring the P-spline \citep{eilers1996flexible} penalty into the model while implementing the model numerically. The P-spline penalty controls the difference between coefficients of consecutive basis functions, and thus yields smoother spline functions.

Specifically, let $l(\bbeta; \boldX, \boldy)$ be the loss function in section \ref{Sec:Methodology}, either the group lasso loss function or the adaptive group lasso loss function. The loss function with smoothness penalty is defined as
\begin{equation}
\label{equ:smoothloss}
l_s(\bbeta;\boldX,\boldy) = l(\bbeta; \boldX, \boldy) + \lambda_s\sum_{j=1}^p\bbeta_j^T\boldD\bbeta_j,
\end{equation}
where
$$\boldD=
\begin{bmatrix}
1&-1&0&.\\
-1&2&-1&.\\
0&-1&2&.\\
.&.&.&.
\end{bmatrix}
$$
A slightly modified soft-thresholding function is used to handle the combination of group lasso penalty and the smoothness penalty.
}

\subsection{\textit{Simulated Examples}}
Here we undertake extensive simulation study to see the performance of our proposed two step selection and estimation approach. We investigate the performance of both uncorrelated and correlated covariates and we consider different sample sizes and varying number of predictors in each case.

\added{In this section, we  consider three different types of generalized models: the logistic regression (Bernoulli distribution), the Poisson regression (Poisson distribution) and the Gamma regression (Gamma distribution).} \added{Through the whole subsection, we choose $l=4$ which implies a cubic B-spline. We choose $m_n=9$ for most cases unless stated otherwise. The choice of $l$ and $m_n$ implies that there are $m_n-l=5$ inner knots, which are evenly placed over the empirical percentiles of the training data.} \added{In this subsection, we  compare the performance of the two-step approach with the Lasso \citep{tibshirani1996regression}, the GAMBoost \mbox{\citep{tutz2006generalized}} and the GAMSEL \citep{chouldechova2015generalized}. We implement our two-step approach with our own package mentioned above. The Lasso is implemented with the scikit-learn package in python. The GAMBoost and GAMSEL methods are implemented using their packages in R.} \added{In the group lasso step, we choose the tuning parameter corresponding to $n_g$ variables, where $n_g$ is the largest number such that $n_g\times m_n <= n$. This choice prevents estimation issues when we have too many parameters. The GIC procedure is applied in the adaptive group lasso step to select tuning parameters. In the GIC procedure, the tuning parameter selection criterion is defined as
\begin{equation}
\label{simgic}
GIC(\lambda)=\frac{1}{n}\{D(\hat{\mu}_{\lambda};\boldY)+a_n|\hat{T}|\}.
\end{equation}
From our results in the previous section, we choose $a_n=(\log\log n)(\log p)m_n$.}
\subsubsection{\added{Logistic Regression}}
\added{First, }we consider the logistic regression
\begin{equation}
\label{simlog}
y_i\sim Bernoulli(\theta_i),\ i=1,...,n,
\end{equation}
where $\theta_i=logit^{-1}[\alpha+\sum_{j=1}^pf_j(x_{ij})]$ and $x_{ij}$ is the $(i,j)-th$ element of the design matrix $X$.  
\begin{exmp}
\label{uncorex}
We first consider the logistic additive model on an independent design matrix case, where each predictor in X is independent of other predictors. Each element of the design matrix is generated from a $Unif(-1,1)$ distribution. \added{We consider 3 different cases with all $n$, $p$ and $s$ increasing, which coincides with our theory in section \ref{Sec:Methodology}. Specifically, the three cases are: $n=100$, $p=200$ and $s=3$; $n=200$, $p=500$ and $s=4$; $n=300$, $p=3000$ and $s=5$. A testing sample of size $1000$ is generated independently to measure the performance. For all three cases, we have nonzero functions $f_1(x)=5\sin(3x)$, $f_2(x)=-4x^4+9.33x^3+5x^2-8.33x$ and $f_3(x)=x(1-x^2)\exp(3x)-4$. These three general terms include a periodic term, a polynomial term and an exponential term. The last two cases have one more function of $f_4(x)=4x$, a linear term. Finally, the last case has an addition $f_5(x)=4\sin(-5\log(\sqrt{x+3})$, a complicated composite function. Without loss of generality, the first $s$ functions are set to be nonzero.} The constants in the functions are to ensure similar signal strength and smoothness. The other functions $f_{s+1}(x)=...=f_{p}(x)=0$.

\added{Our results focus on NV, the average number of variables being selected; TPR, the true positive rate (what percent of the truly nonzero variables are selected); FPR, the false positive rate (where percent of the zero variables are selected); and PE, the prediction error. In the logistic regression problem, our metric to measure the prediction error will be the misclassification rate, which is also the measurement in \cite{chouldechova2015generalized}. The simulation results are averaged over 100 repetitions.}

The simulation results are summarised in table \ref{uncorextable} on page \pageref{uncorextable}. \added{Compared with the classical method Lasso and the existing GAM methods GAMSEL and GAMBoost, the two-step approach performs the best in terms of both variable selection and estimation in the high-dimensional set up. The two-step approach performs significantly better in prediction errors. In variable selection, the two-step approach selects the closest number of variables to the ground truth, while keeping the TPR high and FPR low. The existing GAM algorithms have similar TPR but includes too many false positives. The existing GAM algorithms were not intended for very high-dimensional data, and thus fails to handle the variable selection and prediction at the same time. As mentioned in \cite{fan2001variable}, the tuning parameter in the Lasso for consistent variable selection is not the same as the tuning parameter for best prediction. We can see this may also be true for the group lasso case, since the estimated nonzero coefficients in the group lasso step are  over-penalized. This also proves that an adaptive group lasso step is important, in terms of both variable selection and prediction.}

\end{exmp}

\begin{landscape}
\begin{table}
\centering
\caption{Simulation results for the two-step approach compared with the Lasso, GAMSEL and GAMBoost in the three cases of Example \ref{uncorex}. NV, average number of the variables being selected; TPR, the true positive rate; FPR, the false positive rate; and PE, prediction error (here is the misclassification rate). Results are averaged over 100 repetitions. Enclosed in parentheses are the corresponding standard errors.}
\begin{tabular}{c|cccc|cccc|cccc}
\hline
&\multicolumn{4}{c}{\makecell{n=100\\p=200\\s=3}}&\multicolumn{4}{c}{\makecell{n=200\\p=500\\s=4}}&\multicolumn{4}{c}{\makecell{n=300\\p=3000\\s=5}}\\
\hhline{~------------}
&NV&TPR&FPR&PE&NV&TPR&FPR&PE&NV&TPR&FPR&PE\\
\hline
Two-step&\makecell{3.56\\(1.19)}&\makecell{.920\\(.146)}&\makecell{.004\\(.005)}&\makecell{.148\\(.027)}&\makecell{4.82\\(1.02)}&\makecell{.989\\(.057)}&\makecell{.002\\(.002)}&\makecell{.128\\(.018)}&\makecell{4.92\\(0.535)}&\makecell{.968\\(.086)}&\makecell{.000\\(.000)}&\makecell{.122\\(.018)}\\
Lasso&\makecell{30.0\\(17.9)}&\makecell{.920\\(.144)}&\makecell{.138\\(.090)}&\makecell{.249\\(.041)}&\makecell{64.7\\(19.2)}&\makecell{.978\\(.452)}&\makecell{.122\\(.039)}&\makecell{.229\\(.024)}&\makecell{85.2\\(68.3)}&\makecell{.816\\(.243)}&\makecell{.027\\(.022)}&\makecell{.211\\(.024)}\\
GAMSEL&\makecell{10.1\\(11.1)}&\makecell{.820\\(.209)}&\makecell{.039\\(.055)}&\makecell{.241\\(.035)}&\makecell{14.0\\(12.6)}&\makecell{.943\\(.112)}&\makecell{.021\\(.025)}&\makecell{.214\\(.023)}&\makecell{33.9\\(27.9)}&\makecell{.986\\(.065)}&\makecell{.010\\(.009)}&\makecell{.208\\(.016)}\\
GAMBoost&\makecell{44.7\\(4.84)}&\makecell{.738\\(.055)}&\makecell{.213\\(.025)}&\makecell{.231\\(.027)}&\makecell{85.4\\(6.88)}&\makecell{1.00\\(.000)}&\makecell{.164\\(.014)}&\makecell{.196\\(.018)}&\makecell{138\\(9.64)}&\makecell{.996\\(.028)}&\makecell{.044\\(.003)}&\makecell{.186\\(.015)}\\
\hline
\end{tabular}
\label{uncorextable}
\end{table}
\end{landscape}

In practice, the predictors are sometimes correlated to each other. It's interesting to see how well the procedure performs in correlated predictor cases. Therefore, we also perform the same comparison on correlated predictors.
\begin{exmp}
\label{corex}
In this example, we study the case where the design matrix contains correlated predictors. We generate the data in the following way. First we generate each element of $X_{n\times p}$ independently from $Unif(-1,1)$. Then we generate $u$ from $Unif(-1,1)$, \added{independently from $X_{n\times p}$}. Then all columns of $X$ are transformed using $X_j=(X_j+tu)/\sqrt{1+t^2}$. This procedure controls the correlation among predictors through $t$ such that $corr(x_{ik},x_{ij})=t^2/(1+t^2)$. \added{Here the simulation is run on $n=100$, $p=200$ and $s=3$. All other set-ups are kept  same as example \ref{uncorex}.} In our example, we choose $t=\sqrt{3/7}$, where the correlation is 0.3 and $t=\sqrt{7/3}$, where the correlation is 0.7. 

The results are summarised in table \ref{cor3table} on page \pageref{cor3table}. \added{In the correlated cases, all four methods are influenced, more or less. In terms of variable selection, the two-step approach still has the closest number of selected variables. The methods behave differently in terms of TPR and FPR. GAMBoost tends to have greater numbers in both TPR and FPR, while GAMSEL tends to have both lower numbers. The two-step approach balances between those two methods, while maintaining the smallest FPR among all methods. In terms of the prediction error, the two-step approach significantly beats the other methods. The results show  good performance of the two-step approach, and again emphasize that the adaptive group lasso step is necessary for better selection and estimation.}

\begin{table}[h]
\centering
\caption{Simulation results for the two-step approach compared with the Lasso, GAMSEL and GAMBoost in Example \ref{corex} with correlation 0.3 and 0.7 for $n=100$, $p=200$ and $s=3$. NV, average number of the variables being selected; TPR, the true positive rate; FPR, the false positive rate; and PE, prediction error (here is the misclassification rate). Results are averaged over 100 repetitions. Enclosed in parentheses are the corresponding standard errors.}
\begin{tabular}{c|cccc|cccc}
\hline
&\multicolumn{4}{c}{Cor=0.3}&\multicolumn{4}{c}{Cor=0.7}\\
\hhline{~--------}
&NV&TPR&FPR&PE&NV&TPR&FPR&PE\\
\hline
Two-step&\makecell{2.82\\(.994)}&\makecell{.753\\(.229)}&\makecell{.003\\(.004)}&\makecell{.171\\(.033)}&\makecell{2.05\\(.829)}&\makecell{.557\\(.170)}&\makecell{.002\\(.003)}&\makecell{.174\\(.022)}\\
Lasso&\makecell{37.0\\(38.2)}&\makecell{.690\\(.259)}&\makecell{.176\\(.194)}&\makecell{.312\\(.069)}&\makecell{21.9\\(37.9)}&\makecell{.327\\(.291)}&\makecell{.103\\(.193)}&\makecell{.288\\(.047)}\\
GAMSEL&\makecell{15.4\\(16.0)}&\makecell{.573\\(.285)}&\makecell{.069\\(.079)}&\makecell{.342\\(.065)}&\makecell{12.5\\(9.15)}&\makecell{.397\\(.271)}&\makecell{.057\\(.044)}&\makecell{.264\\(.033)}\\
GAMBoost&\makecell{44.2\\(5.21)}&\makecell{.977\\(.085)}&\makecell{.209\\(.026)}&\makecell{.268\\(.033)}&\makecell{33.7\\(4.52)}&\makecell{.860\\(.178)}&\makecell{.158\\(.014)}&\makecell{.203\\(.026)}\\
\hline
\end{tabular}
\label{cor3table}
\end{table}

This underselection for correlated predictors has been an issue for the lasso and adaptive lasso methods. For nonparametric additive models, \cite{huang2010variable} found the same issue when dealing with correlated predictors. Also the NIS proposed by \cite{fan2011nonparametric} did not perform well in correlated predictors compared to uncorrelated case. \added{Our two-step approach is not affected too much with the correlation, in terms of both variable selection and prediction.}

\end{exmp}

\added{It also happens in the real world that the signal strength is low. Therefore, it is interesting to consider a case where we have lower signal strength than in example \ref{uncorex}.}

\begin{exmp}
\label{ex:lowsignal}
\added{In this example, we reduce the signal strength of example \ref{uncorex} by a factor of 2, while all other assumptions are kept the same. The results are shown in Table \ref{table:lowsignal} on page \pageref{table:lowsignal}. From the table we see that minimal signal strength is an important factor to the performance of variable selection in the generalized models. The performance is impacted by the signal strength for all models. The two-step approach still have the closest number of nonzero variables to the ground truth. Though the true positive rate is lower than that of the Lasso or the GAMBoost, the latter two methods have too many false positives. The Lasso or GAMBoost selects too many variables and should not be considered as good variable selection methods. Moreover, the prediction error of the two-step approach remain the best among all four methods.}
\begin{table}[h]
\centering
\caption{Simulation results for the two-step approach compared with the Lasso, GAMSEL and GAMBoost in Example \ref{ex:lowsignal}, with $n=100$, $p=200$, $s=3$ and signal strength reduced. NV, average number of the variables being selected; TPR, the true positive rate; FPR, the false positive rate; and PE, prediction error (here is the misclassification rate). Results are averaged over 100 repetitions. Enclosed in parentheses are the corresponding standard errors.}
\begin{tabular}{c|cccc}
\hline
&NV&TPR&FPR&PE\\
\hline
Two-step&\makecell{3.91\\(2.05)}&\makecell{.703\\(.240)}&\makecell{.009\\(.009)}&\makecell{.218\\(.033)}\\
Lasso&\makecell{30.0\\(30.5)}&\makecell{.770\\(.304)}&\makecell{.142\\(.154)}&\makecell{.258\\(.036)}\\
GAMSEL&\makecell{15.3\\(18.0)}&\makecell{.510\\(.266)}&\makecell{.070\\(.090)}&\makecell{.377\\(.054)}\\
GAMBoost&\makecell{50.3\\(5.11)}&\makecell{.980\\(.079)}&\makecell{.240\\(.026)}&\makecell{.308\\(.028)}\\
\hline
\end{tabular}
\label{table:lowsignal}
\end{table}
\end{exmp}

\subsubsection{\added{Other link functions}}
\added{In this subsection, we study the performance of the two-step approach numerically on the Poisson regression and Gamma regression. In the Poisson regression, we have
\begin{equation}
\label{simlog}
y_i\sim Poisson(\theta_i),\ i=1,...,n,
\end{equation}
where $\theta_i=\exp[\alpha+\sum_{j=1}^pf_j(x_{ij})]$ and $x_{ij}$ is the $(i,j)-th$ element of the design matrix $X$. In the Gamma regression, we have
\begin{equation}
\label{simlog}
y_i\sim Gamma(\theta_i, \phi),\ i=1,...,n,
\end{equation}
where $\theta_i=\exp[\alpha+\sum_{j=1}^pf_j(x_{ij})]$ and $x_{ij}$ is the $(i,j)-th$ element of the design matrix $X$. The dispersion parameter $\phi$ is assumed to be known. Without loss of generality, we take $\phi=1$.}

\begin{exmp}
\label{ex:otherlink}
\added{In this example, we keep the same set up as in example \ref{uncorex} to generate the design matrix, and use the Poisson distribution/Gamma distribution  above to generate response variables. All other parameters are kept the same as in example \ref{uncorex}, but the signal strength is set to $1/4$ of the original signal strength, and we set $n=100$, $p=200$ and $s=3$. We compare the two-step approach with generalized linear models (GLM) and the GAMBoost. Note that the GAMSEL only supports Gaussian and Binomial link, thus is not used as a comparison here. The GAMBoost only supports generalized models with canonical link. The canonical link for Gamma regression suffers from the risk that the mean might fall outside of its range, thus the canonical link is not useful in practice. Therefore, we only use GAMBoost in Poisson regression as a comparison. Our algorithm works for both Gamma regression and Poisson regression, and to the best of our knowledge, is the only publicly available algorithm that supports both in the high-dimensional settings. The GLMs are run with the scikit-learn package in python.

The results are provided in Table \ref{table:otherlink}. We see the two-step approach works significantly better than the linear model, and than the GAMBoost in the Poisson regression case, except for the true positive rate. The GAMBoost has a perfect true positive rate, which is slightly better than that of our two-step approach. However, the same issue as before is that it selected too many variables and make the false positive rate much higher than tolerable. Moreover, the prediction performance on the two-step approach is also in the first place in both the cases.}

\begin{table}[h]
\centering
\caption{Simulation results for the two-step approach compared with the Lasso, GAMSEL and GAMBoost in Example \ref{ex:otherlink} for Poisson regression and Gamma regression with $n=100$, $p=200$ and $s=3$. NV, average number of the variables being selected; TPR, the true positive rate; FPR, the false positive rate; and PE, prediction error (here is the misclassification rate). Results are averaged over 100 repetitions. Enclosed in parentheses are the corresponding standard errors. The GAMBoost method does not support Gamma regression with non-canonical link function, while the canonical link falls outside of range, therefore it does not support Gamma regression.}
\begin{tabular}{c|cccc|cccc}
\hline
&\multicolumn{4}{c}{Poisson Regression}&\multicolumn{4}{c}{Gamma Regression}\\
\hhline{~--------}
&NV&TPR&FPR&PE&NV&TPR&FPR&PE\\
\hline
Two-step&\makecell{4.30\\(1.51)}&\makecell{.930\\(.172)}&\makecell{.008\\(.009)}&\makecell{2.34\\(.703)}&\makecell{3.57\\(0.98)}&\makecell{.997\\(.033)}&\makecell{.003\\(.005)}&\makecell{14.4\\(19.5)}\\
Lasso&\makecell{13.4\\(9.79)}&\makecell{.867\\(.189)}&\makecell{.054\\(.050)}&\makecell{3.51\\(.403)}&\makecell{12.5\\(7.72)}&\makecell{.887\\(.196)}&\makecell{.048\\(.039)}&\makecell{42.3\\(11.5)}\\
GAMBoost&\makecell{82.1\\(4.27)}&\makecell{1.00\\(.000)}&\makecell{.401\\(.022)}&\makecell{15.4\\(2.12)}&NA&NA&NA&NA\\
\hline
\end{tabular}
\label{table:otherlink}
\end{table}
\end{exmp}

\subsection{\textit{Real data examples}}

In this section, we provide \replaced{three }{}real data examples to illustrate our procedure. In the first example, we  consider the case $n>p$ \added{in the classification set up}, in the second example, we  consider the high-dimensional set up $n<p$ \added{in the classification set up, and in the third example, we consider a Gamma regression model.}

\begin{exmp}
\label{spam}
In this example, we use the data set in Example 1 of \cite{friedman2001elements}, the spam data as an example of the case $n>p$. The data set is available at https://web.stanford.edu/~hastie/ElemStatLearn/data.html. This data set has been studied in many different contexts with the objective being to predict whether an email is a spam or not based on a few features of the emails. There are $n=4601$ observations, among which 1813 (39.4\%) are spams. There are $p=57$ predictors, including 48 continuous real $[0,100]$ attributes of the relative frequency of 48 `spam' words out of the total number of words in the email, 6 continuous real $[0,100]$ attributes of the relative frequency of 6 `spam' characters out of the total number of characters in the email, 1 continuous real attribute of average length of uninterrupted sequences of capital letters, 1 continuous integer attribute of length of longest uninterrupted sequence of capital letters, and 1 continuous integer attribute of total number of capital letters in the e-mail. The data was first log transformed, since most of the predictors have long-tailed distribution, as mentioned in \cite{friedman2001elements}. They were then centered and standardised.

The data was split into a training data set with 3067 observations and a testing data set with 1534 observations. \added{We choose order $l=4$ which implies a cubis B-spline. We choose $m_n=15$, which implies there are $11$ inner knots, evenly placed over the empirical percentiles of the data. We compare the result with the logistic regression with Lasso penalty, the support vector machine (SVM) with Lasso penalty, and the sparse group lasso neural network (SGLNN, \cite{feng2017sparse}, see also \cite{yang2020statistical}). The Lasso and SMV are implemented with the skikit-learn module in python, and the SGLNN is implemented with the algorithm in the paper in python. By changing the tuning parameter or stopping criterion, we get estimations with different sparsity levels. All results are averaged over 50 repetitions. The classification error with different level of sparsity is shown in Figure \ref{fig:spamerr} on page \pageref{fig:spamerr}.} \added{The two-step approach and the neural network perform better than the linear models, which indicates a non-linear relationship. The two-step approach has maximum  accuracy 0.944, while that for the neural network is 0.946. The neural network performs a little better than the two-step approach due to its ability to model the interactions among predictors, but this difference is not significant. However, neural network has no interpretation and takes longer to train.  All four methods have performance increase as more predictors are included, which indicates that all predictors contributes to some effect to the prediction. However, we are able to reach more than 0.9 accuracy with only 15 predictors included. With the GIC criterion, the two-step approach selects $14.6\pm1.52$ predictors, with an average accuracy of $0.914\pm 0.015$.} The most frequently selected functions are shown in Figure \ref{fig:spamfunc} on page \pageref{fig:spamfunc}, \added{which also shows that these functions are truly non-linear}. The plots are of the original functions, i.e., before the logarithm transformation. \added{The estimated functions are close to the results in \cite{friedman2001elements}, Chapter 9, with slight scale difference due to different penalization.} The results show that the additive model by the adaptive group lasso is more suitable for this data than \replaced{linear models. }{}

\begin{figure}[h]
    \centering
    \includegraphics[width=\textwidth]{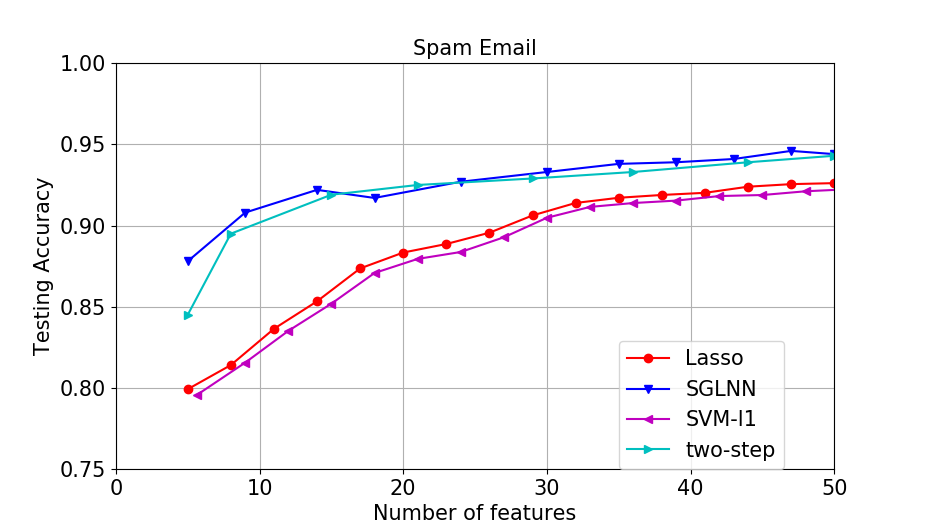}
    \caption{The classification accuracy against the number of nonzero variables measured on a testing set for Example \ref{spam} over 50 repetitions. The two-step approach, the logistic regression with Lasso, the $l_1$ norm penalized SVM and the sparse group lasso neural network are included in comparison.}
    \label{fig:spamerr}
\end{figure}
\begin{figure}[h]
    \centering
    \includegraphics[width=\textwidth]{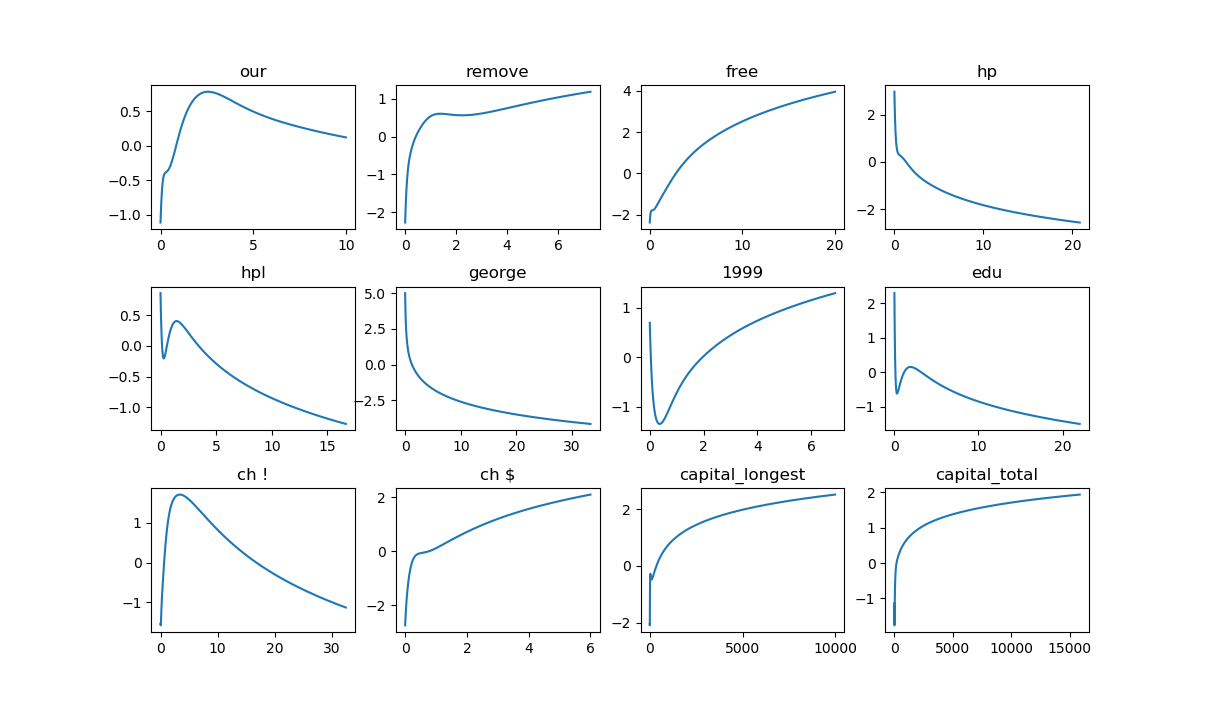}
    \caption{The estimated functions for the most frequently selected functions for Example \ref{spam}.}
    \label{fig:spamfunc}
\end{figure}
\end{exmp}

\begin{exmp}
\label{prostate}
For  high-dimensional classification example, we use the prostate cancer gene expression data described in http://featureselection.asu.edu/datasets.php. The data set has a binary response. 102 observations were studied on 5966 predictor variables, which indicates that the data set is really a high dimensional data set. The responses have values 1 (50 sample points) and 2 (52 sample points), where 1 indicates \replaced{normal }{} and 2 indicates \replaced{tumor }{}. All predictors are continuous predictors, with positive values.

To see the performance of our procedure, we ran 100 replications. In each replication, we randomly choose 76 of the observations as training data set and the rest 26 observations as testing data set. \added{We choose order $l=4$ which implies a cubis B-spline. We choose $m_n=9$, which implies there are $5$ inner knots, evenly placed over the empirical percentiles of the data. Similar to the last example, we compare the result with the logistic regression with Lasso penalty, the SVM with Lasso penalty, and SGLNN. The classification error with different level of sparsity is shown in Figure \ref{fig:prostateerr} on page \pageref{fig:prostateerr}. From the figure we see that compared with linear methods such as the logistic regression or support vector machine, the non-parametric approaches converges faster. The two-step approach reaches a testing accuracy of 0.945 when around 15 variables are included in the model, while the linear methods need over 30 variables to reach competitive results. Compared with neural network, the two-step approach is easier to implement with stabilized performances. A drawback of the non-parametric methods is to easily overfit for small sample, and that's the reason the performance drops as too many variables entered the into the model. With the GIC criterion, the two-step approach selects $3.25\pm 1.67$ predictors, with an average accuracy of $0.914\pm0.016$. To show the non-linear relationship, figure \ref{fig:prostatefunc} on page \pageref{fig:prostatefunc} shows the estimated functions for the 6 most frequently selected variables.}

\begin{figure}[h]
    \centering
    \includegraphics[width=\textwidth]{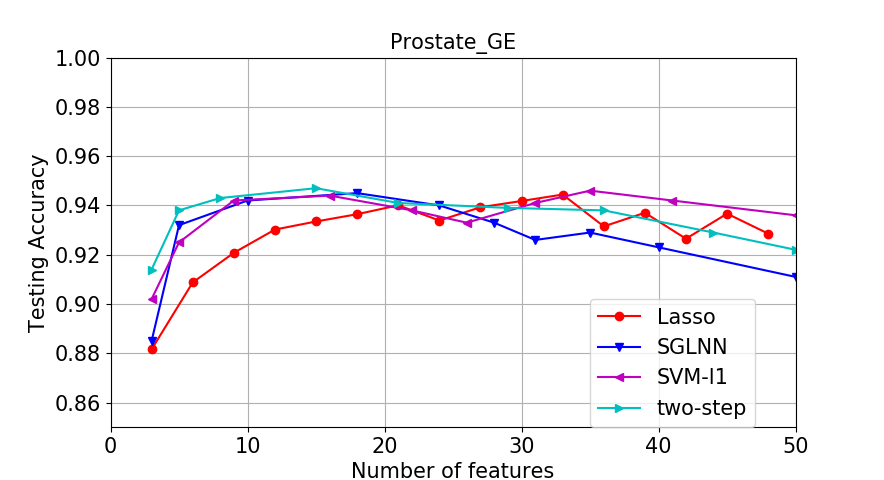}
    \caption{The classification accuracy against the number of nonzero variables measured on a testing set for Example \ref{prostate} over 500 repetitions. The two-step approach, the logistic regression with Lasso, the $l_1$ norm penalized SVM and the sparse group lasso neural network are included in comparison.}
    \label{fig:prostateerr}
\end{figure}

\begin{figure}[h]
    \centering
    \includegraphics[width=\textwidth]{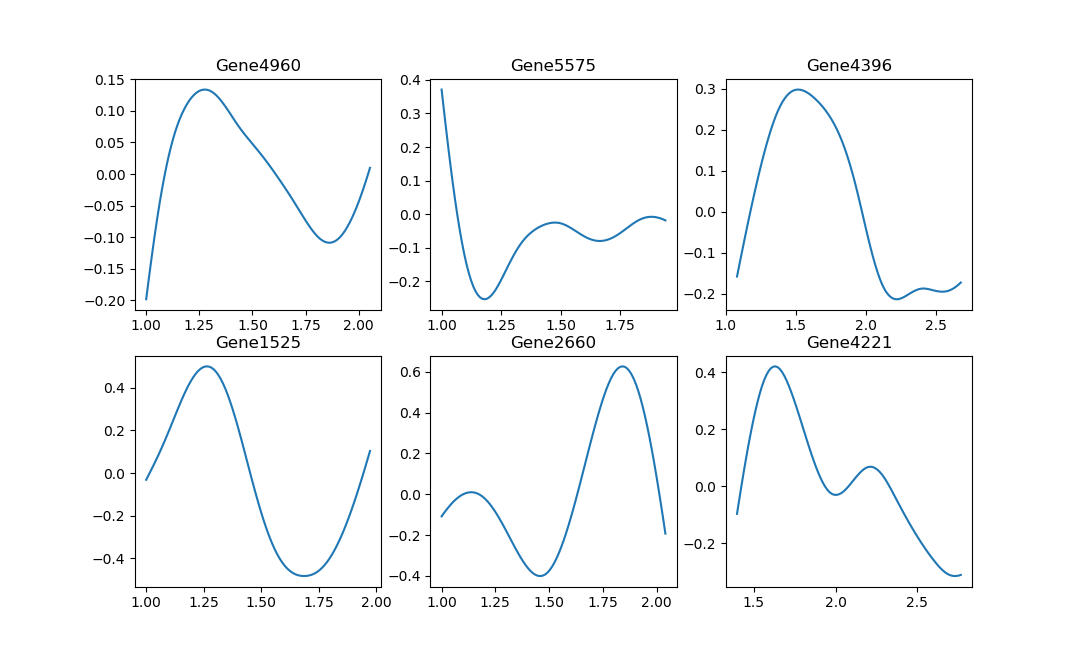}
    \caption{The estimated functions for the most frequently selected functions ordered by descending in frequency for Example \ref{prostate}.}
    \label{fig:prostatefunc}
\end{figure}

\end{exmp}

\begin{exmp}
\label{noaa}
\added{In this example, we investigate the performance of the two-step approach on Gamma regression. The data set is from National Oceanic and Atmospheric Administration (NOAA). We use the storm data, which includes the occurrence of storms in the United States with the time, location, property damage, a narrative description and etc. Here we only take the data in Michigan from 2010 to 2018 and keep the narrative description as our predictor variable and the property damage as our response variable. The description is in text, therefore we applied wording embedding algorithm Word2vec \citep{mikolov2013efficient} to transform each description into a numeric representation vector of length $p=701$, similar word embedding preprocessing can be found in \cite{lee2020actuarial}. The response variable property damage has a long tail distribution, thus we use a Gamma regression here. After removing outliers, the data set contains 3085 observations. In order to study the high-dimensional case, we randomly sample $10\%$ of the observations as our training data ($n=309$) and the rest are used for validation. Moreover, the response is normalized with the location and scale parameters of gamma distribution.

To see the performance of our procedure, we ran 50 replications. We choose order $l=4$ which implies a cubis B-spline. We choose $m_n=9$, which implies there are $5$ inner knots, evenly placed over the empirical percentiles of the data. Since there's limited libraries available for variable selection under high-dimensional gamma model, we compare the two-step approach with the linear regression with Lasso on a logarithm transformation on the response variable. The prediction error with different level of sparsity is shown in Figure \ref{fig:noaaerr} on page \pageref{fig:noaaerr}. With the GIC criterion, the two-step approach selects $34.45\pm 3.52$ predictors, with an average MSE of $0.004334\pm0.000115$. However, from the plot we see that the linear model was not able to reach this accuracy through the whole solution path, with the best accuracy of $0.004337$ at around $80$ nonzero variables. This example also shows the superior of the non-parametric model over linear models.}
\begin{figure}[h]
    \centering
    \includegraphics[width=\textwidth]{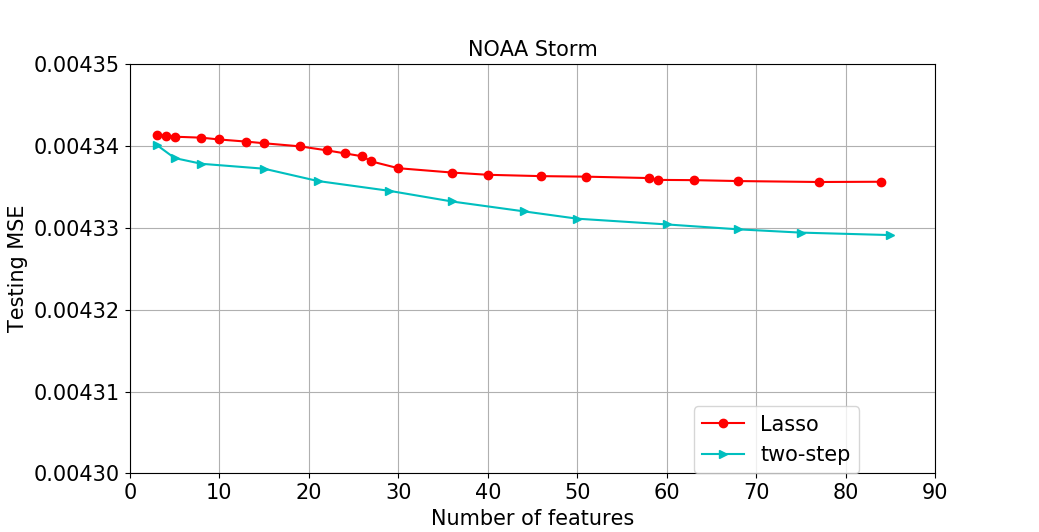}
    \caption{The testing MSE against the number of nonzero variables measured on a testing set for Example \ref{noaa} over 50 repetitions. The two-step approach and logarithm transformation with the Lasso are included in comparison.}
    \label{fig:noaaerr}
\end{figure}

\end{exmp}

\section{Discussion}
\label{Sec:Discussion}
In this paper, we considered ultra high-dimensional ($\log p_n=O(n^{\rho})$) generalised additive model with a diverging number of nonzero functions ($s_n\rightarrow 0\ as\ n\rightarrow\infty$). After using basis expansion on the nonparametric functions, we used two step procedures---group lasso and adaptive group lasso to select the true model. We have proved the screening consistency of the group lasso estimator and the selection consistency of the adaptive group lasso estimator. The rates of convergence of both estimators were also derived, which proved that the adaptive group lasso does have an improvement on the estimator. The whole paper provides a solid foundation for the existing methods. Finally we proved that under this nonparametric set up, the generalised information criterion (GIC) is a good way to select the tuning parameter that consistently selects the true model.

In this paper, we used a fixed design on the data matrix $\boldX$. A random design on $\boldX$ could be considered, i.e., $X$ has a continuous distribution function $f_X(X)$ on its interval $[a,b]$, however, extra assumptions such as the boundedness of the density function are needed to reach the same result. Also we proved the selection consistency of the GIC procedure on the adaptive group lasso estimator, conditioning  that the initial estimator satisfies (\ref{firstrequirement}), which is possessed by the group lasso procedure with probability tending to 1.  However, the theory of screening consistency for the group lasso estimator is still to be established. This is a challenging problem, since there doesn't have to exist a tuning parameter that gives selection consistency in the group lasso procedure, but this is an interesting problem that deserves further investigation. We also discussed the subset selection and subset selection with shrinkage under our set up. The theoretical investigation suggests the other penalty functions may not have clear advantages over the proposed procedure.

\added{Moreover, the heteroskedastic error case is also attracting in high-dimensional GAM. The square root Lasso \citep{belloni2011square} has been proved to overcome this issue, however, it hasn't been extended to the non-parametric set up. It could be interesting to apply square root Lasso on the GAM to incorporate this case. This is a demanding topic that deserves further investigation as well.}

\bibliographystyle{apalike}
\bibliography{ref}

\newpage

\begin{center}
\begin{multicols}{2}

        Kaixu Yang\\
        Department of Statistics and Probability\\
        Michigan State University\\
        619 Red Cedar Rd. Room 507\\
        East Lansing, MI, 48824\\
        yangkaix@msu.edu
        
\columnbreak
  
        Tapabrata Maiti\\
        Department of Statistics and Probability\\
        Michigan State University\\
        619 Red Cedar Rd. Room 424\\
        East Lansing, MI, 48824\\
        maiti@msu.edu
\end{multicols}
\end{center}
\newpage

\begin{appendices}

\section{Derivation of assumption 2}
\label{appendix1}
Though assumption 2 is imposed on the fixed design matrix, however, it holds if the design matrix $\boldX$ is drawn from a continuous density and the density $g_j$ of $X_j$ is bounded away from 0 and infinity by $b$ and $B$, respectively, on the interval $[a,b]$. Let $\bdelta_A$ be the sub-vector of $\bdelta$ which include all nonzero entries. Without loss of generality, let $\bdelta_A=\{\bdelta_1,...,\bdelta_k\}$, where $\bdelta_k\in\R^{m_n}$ and $k=O(s_n)$. Let $\Phi_A$ be the corresponding sub-matrix of $\Phi$.

By lemma 3 in \cite{stone1985additive}, if the design matrix $\boldX$ is drawn from a continuous density and the density $g_j$ of $X_j$ is bounded away from 0 and infinity by $b$ and $B$, respectively, on the interval $[a,b]$, and $\text{card}_B(\bdelta)=O(s_n)$, we have
$$\|\Phi_1\bdelta_1+...+\Phi_k\bdelta_k\|_2\geq \gamma_2^{k-1}(\|\Phi_1\bdelta_1\|_2+...+\|\Phi_k\bdelta_k\|_2)$$for some positive constant $\gamma_2$ such that $\delta_0<1-2\gamma_2^2<1$, where $\delta_0=((1-bB^{-1})/2)$. Together with the triangle inequality, we have
$$\gamma_2^{k-1}(\|\Phi_1\bdelta_1\|_2+...+\|\Phi_k\bdelta_k\|_2)\leq\|\Phi_A\bdelta_A\|_2\leq \|\Phi_1\bdelta_1\|_2+...+\|\Phi_k\bdelta_k\|_2$$
By simple algebra, we have
$$\gamma_2^{2k-2}(\|\Phi_1\bdelta_1\|_2^2+...+\|\Phi_k\bdelta_k\|_2^2)\leq\|\Phi_A\bdelta_A\|_2^2\leq 2(\|\Phi_1\bdelta_1\|_2^2+...+\|\Phi_k\bdelta_k\|_2^2)$$
For any $j=1,...,k$, by lemma 6.2 in \cite{zhou1998local}, we have
$$c_1m_n^{-1}\leq \lambda_{min}(n^{-1}\Phi_j^T\Phi_j)\leq \lambda_{max}(n^{-1}\Phi_j^T\Phi_j)\leq c_2m_n^{-1}$$
for some $c_1$ and $c_2$. Then we have
\begin{align*}
\frac{\bdelta^T\Phi^T\Phi\bdelta}{\|\bdelta\|_2^2}&=\frac{\|\Phi_A\bdelta_A\|_2^2}{\|\bdelta_A\|_2^2}\\
&\geq \frac{\gamma_2^{2k-2}(\|\Phi_1\bdelta_1\|_2^2+...+\|\Phi_k\bdelta_k\|_2^2)}{\|\bdelta_A\|_2^2}\\
&=\gamma_2^{2k-2}\left(\frac{\|\Phi_1\bdelta_1\|_2^2}{\|\bdelta_1\|_2^2}\frac{\|\bdelta_1\|_2^2}{\|\bdelta_A\|_2^2}+...+\frac{\|\Phi_k\bdelta_k\|_2^2}{\|\bdelta_k\|_2^2}\frac{\|\bdelta_k\|_2^2}{\|\bdelta_A\|_2^2}\right)\\
&\geq \gamma_2^{2k-2}c_1nm_n^{-1}\left(\frac{\|\bdelta_1\|_2^2}{\|\bdelta_A\|_2^2}+...+\frac{\|\bdelta_k\|_2^2}{\|\bdelta_A\|_2^2}\right)\\
&=\gamma_2^{2k-2}c_1nm_n^{-1}
\end{align*}
Let $\gamma_0=\gamma_2^{-2}c_1$ and observe that $k=O(s_n)$, we have
$$\frac{\bdelta^T\Phi^T\Phi\bdelta}{n\|\bdelta\|_2^2}\geq \gamma_0\gamma_2^{2s_n}m_n^{-1}$$
Similarly, we have
\begin{align*}
\frac{\bdelta^T\Phi^T\Phi\bdelta}{\|\bdelta\|_2^2}&=\frac{\|\Phi_A\bdelta_A\|_2^2}{\|\bdelta_A\|_2^2}\\
&\leq \frac{2(\|\Phi_1\bdelta_1\|_2^2+...+\|\Phi_k\bdelta_k\|_2^2)}{\|\bdelta_A\|_2^2}\\
&=2\left(\frac{\|\Phi_1\bdelta_1\|_2^2}{\|\bdelta_1\|_2^2}\frac{\|\bdelta_1\|_2^2}{\|\bdelta_A\|_2^2}+...+\frac{\|\Phi_k\bdelta_k\|_2^2}{\|\bdelta_k\|_2^2}\frac{\|\bdelta_k\|_2^2}{\|\bdelta_A\|_2^2}\right)\\
&\leq 2c_2nm_n^{-1}\left(\frac{\|\bdelta_1\|_2^2}{\|\bdelta_A\|_2^2}+...+\frac{\|\bdelta_k\|_2^2}{\|\bdelta_A\|_2^2}\right)\\
&=2c_2nm_n^{-1}
\end{align*}
Let $\gamma_1=c_2$, we have
$$\frac{\bdelta^T\Phi^T\Phi\bdelta}{n\|\bdelta\|_2^2}\leq \gamma_1m_n^{-1}$$

\section{Proofs and lemma and theorems}
\label{appendix2}

The following lemmas are needed in proving theorems.
\begin{lemma}
\label{highprob1}
For any sequence $r_n>0$, under assumption 1 and 3, we have for bounded response such that $|y_i|<c/2$ that
\begin{equation}
\prob\left(\left\|\frac{\Phi^T\left(\boldy-\bmu_y\right)}{n}\right\|_{\infty}\leq r_n\right)\geq 1-2p_nm_n\exp(-\frac{nr_n^2}{2c^2c_{\Phi}^2})
\end{equation}
Specifically, for a diverging sequence $t_n$, taking
$$r_n=\sqrt{2}cc_{\Phi}\sqrt{\frac{\log (p_nm_n)+t_n}{n}}$$
we have for response such that $|y_i|<c/2$ that
\begin{equation}
\prob\left(\left\|\frac{\Phi^T\left(\boldy-\bmu_y\right)}{n}\right\|_{\infty}\leq r_n\right)\geq 1-2\exp(-t_n)
\end{equation}
\end{lemma}

\noindent
\begin{proof}
Observe that
$$\frac{\Phi_j^T\left(\boldy-\bmu_y\right)}{n}=\sum_{i=1}^n\left(\frac{\phi_{ij}(y_i-\mu_{y_i})}{n}\right):=\sum_{i=1}^n\gamma(y_i)$$
It's easy to verify that $E\gamma(y_i)=0$ for $i=1,...,n$ and $|\gamma(y_i)|=\left|\phi_{ij}(y_i-\mu_{y_i})/n\right|\leq cd_i$ for $i=1,...,n$. By assumption 1, we have $\sum_{i=1}^nd_i^2\leq c_{\Phi}^2/n$ for $i=1,...,n$. Apply Bonferroni's inequality and Hoeffding's inequality, we have
\begin{align*}
\prob\left(\left\|\frac{\Phi^T\left(\boldy-\bmu_y\right)}{n}\right\|_{\infty}\leq r_n\right)&=1-\prob\left(\left\|\frac{\Phi^T\left(\boldy-\bmu_y\right)}{n}\right\|_{\infty}\geq r_n\right)\\
&=1-\prob\left(\bigcup_{j=1}^{m_n\times p_n}\left\{\left|\frac{\Phi_j^T\left(\boldy-\bmu_y\right)}{n}\right|\geq r_n\right\}\right)\\
&\geq 1-\sum_{j=1}^{m_n\times p_n}\prob\left(\left|\frac{\Phi_j^T\left(\boldy-\bmu_y\right)}{n}\right|\geq r_n\right)\\
&\geq 1-m_n\times p_n\times 2\exp\left(-\frac{nr_n^2}{2u_n^2c_{\Phi}^2}\right)-c_2n^{1-c_3c_4^2}\\
\end{align*}
with our choice of
$$r_n=\sqrt{2}cc_{\Phi}\sqrt{\frac{\log (p_nm_n)+t_n}{n}}$$
we have
\begin{align*}
\prob\left(\left\|\frac{\Phi^T\left(\boldy-\bmu_y\right)}{n}\right\|_{\infty}\leq r_n\right)&\geq 1-m_n\times p_n\times 2\exp\left(-\frac{nr_n^2}{8c^2c_{\Phi}^2}\right)\\
&=1-m_n\times p_n\times 2\exp\left(-\frac{n2c^2c_{\Phi}^2(\log(p_nm_n)+t_n)}{2c^2c_{\Phi}^2n}\right)\\
&=1-2\exp(-t_n)
\end{align*}
\end{proof}

\begin{lemma}
\label{highprob2}
In the unbounded response case, under assumptions 1 and 3, let $T_n=n^{-1}\|\Phi_j^T(\boldy-\bmu_y)\|_{\infty}=\max_{j=1,...,p_nm_n}n^{-1}|\Phi_j^T(\boldy-\bmu_y)|$, we have
\begin{equation}
ET_n=O(1)n^{-1/2}\sqrt{p_nm_n}
\end{equation}
and then for any diverging sequence $a_n$,
\begin{equation}
\prob\left(T_n\geq a_n\sqrt{\frac{\log(p_nm_n)}{n}}\right)\rightarrow 0\ \text{as}\ n\rightarrow 0
\end{equation}
\end{lemma}
\begin{proof}
By the maximal inequality for sub-Gaussian random variables, for example, see Lemmas 2.2.1 and 2.2.2 in \cite{van1996weak} and application see lemma 2 of \cite{huang2010variable}, we have
$$ET_n\leq Cn^{-1}\sqrt{\log(p_nm_n)}\max_j\|\Phi_j\|_2$$
Then by assumption 1, we have
$$ET_n=O(1)n^{-1/2}\sqrt{p_nm_n}$$
Since $T_n\geq 0$, by Markov's inequality, we have
\begin{equation}
\prob\left(T_n\geq a_n\sqrt{\frac{\log(p_nm_n)}{n}}\right)\leq\frac{ET_n}{n^{-1/2}\sqrt{\log(p_nm_n)}}=\frac{C}{a_n}\rightarrow 0\ \text{as}\ n\rightarrow\infty
\end{equation}
\end{proof}

\begin{remark}
From the two lemmas we see that the difference between the bounded response case and the unbounded response case is the upper bound for the maximum of the random errors. For the bounded case, the error could be bounded by
$$r_n=C\sqrt{\frac{\log (p_nm_n)+t_n}{n}}$$ with any diverging sequence $t_n$. If we take $t_n=O(\log(p_nm_n))$, we have for a different $C$, the bounded response errors to be bounded by
$$r_n=C\sqrt{\frac{\log (p_nm_n)}{n}}$$
with probability converging to 1. For the unbounded response case, with probability converging to 1, we need a diverging sequence $a_n$ instead of a constant multiplied to the main term, i.e.,
$$r_n=a_n\sqrt{\frac{\log (p_nm_n)}{n}}$$
This difference is reflected on the choice of the tuning parameter $\lambda$.
\end{remark}

\noindent
\textbf{Proof of theorem \ref{screening1}}
\begin{proof}

First observe that due to the spline approximation, an error is bought into the model. Let $\theta=\sum_{j=1}^{p_n}f_j$ and $\theta^*=\sum_{j=1}^{p_n}f_{nj}$. By the proof of theorem 1 in \cite{huang2010variable}, we have
$$\|f_j-f_{nj}\|_{\infty}=O(m_n^{-d})$$
Therefore, we have
\begin{align*}
|\theta-\theta^*|\leq \|\sum_{j=1}^{p_n}(f_j-f_{nj})\|_{\infty}\leq\sum_{j=1}^{s_n}\|f_j-f_{nj}\|_{\infty}=O(s_nm_n^{-d})
\end{align*}
Use Taylor expansion on $b'(\theta)$ around $\theta^*$, we have
$$b'(\theta)-b'(\theta^*)=b''(\theta^{**})(\theta-\theta^*)$$
where $\theta^{**}$ lies between $\theta$ and $\theta^*$. By assumption 3, we have
\begin{equation}
|\mu_{y_i}-\mu_{y_i}^*|=|b'(\theta)-b'(\theta^*)|\leq c_1^{-1}|\theta-\theta^*|=O(s_nm_n^{-d}),\ i=1,...,n
\end{equation}
where $\mu_{y_i}^*$ is the mean of the $i^{th}$ observation evaluated at the spline approximated functions. Therefore, we have
$$\|\bmu_y-\bmu_y^*\|_{\infty}=O(s_nm_n^{-d})$$
As a direct result, we have
\begin{equation}
\label{splineapproxerr}
\frac{1}{n}\|\bmu_y-\bmu_y^*\|_2^2=O(s_n^2m_n^{-2d})
\end{equation}

We start with part (i). The proof of this part is similar to the proof of part (i) of
theorem 1 in \cite{huang2010variable}. But because of the non-identity link function, here we have to make some changes. By KKT conditions, a necessary and sufficient condition for $\hat{\bbeta}$ to be a minimiser of the target function is
\begin{equation}
\begin{dcases}
\frac{1}{n}\Phi_k^T(\boldy-\hat{\bmu}_{\boldy}^*)=\frac{\lambda_{n1}\hat{\bbeta}_k}{\|\hat{\bbeta}_k\|_2},\ \forall\ k\ s.t.\ \|\hat{\bbeta}_k\|_2>0\\
\frac{1}{n}\Phi_k^T(\boldy-\hat{\bmu}_{\boldy}^*)\in[-\lambda_{n1},\lambda_{n1}], \ \forall\ k\ s.t.\ \|\hat{\bbeta}_k\|_2=0
\end{dcases}
\end{equation}
where $\hat{\bmu}_{\boldy}^*$ is the mean of response approximated by splines and evaluated at the solution $\hat{\bbeta}$ and the second belonging relationship is element-wise. Let
$$s_k=\frac{\Phi_k^T(\boldy-\hat{\bmu}_{\boldy}^*)}{n\lambda_{n1}}$$
Then we have
\begin{equation}
\begin{dcases}
\|s_k\|_2=1,\ \forall\ k\ s.t.\ \|\hat{\bbeta}_k\|_2>0\\
\|s_k\|_2\leq 1,\ \forall\ k\ s.t.\ \|\hat{\bbeta}_k\|_2=0
\end{dcases}
\end{equation}
We consider the following subsets of $\{1,...,p\}$. Let $A_1$ be such that
\begin{equation}
\label{generalA1}
\left\{k:\|\hat{\bbeta}_k\|_2>0\right\}\subset A_1\subset\left\{k:\frac{1}{n}\Phi_k^T(\boldy-\hat{\bmu}_{\boldy}^*)=\frac{\lambda_{n1}\hat{\bbeta}_k}{\|\hat{\bbeta}_k\|_2}\right\}\cup\{1,...,s_n\}
\end{equation}
Let $A_2=\{1,...,p\}\backslash A_1$, $A_3=A_1\backslash T$, $A_4=A_1\cap T^c$, $A_5=A_2\backslash T^c$ and $A_6=A_2\cap T^c$. Therefore, the relationships are
\begin{center}
\begin{tabular}{ccc}
\hline
&$j\in T$&$j\in T^c$\\
$A_1$: selected $j$ and some $j\in T$&$A_3$&$A_4$\\
$A_2$: $j$ not in $A_1$ (includes unselected only)&$A_5$&$A_6$\\
\hline
\end{tabular}
\end{center}
Then we have
\begin{equation}
\label{kkt1}
\Phi_{A_1}^T(\boldy-\hat{\bmu}_{A_1}^*)=S_{A_1}
\end{equation}
where $S_{A_1}=(S_{K_1}^T,...,S_{K_{q_1}}^T)^T$, $S_{K_i}=n\lambda_{n1}s_{k_i}$ and $\hat{\bmu}_{A_1}^*=b'(\Phi_{A_1}\hat{\bbeta}_{A_1})$. Also from the inequality in KKT, we have
\begin{equation}
\label{kkt2}
-C_{A_2}\leq\Phi_{A_2}^T(\boldy-\hat{\bmu}_{A_1}^*)\leq C_{A_2}
\end{equation}
where $C_{A_2}=(C_{K_1}^T,...,C_{k_{q_2}}^T)^T$ and $C_{K_i}=n\lambda_{n1}\indicator_{\{\|\hat{\bbeta}_{K_i}\|_2=0\}}\cdot e_{m_n\times 1}$, where all the elements of $e$ are 1. Let $\bepsilon^*=\boldy-\bmu_{\boldy}^*$, then from (\ref{kkt1}) we have
$$\Phi_{A_1}^T(\bmu_{\boldy}^*+\bepsilon^*-\hat{\bmu}_{A_1}^*)=S_{A_1}$$
use Taylor expansion on $\bmu_{\boldy}^*$ around $\hat{\bmu}_{A_1}^*$, we have
$$\Phi_{A_1}^T\bSigma_1\Phi_{A_1}(\bbeta_{A_1}-\hat{\bbeta}_{A_1})+\Phi_{A_1}^T\bSigma_1\Phi_{A_2}\bbeta_{A_2}+\Phi_{A_1}^T\bepsilon^*=S_{A_1}$$
where $\bSigma_1=\bSigma(\btheta_1)$, $\btheta_1$ lies on the line segment joining $\Phi\bbeta$ and $\Phi_{A_1}\hat{\bbeta}_{A_1}$, and $\bSigma(\btheta)=\text{diag}(b''(\theta_1),...,b''(\theta_n))$ is the diagonal variance matrix evaluated at $\btheta$. From (\ref{kkt2}), we have
$$-C_{A_2}\leq\Phi_{A_2}^T\bSigma_1\Phi_{A_1}(\bbeta_{A_1}-\hat{\bbeta}_{A_1})+\Phi_{A_2}^T\bSigma_1\Phi_{A_2}\bbeta_{A_2}+\Phi_{A_2}^T\bepsilon^*\leq C_{A_2}$$
Let $\bSigma_{ij}=\Phi_{A_i}^T\bSigma(\btheta_1)\Phi_{A_j}/n$, we have
$$\bSigma_{11}(\bbeta_{A_1}-\hat{\bbeta}_{A_1})+\bSigma_{12}\bbeta_{A_2}=S_{A_1}$$
and
$$-C_{A_2}\leq\bSigma_{21}(\bbeta_{A_1}-\hat{\bbeta}_{A_1})+\bSigma_{22}\bbeta_{A_2}+\Phi_{A_2}^T\bepsilon^*\leq C_{A_2}$$
With our choice of $\lambda_{n1}$, the constants are sufficient large, by lemma 1 in \cite{wei2010consistent}, the eigenvalues of $\bSigma_{11}$ are bounded from below. Thus without loss of generality, we assume $\bSigma_{11}$ is invertible. Then we have
\begin{equation}
\label{kkt11}
\frac{\bSigma_{11}^{-1}S_{A_1}}{n}=\bbeta_{A_1}-\hat{\bbeta}_{A_1}+\bSigma_{11}^{-1}\bSigma_{12}\bbeta_{A_2}+\frac{\bSigma_{11}^{-1}}{n}\Phi_{A_1}^T\bepsilon^*
\end{equation}
and
\begin{equation}
\label{kkt22}\frac{\|\bSigma^{-1/2}(\bmu_y-\bmu_y^*)\|_2}{n}+
n\bSigma_{22}\bbeta_{A_2}-n\bSigma_{21}\bSigma{11}^{-1}\bSigma_{12}\bbeta_{A_2}\leq C_{A_2}-\Phi_{A_2}^T\bepsilon^*-\bSigma_{21}\bSigma_{11}^{-1}S_{A_1}+\bSigma_{21}\bSigma_{11}^{-1}\Phi_{A_1}^T\bepsilon^*
\end{equation}
Define
$$V_{1j}=\frac{1}{\sqrt{n}}\bSigma_{11}^{-1/2}Q_{A_j1}^TS_{A_j},\ j=1,3,4$$
and
$$w_k=\bSigma_1^{1/2}(\boldI-\boldP_1)\bSigma_1^{1/2}\Phi_{A_k}\bbeta_{A_k},k=2,...,6$$
where
$$\boldP_1=\bSigma^{1/2}\Phi_{A_1}(\Phi_{A_1}^T\bSigma\Phi_{A_1})^{-1}\Phi_{A_1}^T\bSigma^{1/2}$$
and $Q_{A_jk}$ is the matrix representing the selection of variables in $A_k$ from $A_j$.

Consider $j=4$. For any $k\in A_4$, we have $\|\hat{\bbeta}_k\|_2>0$, then $\|s_k\|_2^2=1$. Then we have $\|S_{A_4}\|_2^2=\sum_{k\in A_4}N(A_4)$, where $N(A_4)$ is the number of predictors in $A_4$. Thus
\begin{align*}
\|V_{14}\|_2^2&=\frac{1}{n}\|\bSigma_{11}^{-1/2}Q_{A_41}^TS_{A_4}\|_2^2\\
&\geq \frac{1}{n}c_1\|Q_{A_41}^TS_{A_4}\|_2^2\\
&=c_1n\sum_{k\in A_4}\|\lambda_{n1}s_k\|_2^2\\
&\geq c_1n\lambda_{n1}^2(q_1-s_n)
\end{align*}
That is
\begin{equation}
\label{numberofselectedbound}
(q_1-s_n)^+\leq\frac{\|V_{14}\|_2^2}{c_1n\lambda_{n1}^2}
\end{equation}
Then, we need to find a bound for $\|V_{14}\|_2^2$ and $q_1\leq (q_1-s_n)^++s_n$ will be bounded. Using (\ref{kkt11}) and consider
\begin{align*}
V_{14}^T(V_{14}+V_{13})&=S_{A_4}^TQ_{A_41}\frac{\bSigma_{11}^{-1}}{n}S_{A_1}\\
&=S_{A_4}^TQ_{A_41}(\bbeta_{A_1}-\hat{\bbeta}_{A_1}+\bSigma_{11}^{-1}\bSigma_{12}\bbeta_{A_2}+\frac{\bSigma_{11}^{-1}}{n}\Phi_{A_1}^T\bepsilon^*)\\
&=S_{A_4}^TQ_{A_41}\bSigma_{11}^{-1}\bSigma_{12}\bbeta_{A_2}+\frac{S_{A_4}^TQ_{A_41}\bSigma_{11}^{-1}}{n}\Phi_{A_1}^T\bepsilon^*+S_{A_4}^T(\bbeta_{A_4}-\hat{\bbeta}_{A_4})
\end{align*}
Observe $\bbeta_{A_4}=0$, and
$$S_{A_4}^T\hat{\bbeta}_{A_4}=\sum_{k\in A_4}\frac{\lambda_{n1}\hat{\bbeta}_{k}^T\hat{\bbeta}_{k}}{\|\hat{\bbeta}_{k}\|_2}=\sum_{k\in A_4}\lambda_{n1}\|\hat{\bbeta}_{k}\|_2>0$$
we have
$$V_{14}^T(V_{14}+V_{13})\leq S_{A_4}^TQ_{A_41}\bSigma_{11}^{-1}\bSigma_{12}\bbeta_{A_2}+\frac{S_{A_4}^TQ_{A_41}\bSigma_{11}^{-1}}{n}\Phi_{A_1}^T\bepsilon^*$$
On the other hand, by (\ref{kkt22}),
\begin{align*}
\|w_2\|_2^2&=\bbeta_{A_2}^T\Phi_{A_2}^T\bSigma_1^{1/2}(\boldI-\boldP_1)\bSigma_1(\boldI-\boldP_1)\bSigma_1^{1/2}\Phi_{A_2}\bbeta_{A_2}\\
&\leq c_1^{-1}\bbeta_{A_2}^T\Phi_{A_2}^T\bSigma_1^{1/2}(\boldI-\boldP_1)\bSigma_1^{1/2}\Phi_{A_2}\bbeta_{A_2}\\
&=c_1^{-1}\bbeta_{A_2}^T\Phi_{A_2}^T\bSigma_1\Phi_{A_2}\bbeta_{A_2}+c_1^{-1}\bbeta_{A_2}^T\Phi_{A_2}^T\bSigma_1\Phi_{A_1}(\Phi_{A_1}^T\bSigma_1\Phi_{A_1})^{-1}\Phi_{A_1}^T\bSigma_1\Phi_{A_2}\bbeta_{A_2}\\
&=c_1^{-1}\bbeta_{A_2}^T(n\bSigma_{22}\bbeta_{A_2}-n\bSigma_{21}\bSigma_{11}^{-1}\bSigma_{12}\bbeta_{A_2})\\
&\leq c_1^{-1}\bbeta_{A_2}^T(C_{A_2}-\Phi_{A_2}^T\bepsilon^*-\bSigma_{21}\bSigma_{11}^{-1}S_{A_1}+\bSigma_{21}\bSigma_{11}^{-1}\Phi_{A_1}^T\bepsilon^*)\\
&=c_1^{-1}\bbeta_{A_2}^TC_{A_2}-c_1^{-1}\bbeta_{A_2}^T(\Phi_{A_2}^T-\bSigma_{21}\bSigma_{11}^{-1}\Phi_{A_1}^T)\bepsilon^*-c_1^{-1}\bbeta_{A_2}^T\bSigma_{21}\bSigma_{11}^{-1}S_{A_1}\\
&=c_1^{-1}\bbeta_{A_2}^TC_{A_2}-c_1^{-1}\bbeta_{A_2}^T\bSigma_{21}\bSigma_{11}^{-1}S_{A_1}-c_1^{-1}w^T_2\bSigma^{-1}_1\bepsilon^*
\end{align*}
Then we have
$$V_{14}^T(V_{14}+V_{13})+c_1\|w_2\|_2^2\leq \left(\frac{S_{A_4}^TQ_{A_41}\bSigma_{11}^{-1}}{n}\Phi_{A_1}^T-w_2^T\bSigma_1^{-1}\right)\bepsilon^*-S_{A_3}^TQ_{A_31}\bSigma_{11}^{-1}\bSigma_{12}\bbeta_{A_2}+\bbeta_{A_2}^TC_{A_2}$$
Define
$$u=\frac{\Phi_{A_1}\bSigma_{11}^{-1}Q_{A_41}^TS_{A_4}/n-\bSigma_1^{-1}w_2}{\|\Phi_{A_1}\bSigma_{11}^{-1}Q_{A_41}^TS_{A_4}/n-\bSigma_1^{-1}w_2\|_2}$$
Observe
\begin{align*}
&\|\Phi_{A_1}^T\bSigma_{11}^{-1}Q_{A_41}^TS_{A_4}/n-\bSigma_1^{-1}w_2\|_2\\
\leq& 2(\|\Phi_{A_1}^T\bSigma_{11}^{-1}Q_{A_41}^TS_{A_4}/n\|_2^2+\|\bSigma_1^{-1}w_2\|_2^2)\\
\leq& 2\|\Phi_{A_1}^T\bSigma_{11}^{-1}Q_{A_41}^TS_{A_4}/n\|_2^2+2c_1^{-2}\|w_2\|_2^2\\
=&2\|V_{14}\|_2^2+2c_1^{-2}\|w_2\|_2^2
\end{align*}
Observe $c_1<c_1^{-1}$ implies $c_1<1$, then
\begin{align}
\label{boundv14w2v1413}
\|V_{14}\|_2^2+c_1\|w_2\|_2^2+V_{14}^TV_{13}\leq&(2c_1^{-2}\|V_{14}\|_2^2+2c_1^{-2}\|w_2\|_2^2)^{1/2}|u^T\bepsilon^*|\nonumber\\
&+\sqrt{n}\|V_{13}\|_2\|\bSigma_{11}^{-1/2}\bSigma_{12}\bbeta_{A_2}\|_2+\lambda_{n1}\|\bbeta_{A_5}\|_1
\end{align}
By (\ref{numberofselectedbound}), we have
\begin{align*}
\|V_{13}\|_2^2&=\frac{1}{n}\|\bSigma_{11}^{-1/2}Q_{A_31}^TS_{A_3}\|_2^2\\
&\leq c_1^{-1}\frac{\|Q_{A_31}S_{A_3}\|_2^2}{n}\\
&=c_1^{-1}\sum_{k\in A_3}\|\lambda_{n1}s_k\|_2^2\\
&\leq c_1^{-1}n\lambda_{n1}^2N(A_3)
\end{align*}
By (\ref{boundv14w2v1413}), we have
\begin{align}
&\|V_{14}\|_2^2+c_1\|w_2\|_2^2\nonumber\\
\leq&c_1^{-1}(2\|V_{14}\|_2^2+2\|w_2\|_2^2)^{1/2}|u^T\bepsilon^*|+\sqrt{c_1^{-1}n\lambda_{n1}^2N(A_3)}\|V_{14}\|_2\nonumber\\
&+\sqrt{c_1^{-1}n\lambda_{n1}^2N(A_3)}\|\bSigma_{11}^{-1/2}\bSigma_{12}\bbeta_{A_2}\|_2+\lambda_{n1}\|\bbeta_{A_5}\|_1
\end{align}
Define
$$B_1=\sqrt{c_1n\lambda_{n1}^2s_n}\ \text{and}\ B_2=\sqrt{c_1^{-1}n\lambda_{n1}^2s_n}$$
consider the event
$$\mathcal{E}=\left\{|u^T\bepsilon^*|^2\leq\frac{(|A_1|\vee m_n)c_1^2n\lambda_{n1}^2}{4m_n}=(|A_1|\vee m_n)\frac{c_1^3B_1^2}{4s_nm_n}\right\}$$
later we will show that this event holds with probability tending to 1. On the event $\mathcal{E}$, by (\ref{numberofselectedbound}), we have
$$\|V_{14}\|_2^2\geq\frac{q_1}{s_n}B_1^2-B_1^2$$
then
$$|u^T\bepsilon^*|^2\leq\frac{c_1^3q_1m_nB_1^2}{4s_nm_n}\leq\frac{c_1^3}{4}(\|V_{14}\|_2^2+B_1^2)$$
and we have
\begin{align*}
c_1^{-1}(2\|V_{14}\|_2^2+2\|w_2\|_2^2)^{1/2}|u^T\bepsilon^*|&\leq c_1^{-3}|u^T\bepsilon^*|^2+\frac{c_1^3}{4}c_1^{-2}(2\|V_{14}\|_2^2+2\|w_2\|_2^2)\\
&\leq \frac{1}{4}(\|V_{14}\|_2^2+B_1^2)+\frac{c_1^3}{4}c_1^{-2}(2\|V_{14}\|_2^2+2\|w_2\|_2^2)\\
&\leq \frac{3}{4}\|V_{14}\|_2^2+\frac{1}{4}B_1^2+\frac{c_1}{2}\|w_2\|_2^2
\end{align*}
Then we have
\begin{align*}
\|V_{14}\|_2^2+c_1\|w_2\|_2^2\leq&\frac{3}{4}\|V_{14}\|_2^2+\frac{1}{4}B_1^2+\frac{c_1}{2}\|w_2\|_2^2+\sqrt{c_1^{-1}n\lambda_{n1}^2N(A_3)}\|V_{14}\|_2\nonumber\\
&+\sqrt{c_1^{-1}n\lambda_{n1}^2N(A_3)}\|\bSigma_{11}^{-1/2}\bSigma_{12}\bbeta_{A_2}\|_2+\lambda_{n1}\|\bbeta_{A_5}\|_1
\end{align*}
i.e.
\begin{align*}
\|V_{14}\|_2^2+2c_1\|w_2\|_2^2\leq B_1^2+4\sqrt{c_1^{-1}n\lambda_{n1}^2N(A_3)}(\|V_{14}\|_2+\|\bSigma_{11}^{-1/2}\bSigma_{12}\bbeta_{A_2}\|_2)+\lambda_{n1}\|\bbeta_{A_5}\|_1
\end{align*}
Consider the set $A_1$ that contains all $\bbeta_k\neq 0$. We have $q_1\geq s_n$ and
\begin{equation}
\label{largeA1}
\left\{k:\|\hat{\bbeta}_k\|_2>0\ \text{or}\ k\notin T^c\right\}\subset A_1\subset\left\{k:\frac{1}{n}\Phi_k^T(\boldy-\hat{\bmu}_{\boldy}^*)=\frac{\lambda_{n1}\hat{\bbeta}_k}{\|\hat{\bbeta}_k\|_2}\ \text{or}\ k\notin T^c\right\}
\end{equation}
Then we have $A_5=\emptyset$, $N(A_3)=s_n\leq q_1$ and $\bbeta_{A_2}=\boldzero$. Then we have
$$\|V_{14}\|_2^2\leq B_1^2+4\sqrt{c_1^{-1}n\lambda_{n1}^2s_n}\|V_{14}\|_2=B_1^2+4B_2\|V_{14}\|_2$$
Use the truth that $x^2\leq c+2bx$ implies $x^2\leq 2c+4b^2$, we have
$$\|V_{14}\|_2^2\leq 2B_1^2+16B_2^2$$
Then we have from (\ref{numberofselectedbound}) that
\begin{equation*}
(q_1-s_n)^+\leq\frac{\|V_{14}\|_2^2}{c_1n\lambda_{n1}^2}\leq \frac{2B_1^2+16B_2^2}{c_1n\lambda_{n1}^2}=c_5s_n
\end{equation*}
where $c_5=(2c_1^2+16)/c_1^2$, i.e.
\begin{equation}
\label{boundnumberofvariables}
(q_1-s_n)^++s_n\leq (c_5+1)s_n
\end{equation}
We note that the constant $c_5$ only depends on $c_1$ and (\ref{largeA1}) simply requires larger $A_1$, (\ref{boundnumberofvariables}) holds for all $A_1$ satisfying (\ref{generalA1}). Note that (\ref{boundnumberofvariables}) holds if
\begin{equation}
\label{highprobabilitynumberofselection}
q_1\leq N(A_1\cup A_5)\leq\frac{n}{m_n}\ \text{and}\ |u^T\bepsilon^*|^2\leq\frac{(|A_1|\vee m_n)c_1^2n\lambda_{n1}^2}{4m_n}
\end{equation}
So it remains to show that (\ref{highprobabilitynumberofselection}) holds with probability tending to 1. Define
\begin{align}
x_m^*=&\max_{|A|=m}\max_{\|U_{A_k}\|_2=1,k=1,...,m}\left|\bepsilon^{*T}\right.\nonumber\\
&\left.\frac{\Phi_A(\Phi_A^T\bSigma_A\Phi_A)^{-1}\bar{S}_A-\bSigma_A^{-1/2}(\boldI-\bSigma_A^{1/2}\Phi_A(\Phi_A^T\bSigma_A\Phi_A)^{-1}\Phi_A^T\bSigma_A^{1/2})\bSigma_A^{1/2}\Phi\bbeta}{\|\Phi_A(\Phi_A^T\bSigma_A\Phi_A)^{-1}\bar{S}_A-\bSigma_A^{-1/2}(\boldI-\bSigma_A^{1/2}\Phi_A(\Phi_A^T\bSigma_A\Phi_A)^{-1}\Phi_A^T\bSigma_A^{1/2})\bSigma_A^{1/2}\Phi\bbeta\|_2}\right|
\end{align}
for $|A|=q_1=m\geq 0$, $\bar{S}_A=(\bar{S}_{A_1}^T,...,\bar{S}_{A_m}^T)^T$ where $\bar{S}_{A_k}=\lambda_{n1}U_{A_k}$, $\|U_{A_k}\|_2=1$ and $\bSigma_A$ is the variance matrix evaluated at some $\theta$ corresponding to the remainder of the Taylor expansion when the subset $A$ is considered. To simplify the notations, let $Q_A=\lambda_{n1}\Phi_A(\Phi_A^T\bSigma_A\Phi_A)^{-1}$ and $P_A=\bSigma_A^{1/2}\Phi_A(\Phi_A^T\bSigma_A\Phi_A)^{-1}\Phi_A^T\bSigma_A^{1/2}$, then we have
\begin{equation}
x_m^*=\max_{|A|=m}\max_{\|U_{A_k}\|_2=1,k=1,...,m}\left|\bepsilon^{*T}\frac{Q_AU_A-\bSigma_A^{-1/2}(\boldI-P_A)\bSigma_A^{1/2}\Phi\bbeta}{\|Q_AU_A-\bSigma_A^{-1/2}(\boldI-P_A)\bSigma_A^{1/2}\Phi\bbeta\|_2}\right|
\end{equation}

Define
$$\Omega^*_{m_0}=\{(U,\bepsilon^*):x_m^*\leq C\sqrt{(|A|\vee 1)m_n\log(p_nm_n)},\forall m=|A|\geq m_0\}$$
and
$$\Omega_{m_0}=\{(U,\bepsilon):x_m^{**}\leq C\sqrt{(|A|\vee 1)m_n\log(p_nm_n)},\forall m=|A|\geq m_0\}$$
for a large enough generic constant $C$, where
$$x_m^{**}=\max_{|A|=m}\max_{\|U_{A_k}\|_2=1,k=1,...,m}\left|\bepsilon^{T}\frac{Q_AU_A-\bSigma_A^{-1/2}(\boldI-P_A)\bSigma_A^{1/2}\Phi\bbeta}{\|Q_AU_A-\bSigma_A^{-1/2}(\boldI-P_A)\bSigma_A^{1/2}\Phi\bbeta\|_2}\right|$$
By triangle inequality and Cauchy-Schwarz inequality, we have
\begin{align*}
&\left|\bepsilon^{*T}\frac{Q_AU_A-\bSigma_A^{-1/2}(\boldI-P_A)\bSigma_A^{1/2}\Phi\bbeta}{\|Q_AU_A-\bSigma_A^{-1/2}(\boldI-P_A)\bSigma_A^{1/2}\Phi\bbeta\|_2}\right|\\
\leq &\left|\bepsilon^{T}\frac{Q_AU_A-\bSigma_A^{-1/2}(\boldI-P_A)\bSigma_A^{1/2}\Phi\bbeta}{\|Q_AU_A-\bSigma_A^{-1/2}(\boldI-P_A)\bSigma_A^{1/2}\Phi\bbeta\|_2}\right|+\|\btheta_n\|_2\\
\leq&\left|\bepsilon^{T}\frac{Q_AU_A-\bSigma_A^{-1/2}(\boldI-P_A)\bSigma_A^{1/2}\Phi\bbeta}{\|Q_AU_A-\bSigma_A^{-1/2}(\boldI-P_A)\bSigma_A^{1/2}\Phi\bbeta\|_2}\right|+Cn^{1/2}s_nm_n^{-d}\\
\leq&\left|\bepsilon^{T}\frac{Q_AU_A-\bSigma_A^{-1/2}(\boldI-P_A)\bSigma_A^{1/2}\Phi\bbeta}{\|Q_AU_A-\bSigma_A^{-1/2}(\boldI-P_A)\bSigma_A^{1/2}\Phi\bbeta\|_2}\right|+C\sqrt{(|A|\vee 1)m_n\log(p_nm_n)}
\end{align*}
Then we have
$$(U,\bepsilon)\in\Omega_{m_0}\ \Rightarrow\ (U,\bepsilon^*)\in\Omega_{m_0}^*\ \Rightarrow\ |u^T\bepsilon^*|^2\leq |x_m^*|^2\leq\frac{(|A_1|\vee m_n)c_1^2n\lambda_{n1}^2}{4m_n}\ \text{for\ }q_1\geq m_0\geq 0$$
Since $\epsilon_i$'s are sub-Gaussian random variables by assumption 2, we have
\begin{align*}
&1-\prob\left((U,\bepsilon)\in\Omega_q\right)\\
=&\prob\left(x_m^{**}>C\sqrt{(m\vee 1)m_n\log(p_nm_n)},\forall m=|A|\geq m_0\right)\\
\leq&\sum_{m=0}^{\infty}\prob\left(x_m^{**}>C\sqrt{(m\vee 1)m_n\log(p_nm_n)}\right)\\
\leq&\sum_{m=0}^{\infty}{p_n\choose m}\prob\left(\left|\bepsilon^{T}\frac{Q_AU_A-\bSigma_A^{-1/2}(\boldI-P_A)\bSigma_A^{1/2}\Phi\bbeta}{\|Q_AU_A-\bSigma_A^{-1/2}(\boldI-P_A)\bSigma_A^{1/2}\Phi\bbeta\|_2}\right|>C\sqrt{(m\vee 1)m_n\log(p_nm_n)}\right)\\
\leq&2\sum_{m=0}^{\infty}{p_n\choose m}\exp\left(-C(m\vee 1)m_n\log(p_nm_n)\right)\\
=&2(p_nm_n)^{-Cm_n}+2\sum_{m=1}^{\infty}{p_n\choose m}(p_nm_n)^{-Cmm_n}\\
\leq&2(p_nm_n)^{-Cm_n}+2\sum_{m=1}^{\infty}\frac{1}{m!}\left(\frac{p_n}{(p_nm_n)^{Cm_n}}\right)^m\\
=&2(p_nm_n)^{-Cm_n}+2\exp\left(\frac{p_n}{(p_nm_n)^{Cm_n}}\right)-2 \rightarrow 0\ \text{as}\ n\rightarrow\infty
\end{align*}
Therefore, the proof of part (i) is complete.

Then we prove part (ii). Consider the bounded response case. For a sequence $N_n$ such that $\|\hat{\bbeta}-\bbeta^0\|_2\leq N_n$, define $t=N_n/(N_n+\|\hat{\bbeta}-{\bbeta^0}\|_2)$, then consider the convex combination $\bbeta^*=t\hat{\bbeta}+(1-t){\bbeta^0}$. We have $\bbeta^*-{\bbeta^0}=t(\hat{\bbeta}-{\bbeta^0})$, which implies
\begin{equation}
\label{an}
\|\bbeta^*-{\bbeta^0}\|_2=t\|\hat{\bbeta}-{\bbeta^0}\|_2=\frac{N_n\|\hat{\bbeta}-{\bbeta^0}\|_2}{N_n+\|\hat{\bbeta}-{\bbeta^0}\|_2}\leq N_n
\end{equation}
Recall the log likelihood function
$$l_n(\bbeta)=\frac{1}{n}\sum_{i=1}^n\left[y_i\left(\alpha+{\bbeta}^T\Phi_i\right)-b\left(\alpha+{\bbeta}^T\Phi_i\right)\right]$$
\begin{align}
\label{grplassothmtaylor}
l_n(\bbeta^*)&=l_n({\bbeta^0})+\frac{1}{n}\sum_{i=1}^n\left[y_i\Phi_i-\mu_{y_i}^*\Phi_i\right]^T(\bbeta^*-{\bbeta^0})\nonumber\\
&-\frac{1}{2n}\sum_{i=1}^n(\bbeta^*-{\bbeta^0})^T\Phi_i^Tb''(\alpha+{\bbeta^{**}}^T\Phi_i)\Phi_i(\bbeta^*-{\bbeta^0})\nonumber\\
&=l_n({\bbeta^0})+\frac{(\bbeta^*-{\bbeta^0})^T\Phi^T(\boldy-\bmu_y^*)}{n}-\frac{1}{2n}(\bbeta^*-{\bbeta^0})^T\Phi^T\bSigma(\bbeta^{**})\Phi(\bbeta^*-{\bbeta^0})
\end{align}
where ${\bbeta}^{**}$ lines on the line joining $\bbeta^*$ and ${\bbeta^0}$, and
$$\bSigma(\bbeta^{**})=\text{diag}\left(b''(\alpha+{\bbeta^{**}}^T\Phi_1),...b''(\alpha+{\bbeta^{**}}^T\Phi_n)\right)$$
is the variance matrix of response when the coefficients take value on ${\bbeta}^{**}$. On the other hand, by convexity of the log likelihood function,
$$l_n({\bbeta^*})=l_n(t\hat{\bbeta}+(1-t){\bbeta^0})\geq tl_n(\hat{\bbeta})+(1-t)l_n(\bbeta^0)$$
by norm inequality, we have
$$\sum_{j=1}^{p_n}\|\bbeta^*_j\|_2=\sum_{j=1}^{p_n}\|t\hat{\bbeta}_j+(1-t){\bbeta^0_j}\|_2\leq \sum_{j=1}^{p_n}(t\|\hat{\bbeta}_j\|_2+(1-t)\|\bbeta^0_j\|_2)$$
joining the two inequalities above and by the definition of $\hat{\bbeta}$ gives
$$l_n({\bbeta^*})-\lambda_{n1}\sum_{j=1}^{p_n}\|\bbeta^*_j\|_2\geq tl_n(\hat{\bbeta})+(1-t)l_n(\bbeta^0_j)-\lambda_{n1}\sum_{j=1}^{p_n}(t\|\hat{\bbeta}_j\|_2+(1-t)\|{\bbeta^0_j}\|_2)\geq l_n(\bbeta^0)-\lambda_{n1}\sum_{j=1}^{p_n}\|\bbeta^0_j\|_2$$
which implies
\begin{equation}
\label{lhlb}
l_n({\bbeta^*})-l_n({\bbeta^0})\geq\lambda_{n1}\sum_{j=1}^{p_n}\|\bbeta^*_{j}\|_2-\lambda_{n1}\sum_{j=1}^{p_n}\|\bbeta^0_{j}\|_2
\end{equation}
By (\ref{grplassothmtaylor}) and (\ref{lhlb}) together we have
\begin{equation*}
\lambda_{n1}\sum_{j=1}^{p_n}\left(\|\bbeta^*_{j}\|_2-\|\bbeta^0_{j}\|_2\right)\leq\frac{(\bbeta^*-{\bbeta^0})^T\Phi^T(\boldy-\bmu_{\boldy}^*)}{n}-\frac{1}{2n}(\bbeta^*-{\bbeta^0})^T\Phi^T\bSigma(\bbeta^{**})\Phi(\bbeta^*-{\bbeta^0})
\end{equation*}
and move one term to the left hand side, we have
\begin{align}
\label{beforerhsbd}
&\frac{1}{2n}(\bbeta^*-{\bbeta^0})^T\Phi^T\bSigma(\bbeta^{**})\Phi(\bbeta^*-{\bbeta^0})\nonumber\\
\leq&\frac{\left(\bbeta^*-{\bbeta^0}\right)^T\Phi^T\left(\boldy-\bmu_y^*\right)}{n}+\lambda_{n1}\sum_{j=1}^{p_n}\left(\|\bbeta^0_{j}\|_2-\|\bbeta^*_{j}\|_2\right)\nonumber\\
=&\frac{\left(\bbeta^*-{\bbeta^0}\right)^T\Phi^T\left(\boldy-\bmu_y\right)}{n}+\frac{\left(\bbeta^*-{\bbeta^0}\right)^T\Phi^T\left(\bmu_y^*-\bmu_y\right)}{n}+\lambda_{n1}\sum_{j=1}^{p_n}\left(\|\bbeta^0_{j}\|_2-\|\bbeta^*_{j}\|_2\right)\nonumber\\
\end{align}
We have for the second term
\begin{align}
\label{rhsbd}
&\frac{\left(\bbeta^*-{\bbeta^0}\right)^T\Phi^T\left(\bmu_y^*-\bmu_y\right)}{n}\nonumber\\
=&\frac{\left(\bbeta^*-{\bbeta^0}\right)^T\Phi^T\bSigma(\bbeta^{**})^{1/2}\bSigma(\bbeta^{**})^{-1/2}\left(\bmu_y^*-\bmu_y\right)}{n}\nonumber\\
\leq&\frac{\|\bSigma(\bbeta^{**})^{1/2}\Phi\left(\bbeta^*-{\bbeta^0}\right)\|_2\|\bSigma(\bbeta^{**})^{-1/2}\left(\bmu_y^*-\bmu_y\right)\|_2}{n}\nonumber\\
\leq&\frac{(\bbeta^*-{\bbeta^0})^T\Phi^T\bSigma(\bbeta^{**})\Phi(\bbeta^*-{\bbeta^0})}{4n}+\frac{\|\bSigma(\bbeta^{**})^{-1/2}\left(\bmu_y^*-\bmu_y\right)\|_2^2}{n}\nonumber\\
\leq&\frac{(\bbeta^*-{\bbeta^0})^T\Phi^T\bSigma(\bbeta^{**})\Phi(\bbeta^*-{\bbeta^0})}{4n}+c_1d_n
\end{align}
where $d_n=O(s_n^2m_n^{-2d})$, the first inequality follows from Cauchy-Schwarz inequality, the second inequality follows from the identity $uv\leq u^2/4+v^2$, and the third inequality follow from assumption 3 and (\ref{splineapproxerr}).
Then joining (\ref{beforerhsbd}) and (\ref{rhsbd}), we have
\begin{equation}
\label{afterrhsbd}
\frac{(\bbeta^*-{\bbeta^0})^T\Phi^T\bSigma(\bbeta^{**})\Phi(\bbeta^*-{\bbeta^0})}{4n}\leq \frac{\left(\bbeta^*-{\bbeta^0}\right)^T\Phi^T\left(\boldy-\bmu_y\right)}{n}+\lambda_{n1}\sum_{j=1}^{p_n}\left(\|\bbeta^0_{j}\|_2-\|\bbeta^*_{j}\|_2\right)+c_1d_n
\end{equation}

For the first term on the right hand side of (\ref{afterrhsbd}), we have
\begin{align}
\label{rhsbd2}
&\frac{\left(\bbeta^*-{\bbeta^0}\right)^T\Phi^T\left(\boldy-\bmu_y\right)}{n}\nonumber\\
=&\frac{\left(\bbeta^*-{\bbeta^0}\right)^T\Phi^T\bSigma(\bbeta^{**})^{1/2}\bSigma(\bbeta^{**})^{-1/2}\left(\boldy-\bmu_y\right)}{n}\nonumber\\
\leq&\frac{(\bbeta^*-{\bbeta^0})^T\Phi^T\bSigma(\bbeta^{**})\Phi(\bbeta^*-{\bbeta^0})}{8n}+\frac{2\|\bSigma(\bbeta^{**})^{-1/2}\left(\boldy-\bmu_y\right)\|_2^2}{n}
\end{align}
where the inequality is by the identity $a^Tb\leq \|a\|_2^2/8+2\|b\|_2^2$. Joining (\ref{afterrhsbd2}) and (\ref{rhsbd2}), we have
\begin{equation}
\frac{(\bbeta^*-{\bbeta^0})^T\Phi^T\bSigma(\bbeta^{**})\Phi(\bbeta^*-{\bbeta^0})}{8n}\leq\frac{2\|\bSigma(\bbeta^{**})^{-1/2}\left(\boldy-\bmu_y\right)\|_2^2}{n}+\lambda_{n1}\sum_{j=1}^{p_n}\left(\|\bbeta^0_{j}\|_2-\|\bbeta^*_{j}\|_2\right)+c_1d_n
\end{equation}
By remark \ref{REremark}, we have
\begin{equation}
\label{afterrhsbd2}
\frac{\gamma_0c_1\gamma_2^{2s_n}m_n^{-1}}{8}\|\bbeta^*-\bbeta^0\|_2^2\leq\frac{2\|\bSigma(\bbeta^{**})^{-1/2}\left(\boldy-\bmu_y\right)\|_2^2}{n}+\lambda_{n1}\sum_{j=1}^{p_n}\left(\|\bbeta^0_{j}\|_2-\|\bbeta^*_{j}\|_2\right)+c_1d_n
\end{equation}
Observe that
\begin{align*}
\|\bSigma(\bbeta^{**})^{-1/2}\left(\boldy-\bmu_y\right)\|_2^2&\leq c_1^{-1}\|\boldy-\bmu_y\|_2^2\\
&\leq \frac{c_1^{-1}m_n}{\gamma_0\gamma_2^{2s_n}}\|\Phi^T(\boldy-\bmu_y)\|_2^2
\end{align*}
Then by lemma \ref{highprob1}, we have
\begin{equation}
\label{afterrhsbd3}
\frac{\gamma_0c_1\gamma_2^{2s_n}m_n^{-1}}{8}\|\bbeta^*_{\{T\cup\hat{T}\}}-\bbeta^0_{\{T\cup\hat{T}\}}\|_2^2\leq O_P\left(s_nm_n\frac{\log(p_nm_n)}{n\gamma_2^{2s_n}}\right)+\lambda_{n1}\sum_{j=1}^{p_n}\left(\|\bbeta^0_{j}\|_2-\|\bbeta^*_{j}\|_2\right)+O(s_n^2m_n^{-2d})
\end{equation}
Observe that
\begin{align}
\label{rhsbd3}
&\lambda_{n1}\sum_{j=1}^{p_n}\left(\|\bbeta^0_{j}\|_2-\|\bbeta^*_{j}\|_2\right)\nonumber\\
\leq&\lambda_{n1}\sum_{j\in T\cup\hat{T}}\left\|\bbeta^0_j-\bbeta^*_j\right\|_2\nonumber\\
\leq&\lambda_{n1}\sqrt{s_n}\left\|\bbeta^*_{\{T\cup\hat{T}\}}-\bbeta^0_{\{T\cup\hat{T}\}}\right\|_2\nonumber\\
\leq&\frac{\gamma_0c_1\gamma_2^{2s_n}m_n^{-1}}{16}\left\|\bbeta^*_{\{T\cup\hat{T}\}}-\bbeta^0_{\{T\cup\hat{T}\}}\right\|_2^2+\frac{4\lambda_{n1}^2s_n}{\gamma_0c_1\gamma_2^{2s_n}m_n^{-1}}
\end{align}
where the first two inequalities are by norm inequality, and the third inequality is by the identity $a^Tb\leq \|a\|_2^2+\|b\|_2^2/4$. Joining (\ref{afterrhsbd3}) and (\ref{rhsbd3}), we have
\begin{equation}
\label{afterrhsbd4}
\left\|\bbeta^*-\bbeta^0\right\|_2^2=O_P\left(s_n\gamma_2^{-2s_n}\frac{m_n^2\log(p_nm_n)}{n}\right)+O(\lambda_{n1}^2m_n^2s_n\gamma_2^{-2s_n})+O(s_n^2m_n^{1-2d}\gamma_2^{-2s_n})
\end{equation}
For some $N_n$ such that
$$\|{\bbeta}^*-{\bbeta^0}\|_2\leq N_n/2$$
By definition of ${\bbeta}^{*}$, we have
$$\|{\bbeta}^{*}-{\bbeta^0}\|_2=\frac{N_n}{N_n+\|\hat{\bbeta}-{\bbeta^0}\|_2}\|\hat{\bbeta}-{\bbeta^0}\|_2\leq\frac{N_n}{2}$$
The inequality above implies
$$\|\hat{\bbeta}-{\bbeta^0}\|_2\leq N_n$$
Therefore,
$$\|\hat{\bbeta}-{\bbeta^0}\|_2^2=O_P\left(s_n\gamma_2^{-2s_n}\frac{m_n^2\log(p_nm_n)}{n}\right)+O(\lambda_{n1}^2m_n^2s_n\gamma_2^{-2s_n})+O(s_n^2m_n^{1-2d}\gamma_2^{-2s_n})$$
In the unbounded response case, the only difference that we have to make is in (\ref{rhsbd2}), we have
\begin{align}
\label{rhsbd4}
&\frac{\left(\bbeta^*-{\bbeta^0}\right)^T\Phi^T\left(\boldy-\bmu_y\right)}{n}\nonumber\\
\leq&\frac{\gamma_0c_1}{8}\|\bbeta^*_{\{T\cup\hat{T}\}}-\bbeta^0_{\{T\cup\hat{T}\}}\|_2^2+O_P\left(s_nm_na_n\frac{\log(p_nm_n)}{n}\right)
\end{align}
where the convergence rate is by lemma \ref{highprob2}. Then with the choice of $\lambda_{n1}$ for this case, we have
$$\|\hat{\bbeta}-{\bbeta^0}\|_2^2=O_P\left(s_n\gamma_2^{-2s_n}\frac{m_n^2\log(p_nm_n)}{n}\right)+O(\lambda_{n1}^2m_n^2s_n\gamma_2^{-2s_n})+O(s_n^2m_n^{1-2d}\gamma_2^{-2s_n})$$
Part (iii) is a direct result of part (ii). By assumption 4, we have $\|f_j\|_2\geq c_{f,n}>0$, and we have
$$\|f_{nj}\|_2\geq\|f_j\|_2-\|f_j-f_{nj}\|_2\geq c_{f,n}-O(m_n^{-d})\geq\frac{1}{2}c_{f,n}$$
for large $n$. By the properties of spline in \cite{de2001practical}, see for example \cite{stone1986dimensionality} and \cite{huang2010variable}, there exist positive constants $c_1$ and $c_2$ such that
$$c_1m_n^{-1}\|\bbeta^0_{j}\|_2^2\leq\|f_{nj}\|_2\leq c_2m_n^{-1}\|\bbeta^0_{j}\|_2$$
Then we have $\|\bbeta^0_{j}\|_2^2\geq c_2^{-1}m_n\|f_{nj}\|_2^2\geq 0.25c_2^{-1}m_nc_{f,n}^2$. Suppose there is a $j\in T$ such that $\|\hat{\bbeta}_j\|_2=0$, then we have
$$\|\bbeta^0_{j}\|_2\geq 0.25c_2^{-1}m_nc_{f,n}^2$$
which is a contradiction to the result in (ii) and the theorem assumption. Therefore, part (iii) follows.
\end{proof}

\noindent
\textbf{Proof of theorem \ref{selection1}}
\begin{proof}
We start with part (i). To prove part (i), it's equivalent to prove that the selection is done as it is performed right on the active set, and none of the nonzero components are dropped with probability tending to 1. Let
$$\hat{\bbeta}_{NZ}=\argmin_{\bbeta\in\R^{p_nm_n}:\bbeta_{T^c}=0}L_a(\bbeta;\lambda_{n2})$$
be the adaptive group lasso estimator restricted to the true nonzero components. First we show that with probability converging to 1, $\hat{\bbeta}_{NZ}$ is the solution to minimizing (\ref{adagrouplasso}), i.e., with probability converging to 1, the minimiser of (\ref{adagrouplasso}) is $\hat{\bbeta}_{NZ}$. Note that the adaptive group lasso is a convex optimization problem with affine constraints, therefore the KKT conditions are necessary and sufficient. The KKT conditions for a vector $\bbeta\in\R^{p_nm_n}$ to be the solution of (\ref{adagrouplasso}) is
\begin{equation}
\label{adaKKT}
\begin{dcases}
\frac{1}{n}\Phi_j^T(\boldy-\bmu^*)=\lambda_{n2}w_{nj}\frac{\bbeta_{j}}{\|\bbeta_{j}\|_2},\ \text{if}\ \|\bbeta_{j}\|_2>0\\
\|\frac{1}{n}\Phi_j^T(\boldy-\bmu^*)\|_2\leq\lambda_{n2}w_{nj},\ \text{if}\ \|\bbeta_{j}\|_2=0
\end{dcases}
\end{equation}
where $\bmu^*=b'(\Phi\bbeta)$. It is sufficient to show that
$$\prob\left(\hat{\bbeta}_{NZ}\ \text{satisfies}\ (\ref{adaKKT})\right)\rightarrow 1$$
Note that for any $j\in T$, we have the KKT conditions for $\hat{\bbeta}_{NZ}$ that
\begin{equation}
\begin{dcases}
\frac{1}{n}\Phi_j^T(\boldy-\hat{\bmu}_{NZ}^*)=\lambda_{n2}w_{nj}\frac{\hat{\bbeta}_{NZj}}{\|\hat{\bbeta}_{NZj}\|_2},\ \text{if}\ \|\bbeta_{j}\|_2>0,\ j\in T\\
\|\frac{1}{n}\Phi_j^T(\boldy-\hat{\bmu}_{NZ}^*)\|_2\leq\lambda_{n2}w_{nj},\ \text{if}\ \|\bbeta_{j}\|_2=0,\ j\in T
\end{dcases}
\end{equation}
which are the equality condition in (\ref{adaKKT}) and part of the inequality condition in (\ref{adaKKT}). Therefore, it suffices to show that
\begin{equation}
\prob\left(\|\frac{1}{n}\Phi_j^T(\boldy-\hat{\bmu}_{NZ}^*)\|_2\leq\lambda_{n2}w_{nj},\ \forall\ j\notin T\right)\rightarrow 1
\end{equation}
This is equivalent to show that
\begin{equation}
\prob\left(\|\frac{1}{n}\Phi_j^T(\boldy-\hat{\bmu}_{NZ}^*)\|_2>\lambda_{n2}w_{nj},\ \exists\ j\notin T\right)\rightarrow 0
\end{equation}
Use Taylor expansion on $\frac{1}{n}\Phi_j^T(\boldy-\hat{\bmu}_{NZ}^*)$, we have
$$\frac{1}{n}\Phi_j^T(\boldy-\hat{\bmu}_{NZ}^*)=\frac{1}{n}\Phi_j(\boldy-\bmu_y)+\frac{1}{n}\Phi_j^T(\bmu_y-b'(\Phi\bbeta^0))+\frac{1}{n}\Phi_j^T\bSigma\Phi(\hat{\bbeta}_{NZ}-\bbeta^0)$$
where $\bSigma$ is the variance matrix evaluated at some $\bbeta^*$ located on the line segment joining $\bbeta^0$ and $\hat{\bbeta}_{NZ}$. Then we have
\begin{align*}
&\prob\left(\|\frac{1}{n}\Phi_j^T(\boldy-\hat{\bmu}_{NZ}^*)\|_2>\lambda_{n2}w_{nj},\ \exists\ j\notin T\right)\\
\leq&\prob\left(\|\frac{1}{n}\Phi_j(\boldy-\bmu_y)\|_2>\frac{\lambda_{n2}w_{nj}}{3},\ \exists\ j\notin T\right)+\prob\left(\|\frac{1}{n}\Phi_j^T(\bmu_y-b'(\Phi\bbeta^0))\|_2>\frac{\lambda_{n2}w_{nj}}{3},\ \exists\ j\notin T\right)\\
&+\prob\left(\|\frac{1}{n}\Phi_j^T\bSigma\Phi(\hat{\bbeta}_{NZ}-\bbeta^0)\|_2>\frac{\lambda_{n2}w_{nj}}{3},\ \exists\ j\notin T\right)\\
\equiv& P_1+P_2+P_3
\end{align*}
Now let's consider $P_1$. By assumption 3, the errors $y_i-\mu_{y_i}$'s are sub-Gaussian. For bounded responses, we have by lemma \ref{highprob1} and assumption 6 that
\begin{align*}
P_1&=\prob\left(\|\frac{1}{n}\Phi_j^T(\boldy-\hat{\bmu}_{NZ}^*)\|_2>\lambda_{n2}w_{nj},\ \exists\ j\notin T\right)\\
&\leq\prob\left(\|\frac{1}{n}\Phi_j^T(\boldy-\hat{\bmu}_{NZ}^*)\|_2>C\lambda_{n2}r_n,\ \exists\ j\notin T\right)+o(1)\\
&=\prob\left(\max_{j\notin T}\|\frac{1}{n}\Phi_j^T(\boldy-\hat{\bmu}_{NZ}^*)\|_2>C\lambda_{n2}r_n\right)+o(1)\\
&\leq\prob\left(\max_{j\notin T}\|\frac{1}{n}\Phi_j^T(\boldy-\hat{\bmu}_{NZ}^*)\|_2>C\lambda_{n2}r_n|\max_{j\notin T}\|\frac{1}{n}\Phi_j^T(\boldy-\hat{\bmu}_{NZ}^*)\|_2\leq Cn^{-1/2}\sqrt{\log(s_n^*m_n)}\right)+o(1)\\
&\rightarrow 0\ as\ n\rightarrow\infty
\end{align*}
By lemma \ref{highprob2}, we have
\begin{equation}
E\left(\max_{j\notin T,k=1,...,m_n}\left\|\frac{1}{n}\Phi_{jk}^T(\boldy-\bmu_y)\right\|_2\right)\leq c_6n^{-1/2}\sqrt{\log(s_n^*m_n)}
\end{equation}
for some constant $c_6$. Observe that by assumption 5, we have $w_{nj}=O_P(r_n)\leq Cr_n$ for some general constant $C$. Then we have by Markov's inequality and assumption 6 that
\begin{align*}
P_1&=\prob\left(\|\frac{1}{n}\Phi_j^T(\boldy-\hat{\bmu}_{NZ}^*)\|_2>\lambda_{n2}w_{nj},\ \exists\ j\notin T\right)\\
&\leq\prob\left(\|\frac{1}{n}\Phi_j^T(\boldy-\hat{\bmu}_{NZ}^*)\|_2>C\lambda_{n2}r_n,\ \exists\ j\notin T\right)+o(1)\\
&=\prob\left(\max_{j\notin T}\|\frac{1}{n}\Phi_j^T(\boldy-\hat{\bmu}_{NZ}^*)\|_2>C\lambda_{n2}r_n\right)+o(1)\\
&\leq\frac{E\left(\max_{j\notin T,k=1,...,m_n}\left\|\frac{1}{n}\Phi_{jk}^T(\boldy-\bmu_y)\right\|_2\right)}{C\lambda_{n2}r_n}+o(1)\\
&\leq\frac{c_6\sqrt{\log(s_n^*m_n)}}{Cn^{1/2}\lambda_{n2}r_n}+o(1)\rightarrow 0\ as\ n\rightarrow\infty
\end{align*}
Then we consider $P_2$. We have shown that
$$\frac{1}{n}\|\bmu_y-\bmu_y^*\|_2^2=O(s_n^2m_n^{-2d})$$
This implies that
$$\frac{1}{\sqrt{n}}\|\bmu_y-\bmu_y^*\|_2=O(s_nm_n^{-d})$$
Then by assumption 1,
\begin{align*}
&\max_{j\notin T}\left\|\frac{1}{n}\Phi_j(\bmu_y-\bmu_y^*)\right\|_2\\
\leq&Cm_n^{-1/2}\frac{1}{\sqrt{n}}\left\|\bmu_y-\bmu_y^*\right\|_2\\
=&O(s_nm_n^{-d-1/2})
\end{align*}
By assumption 6, we have $P_2\rightarrow 0$ as $n\rightarrow\infty$. Next, we look at $P_3$. By the definition of $\hat{\bbeta}_{NZ}$, we have by norm inequality
$$\frac{1}{n}\Phi_j^T\bSigma\Phi(\hat{\bbeta}_{NZ}-\bbeta^0)=\frac{1}{n}\Phi_j^T\bSigma\Phi_T(\hat{\bbeta}_{NZT}-\bbeta^0_T)$$
The MLE on the true nonzero set has a rate of convergence $\sqrt{s_nm_n/n}$. The penalised solution has been proved to be close to the MLE asymptotically (\cite{zhang2008sparsity}; \cite{fan2001variable}; \cite{lv2009unified}). Knowing the true nonzero set, the rate of convergence of $\hat{\bbeta}_{NZ}$ is $\sqrt{s_nm_n/n}$. Then we have
\begin{align*}
P_3=&\prob\left(\left\|\frac{1}{n}\Phi_j^T\bSigma\Phi_T(\hat{\bbeta}_{NZT}-\bbeta^0_T)\right\|_2>\frac{\lambda_{n2}w_{nj}}{3},\ \exists\ j\notin T\right)\\
\leq&\prob\left(\left\|\frac{1}{n}\Phi_j^T\bSigma\Phi_T(\hat{\bbeta}_{NZT}-\bbeta^0_T)\right\|_2>C\lambda_{n2}r_n,\ \exists\ j\notin T\right)+o(1)\\
\leq&\prob\left(\max_{j\notin T}\left\|\frac{1}{n}\Phi_j^T\bSigma\Phi_T\right\|_2>\frac{C\lambda_{n2}r_n}{a_n\sqrt{s_nm_n/n}}\right)+\prob\left(\left\|\hat{\bbeta}_{NZT}-\bbeta^0_{T}\right\|_2>a_n\sqrt{\frac{s_nm_n}{n}}\right)+o(1)\\
\rightarrow&0\ \text{as}\ n\rightarrow\infty
\end{align*}
for any diverging sequence $a_n$, where the first probability in the last step goes to 0 by assumption 1 that the left hand side is of order $m_n^{-1/2}$ and assumption 6. The second probability goes to 0 by the rate of convergence of $\hat{\bbeta}_{NZT}$.

Therefore, we have that $\hat{\bbeta}_{NZ}$ is our adaptive group lasso solution with probability converging to 1. The components selected by adaptive group lasso is asymptotically at most those which are actually nonzero. Then we want to prove that the true nonzero components are all selected with probability converging to 1. By our assumptions, we have
\begin{align*}
\min_{j\in T}\|\hat{\bbeta}_{NZj}\|_2&\geq\min_{j\in T}\|\bbeta_{j}^0\|_2-\|\hat{\bbeta}_{NZj}-\bbeta_j^0\|_2\\
&\geq c_2^{-1/2}m_n^{1/2}c_{f,n}-o_P(1)\\
&>0
\end{align*}
Therefore, none of the true nonzero components are estimated as zero. Combining the two results above, we have that with probability converging to 1, the components selected by the adaptive group lasso are exactly the true nonzero components, i.e.,
$$\prob\left(\hat{\bbeta}_{AGL}\overset{0}{=}\bbeta^0\right)\rightarrow 1\ \text{as}\ n\rightarrow\infty$$
Part (i) is proved. Then we look at part (ii), where based on the result in part (i), we only consider the high probability event that the selection of the adaptive group lasso estimator is perfect. Similar to part (ii) of theorem \ref{screening1}, we consider a convex combination of $\bbeta^0$ and $\hat{\bbeta}_{AGL}$
$$\bbeta^*=t\hat{\bbeta}_{AGL}+(1-t)\bbeta^0$$
where $t=N_n/(N_n+\|\hat{\bbeta}_{AGL}-\bbeta^0\|_2)$ for some sequence $N_n$. Similar to (\ref{beforerhsbd}), we have
\begin{equation}
\begin{split}
\frac{1}{2n}(\bbeta^*_T-\bbeta^0_T)^T&\Phi_T^T\bSigma\Phi_T(\bbeta^*_T-\bbeta^0_T)\leq\frac{(\bbeta^*_T-\bbeta^0_T)^T\Phi^T_T(\boldy-\bmu_y)}{n}\\
&+\frac{(\bbeta^*_T-\bbeta^0_T)^T\Phi^T_T(\bmu_y-\bmu_y^*)}{n}+\lambda_{n2}\sum_{j=1}^{s_n}w_{nj}(\|\bbeta^0\|_2-\|\bbeta^*\|_2)
\end{split}
\end{equation}
Then by the fact that $|a^Tb|\leq\|a\|_2^2+\|b\|_2^2/4$, we have
\begin{align*}
&\frac{1}{2n}(\bbeta^*_T-\bbeta^0_T)^T\Phi_T^T\bSigma\Phi_T(\bbeta^*_T-\bbeta^0_T)\\
\leq&\frac{(\bbeta^*_T-\bbeta^0_T)^T\Phi^T_T(\boldy-\bmu_y)}{n}+\frac{1}{4n}(\bbeta^*_T-\bbeta^0_T)^T\Phi_T^T\bSigma\Phi_T(\bbeta^*_T-\bbeta^0_T)\\
&+\frac{\|\bSigma^{-1/2}(\bmu_y-\bmu_y^*)\|_2^2}{n}+\lambda_{n2}\sum_{j=1}^{s_n}w_{nj}(\|\bbeta^0\|_2-\|\bbeta^*\|_2)
\end{align*}
Then by (\ref{splineapproxerr}),
$$\frac{1}{4n}(\bbeta^*_T-\bbeta^0_T)^T\Phi_T^T\bSigma\Phi_T(\bbeta^*_T-\bbeta^0_T)\leq\frac{(\bbeta^*_T-\bbeta^0_T)^T\Phi^T_T(\boldy-\bmu_y)}{n}+O(s_n^2m_n^{-2d})+\lambda_{n2}\sum_{j=1}^{s_n}w_{nj}(\|\bbeta^0\|_2-\|\bbeta^*\|_2)$$
By (\ref{restrictedstrongconvexity}), the fact that $|a^Tb|\leq\|a\|_2^2+\|b\|_2^2/4$ and norm inequality, we have
\begin{align*}
\frac{\gamma_0c_1\gamma_2^{2s_n}m_n^{-1}}{4}\|\bbeta^*-\bbeta^0\|_2^2\leq&\frac{(\bbeta^*_T-\bbeta^0_T)^T\Phi^T_T(\boldy-\bmu_y)}{n}+O(s_n^2m_n^{-2d})+\frac{2(\max_{j\in T}w_{nj})^2}{\gamma_0c_1}\lambda_{n2}^2s_n\\
&+\frac{\gamma_0c_1\gamma_2^{2s_n}m_n^{-1}}{8}\|\bbeta^*-\bbeta^0\|_2^2
\end{align*}
Then by assumption 6,
$$\frac{\gamma_0c_1\gamma_2^{2s_n}m_n^{-1}}{8}\|\bbeta^*_T-\bbeta^0_T\|_2^2\leq\frac{(\bbeta^*_T-\bbeta^0_T)^T\Phi^T_T(\boldy-\bmu_y)}{n}+O(s_n^2m_n^{-2d})+O(\lambda_{n2}^2s_n)$$
Use the fact that $|a^Tb|\leq\|a\|_2^2+\|b\|_2^2/4$ on the first term of the right hand side, we have
\begin{align*}
\frac{\gamma_0c_1\gamma_2^{2s_n}m_n^{-1}}{8}\|\bbeta^*_T-\bbeta^0_T\|_2^2\leq&\frac{\gamma_0c_1\gamma_2^{2s_n}m_n^{-1}}{16}\|\bbeta^*_T-\bbeta^0_T\|_2^2\\
&+\frac{4}{\gamma_0c_1\gamma_2^{2s_n}m_n^{-1}n^2}\|\Phi_T^T(\boldy-\bmu_y)\|_2^2+O(s_n^2m_n^{-2d})+O(\lambda_{n2}^2s_n)
\end{align*}
By norm inequality and lemma \ref{highprob1}, we have
$$\frac{4}{\gamma_0c_1\gamma_2^{2s_n}m_n^{-1}n^2}\|\Phi_T^T(\boldy-\bmu_y)\|_2^2\leq\frac{4}{\gamma_0c_1\gamma_2^{2s_n}m_n^{-1}n^2}s_nm_n\|\Phi_T^T(\boldy-\bmu_y)\|_{\infty}=O_P\left(s_n\gamma_2^{-2s_n}m_n\frac{\log(s_nm_n)}{n}\right)$$
Combine the last two results, we have with probability converging to 1,
$$\|\bbeta^*_T-\bbeta^0_T\|_2^2=O_p\left(s_n\gamma_2^{-2s_n}m_n^2\frac{\log(s_nm_n)}{n}\right)+O(s_n^2\gamma_2^{-2s_n}m_n^{1-2d})+O(\lambda_{n2}^2m_n^2s_n\gamma_2^{-2s_n})$$
Then similar to the argument in the proof of part (ii) of theorem \ref{screening1}, we have
$$\sum_{j\in T}\|\hat{\bbeta}_{AGLj}-\bbeta^0_j\|_2^2=O_p\left(s_n\gamma_2^{-2s_n}m_n^2\frac{\log(s_nm_n)}{n}\right)+O(s_n^2\gamma_2^{-2s_n}m_n^{1-2d})+O(\lambda_{n2}^2m_n^2s_n\gamma_2^{-2s_n})$$
In the unbounded response case, we replace lemma \ref{highprob1} with lemma \ref{highprob2} and get
$$\sum_{j\in T}\|\hat{\bbeta}_{AGLj}-\bbeta^0_j\|_2^2=O_p\left(s_n\gamma_2^{-2s_n}m_n^2a_n\frac{\log(s_nm_n)}{n}\right)+O(s_n^2\gamma_2^{-2s_n}m_n^{1-2d})+O(\lambda_{n2}^2m_n^2s_n\gamma_2^{-2s_n})$$
for any diverging sequence $a_n$. Part (ii) is proved.
\end{proof}

\noindent
\textbf{Proof of theorem \ref{GICconsistency}}
\begin{proof}
The idea of the proof is similar to the proofs in \cite{fan2013tuning}, but due to the group penalization structure, some changes have to be made. First, the GIC criterion has the solution of adaptive group lasso, which is not easy to study. So we use a proxy, the MLE on the nonzero components selected by the adaptive group lasso estimator. Let
\begin{equation}
\hat{\bbeta}^*(A)=\argmax_{\{\bbeta\in\R^{p_nm_n}:\text{supp}_B(\bbeta)=A\}}\frac{1}{n}\sum_{i=1}^n\left[y_i\left(\bbeta^T\Phi_i\right)-b\left(\bbeta^T\Phi_i\right)\right]
\end{equation}
for a given $A\subset\{1,...,p\}$, and the proxy of GIC is defined as
\begin{equation}
GIC_{a_n}^*(A)=\frac{1}{n}\{D(\hat{\mu}^*_A;\boldY)+a_n|A|\}
\end{equation}
where $\hat{\mu}^*_A=b'(\Phi\hat{\bbeta}^*(A))$. The first result is that the proxy $GIC_{a_n}^*(T)$ well approximates $GIC_{a_n}(\lambda_{0})$. To prove this, observe by the definition of $\hat{\bbeta}_0=\hat{\bbeta}^*(T)$, we have the first order necessary condition
\begin{equation}
\frac{\partial}{\partial\bbeta}l_n(\hat{\bbeta}_0)=\boldzero
\end{equation}
Use Taylor expansion and by assumptions 1 and 2, we have
\begin{align}
\label{GICproxybound}
0&\geq GIC_{a_n}^*(T)-GIC_{a_n}(\lambda_{0})\nonumber\\
&=\frac{1}{n}\left(l_n(\hat{\bbeta}(\lambda_{n0}))-\l_n(\hat{\bbeta}_0)\right)\nonumber\\
&=-\frac{1}{n}\left(\hat{\bbeta}(\lambda_{n0})-\hat{\bbeta}_0\right)^T\Phi^T\bSigma(\bbeta^*)\Phi\left(\hat{\bbeta}(\lambda_{n0})-\hat{\bbeta}_0\right)\nonumber\\
&\geq -c_1\gamma_0\left\|\hat{\bbeta}(\lambda_{n0})-\hat{\bbeta}_0\right\|_2^2
\end{align}
where $\bbeta^*$ lies on the line segment joining $\hat{\bbeta}(\lambda_{n0})$ and $\hat{\bbeta}_0$. Then we need to bound $\left\|\hat{\bbeta}(\lambda_{n0})-\hat{\bbeta}_0\right\|_2^2$. By the definition of $\hat{\bbeta}(\lambda_{n0})$, we have
\begin{equation}
\Phi_T^T\left(\boldy-b'(\Phi_T\hat{\bbeta}_T(\lambda_{n0}))\right)+n\lambda_{n0}\bnu_T=\boldzero
\end{equation}
where the elements of $\bnu_T$ are $w_{nj}\hat{\bbeta}_j(\lambda_{n0})/\|\hat{\bbeta}_j(\lambda_{n0})\|_2$ for $j\in T$. On the other hand, by the definition of $\hat{\bbeta}_0$, we have
\begin{equation}
\Phi^T_T\left(\boldy-b'(\Phi_T\hat{\bbeta}_{0T})\right)=\boldzero
\end{equation}
Together we have
\begin{equation}
\Phi_T^T\left(b'(\Phi_T\hat{\bbeta}_{0T})-b'(\Phi_T\hat{\bbeta}_T(\lambda_{n0}))\right)+n\lambda_{n0}\bnu_T=\boldzero
\end{equation}
Use Taylor expansion on the left hand side of the equation, we have
\begin{equation}
\Phi^T_T\bSigma(\bbeta^{**})\Phi_T\left(\hat{\bbeta}_T(\lambda_{n0})-\hat{\bbeta}_{0T}\right)=n\lambda_{n0}\bnu_T
\end{equation}
where $\bbeta^{**}$ lies on the line segment joining $\hat{\bbeta}_T(\lambda_{n0})$ and $\hat{\bbeta}_{0T}$. Taking 2 norm and together with assumptions 1 and 2 and the results in theorem \ref{screening1}, we have
\begin{equation}
\left\|\hat{\bbeta}_T(\lambda_{n0})-\hat{\bbeta}_{0T}\right\|_2\leq C\lambda_{n0}\|\boldw_T\|_2\leq C\lambda_{n0}\sqrt{s_n}\|\boldw_T\|_{\infty}
\end{equation}
where $\boldw_T=(w_{nj},j\in T)'$. Then we have
\begin{equation}
\|\hat{\bbeta}(\lambda_{n0})-\hat{\bbeta}_0\|_2=O(\lambda_{n0}\sqrt{s_n})
\end{equation}
Choose $a_n$ to be any diverging sequence, then we have
\begin{equation}
\|\hat{\bbeta}(\lambda_{n0})-\hat{\bbeta}_0\|_2=o(\lambda_{n0}\sqrt{s_na_n})
\end{equation}
Then by (\ref{GICproxybound}), we have
\begin{equation}
GIC_{a_n}(\lambda_{0})-GIC_{a_n}^*(T)=o(\lambda_{n0}\sqrt{s_na_n})
\end{equation}
As a direct result,
\begin{align}
\label{GICapprox}
GIC_{a_n}(\lambda)-GIC_{a_n}(\lambda_{n0})&\geq (GIC_{a_n}^*(\alpha_{\lambda})-GIC_{a_n}^*(T))+(GIC_{a_n}^*(T)-GIC_{a_n}(\lambda_{n0}))\nonumber\\
&=(GIC_{a_n}^*(\alpha_{\lambda})-GIC_{a_n}^*(T))+o_p(\lambda_{n0}\sqrt{s_na_n})
\end{align}
The using this proxy, next we prove that the proxy $GIC^*$ is able to detect the distance between a selected model and the true model. Since the $GIC^*$ depends only on the MLE and has nothing to do with the penalization, this is the same as the generalised linear model, but with the spline line approximation error being considered.

Due to the estimation problem, we are only interested in the models $A$ such that $|A|\leq K$ where $Km_n=o(n)$. As the proof in \cite{fan2013tuning}, we consider the underfitted model and overfitted model (defined in their paper). Briefly, the underfitted models are $A$ such that $A\not\supset T$ and the overfitted models are $A$ such that $A\supsetneq T$. Also in the result of theorem \ref{screening1}, the model size $|A|=O(s_n)=o(n)$ and thus the KL divergence has a unique minimiser for every such model $A$, as discussed in \cite{fan2013tuning}.

Lemma \ref{gicunderfitted} implies that for all underfitted models
\begin{align*}
GIC_{A_n}^*(A)-GIC_{a_n}^*(T)&=2|A|I(\bbeta^*(A))+(|A|-|T|)a_nn^{-1}+|A|O_P(R_n)\\
&\geq \delta_n-s_na_nn^{-1}-O_P(KR_n)\\
&\geq\frac{\delta_n}{2}
\end{align*}
if $\delta_nK^{-1}R_n^{-1}\rightarrow\infty$ and $a_n=o(\delta_ns_n^{-1}n)$. This result states that there is a negligible increment on the $GIC^*$ if one of the nonzero component is missed, when the parameters satisfy the conditions. Lemma \ref{gicoverfitted} implies that for all overfitted models
$$GIC_{a_n}^*(A)-GIC_{a_n}^*(T)=\frac{|A|-|T|}{n}[a_n-O_P(\psi_n)]>\frac{a_n}{2n}$$
if $a_n\psi_n\rightarrow\infty$. This result states that there is a negligible increment on the $GIC^*$ if one of the zero component is selected along with the true model, when the parameters satisfy the conditions. Therefore,
\begin{equation}
\prob\left(\inf_{A\not\supset T}GIC^*_{a_n}(A)-GIC^*_{a_n}(T)>\frac{\delta_n}{2}\ \text{and}\ \inf_{A\supsetneq T}GIC^*_{a_n}(A)-GIC^*_{a_n}(T)>\frac{a_n}{2n}\right)\rightarrow 1
\end{equation}
Combine this result with (\ref{GICapprox}) and theorem assumptions, we have
$$\prob\{\inf_{\lambda\in\Omega_-\cup\Omega_+}GIC_{a_n}(\lambda)>GIC_{a_n}(\lambda_{n0})\}\rightarrow 1$$
\end{proof}
\begin{lemma}
\label{gicunderfitted}
Under assumptions 2 and 3, as $n\rightarrow\infty$, we have
$$\sup_{\substack{|A|\leq K\\A\subset\{1,...,p_n\}}}\frac{1}{n|A|}\left|D(\hat{\bmu}_{A}^*;\boldY)-D(\hat{\bmu}_0^*;\boldY)-2I(\bbeta^*(A))\right|=O_P(R_n)$$
where either a) the responses are bounded or Gaussian distributed, $R_n=\sqrt{\gamma_nm_n\log(p_n)/n}$, and $m_n\log(p_n)=o(n)$; or b) the responses are unbounded and non-Gaussian distributed, $R_n=\sqrt{\gamma_nm_n\log(p_n)/n}+\gamma_n^2m_nM_n^2\log(p_n)/n$ and $\log(p)=o(\min\{n(\log n)^{-1}K^{-2}m_n^{-1}\gamma_n^{-1},nM_n^{-2}\})$.
\end{lemma}
\begin{proof}
lemma \ref{gicunderfitted} is a direct result from lemma \ref{prop4} and lemma \ref{prop5}.
\end{proof}

\begin{lemma}
\label{gicoverfitted}
Under assumption 1, 2 and 3, and suppose $\log p=O(n^{\kappa})$ for some $0<\kappa<1$, as $n\rightarrow\infty$, we have
$$\frac{1}{|A|-|T|}\left(D(\hat{\bmu}_{A}^*;\boldY)-D(\hat{\bmu}_0^*;\boldY)\right)=O_P(\psi_n)$$
uniformly for all $A\supsetneq T$ with $|A|<K$ and either a) $\psi_n=m_n\sqrt{\gamma_n\log (p_n)}$ when the responses are bounded, $K=O(\min\{n^{(1-2\kappa)/6},n^{(1-3\kappa)/8}\})$ and $\kappa\leq 1/2$; or b) $\psi_n=m_n\gamma_n\log (p_n)$ when the responses are Gaussian bounded; or when the response are unbounded and non-Gaussian distributed, and the last three terms in lemma \ref{prop3} are dominated by $m_n\gamma_n\log p_n$.
\end{lemma}
\begin{proof}
lemma \ref{gicoverfitted} is a direct result from lemma \ref{lemma3} and \ref{prop3}.
\end{proof}
\begin{lemma}
\label{lemma2}
Under assumptions 2-3, let $\gamma_n$ be a slowly diverging sequence, if $\gamma_nL_n\sqrt{Km_n\log P_n/n}\rightarrow 0$ as $n\rightarrow\infty$, where $L_n=O(1)$ for the bounded case and $L_n=O(M_n+\sqrt{\log n})$ for the unbounded case, then we have
$$\sup_{|A|\leq K}\frac{1}{|A|}Z_{A}\left(\gamma_nL_n\sqrt{|A|m_n\frac{\log p_n}{n}}\right)=O_P\left(\gamma_n^2L_n^2\frac{m_n\log p_n}{n}\right)$$
where
$$Z_A(N)=\sup_{\bbeta\in\mathcal{B}_A(N)}\frac{1}{n}\left|l_n(\bbeta)-l_n(\bbeta^*(A))-E\left[l_n(\bbeta)-l_n(\bbeta^*(A))\right]\right|$$
and
$$\mathcal{B}_A(N)=\left\{\bbeta\in\R^P:\|\bbeta-\bbeta^*(A)\|_2\leq N,\text{supp}_B(\bbeta)=A\right\}\cup\left\{\bbeta^*(A)\right\}$$
\end{lemma}
\begin{proof}
Define
$$\Omega_n=\{\|\bepsilon\|_{\infty}\leq\tilde{L_n}\}$$
If we take $\tilde{L_n}=C\sqrt{\log n}$, \cite{fan2013tuning} has showed that $\prob(\Omega_n)\rightarrow 1$. Let
$$\tilde{Z}_A(N)=\sup_{\bbeta\in\mathcal{B}_A(N)}\frac{1}{n}\left|l_n(\bbeta)-l_n(\bbeta^*(A))-E\left[l_n(\bbeta)-l_n(\bbeta^*(A))|\Omega_n\right]\right|$$
Then we have
$$\sup_{|A|\leq K)}\frac{1}{|A|}Z_A(N)\leq \sup_{|A|\leq K)}\frac{1}{|A|}\tilde{Z}_A(N)+\sup_{|A|\leq K,\bbeta\in\mathcal{B}_A(N)}\frac{1}{|A|}R_A(\bbeta)$$
where
$$R_A(\bbeta)=\frac{1}{n}\left|E\left[l_n(\bbeta)-l_n(\bbeta^*(A))\right]-E\left[l_n(\bbeta)-l_n(\bbeta^*(A))|\Omega_n\right]\right|$$
By the definition of $l_n$, we have
\begin{align*}
R_A(\bbeta)&=\frac{1}{n}\left|E[\bepsilon|\Omega_n]^T\Phi(\bbeta-\bbeta^*(A))\right|\\
&\leq\frac{1}{n}\left\|E[\bepsilon|\Omega_n]\right\|_2\left\|\Phi(\bbeta-\bbeta^*(A)\right\|_2\\
&=\sqrt{\frac{1}{n}\sum_{i=1}^n(E[\epsilon_i|\Omega_n])^2}\cdot\frac{1}{\sqrt{n}}\left\|\Phi(\bbeta-\bbeta^*(A)\right\|_2\\
&\leq C\tilde{L}_n\exp(-C\tilde{L}_n)\|\bbeta-\bbeta^*(A)\|_2
\end{align*}
where the first inequality is Cauchy-Schwartz inequality, and the second inequality is lemma 1 in \cite{fan2013tuning} and assumption 1. Then we have
$$\sup_{|A|\leq K,\bbeta\in\mathcal{B}_A(N)}\frac{1}{|A|}R_A(\bbeta)=C\tilde{L}_n\exp(-C\tilde{L}_n)N$$
Taking $\tilde{L}_n=C\sqrt{\log n}$, $N=\gamma_nL_n\sqrt{|A|\log(p_nm_n)/n}$ and under the lemma assumption, we have
\begin{equation}
\sup_{|A|\leq K,\bbeta\in\mathcal{B}_A(N)}\frac{1}{|A|}R_A(\bbeta)=o(\log(p_nm_n)/n)
\end{equation}
Then let's consider $\tilde{Z}_A(N)$. For any $\bbeta_1,\bbeta_2\in\mathcal{B}_A(N)$, by the mean value theorem, we have $b(\Phi_i^T\bbeta_1)-B(\Phi_i^T\bbeta_2)=b'(\Phi_i^T\tilde{\bbeta})\Phi_i^T(\bbeta_1-\bbeta_2)$, where $\tilde{\bbeta}$ lies on the line segment joining $\bbeta_1$ and $\bbeta_2$. We have the likelihood function
\begin{align*}
&|-y_i\Phi_i^T\bbeta_1+b(\Phi_i^T\bbeta_1)-(-y_i\Phi_i^T\bbeta_2+b(\Phi_i^T\bbeta_2))|\\
=&|(-y_i+b'(\Phi_i^T\tilde{\bbeta}))|\Phi_i^T\bbeta_1-\Phi_i^T\bbeta_2|\\
\leq& (\tilde{L}_n+2M_n)|\Phi_i^T(\bbeta_1-\bbeta_2)|
\end{align*}
to be Lipschitz continuous. Let $w_1,...,w_n$ be a Rademacher sequence independent of $\bepsilon$. By the symmetrization theorem and the concentration inequality, see chapter 14 of \cite{Bhlmann:2011:SHD:2031491}, we have
\begin{align*}
E[\tilde{Z}_A(N)|\Omega_n]\leq& 2E\left[\sup_{\bbeta\in\mathcal{B}_A(N)}\frac{1}{n}\left|\sum_{i=1}^nw_i[-y_i\Phi_i^T\bbeta+b(\Phi_i^T\bbeta)-(-y_i\Phi_i^T\bbeta^*(A)+b(\Phi_i^T\bbeta^*(A)))]|\Omega_n\right|\right]\\
\leq&4L_nE\left[\sup_{\bbeta\in\mathcal{B}_A(N)}\frac{1}{n}\left|\sum_{i=1}^nw_i[\Phi_i(\bbeta-\bbeta^*(A)]|\Omega_n\right|\right]\\
\leq& 4L_nE\left[\left(\sup_{\bbeta\in\mathcal{B}_A(N)}\|\bbeta-\bbeta^*(A)\|_2\right)\left(\sum_{j\in A}\sum_{i=1}^n\sum_{k=1}^{m_n}\left|\frac{1}{n^2}(w_i\phi_{ijk})^2\right|\right)^{1/2}\right]\\
\leq&4L_nN\sqrt{\frac{|A|m_n}{n}}
\end{align*}
where the second last inequality is by Cauchy-Schwartz inequality, and the last inequality is by the definition of $\mathcal{B}_A(N)$ and $w_i$. Then since
$$\frac{1}{n}\sum_{i=1}^n(L_n\Phi_i^T(\bbeta(A)-\bbeta^0))^2\leq CL_n^2N^2$$
Apply Massart's inequality, see theorem 14.2 in \cite{Bhlmann:2011:SHD:2031491}, we have
$$\prob\left(\tilde{Z}_A(N)\geq E[\tilde{Z}_A(N)|\Omega_n]+t\right)\leq\exp\left(-\frac{1}{CL_n^2N^2}\frac{nt^2}{2}\right)$$
Take $t=4L_nNu\sqrt{|A|m_n/n}$ with $u>0$, $N=L_n\sqrt{|A|m_n/n}(1+u)$, $u=\gamma_n\sqrt{\log p_n}$ and observe that ${p_n\choose k}\leq(pe/k)^k$, we have
\begin{align*}
&\prob\left(\sup_{|A|\leq K}\frac{1}{|A|}\tilde{Z}_A(N)\geq 4L_n^2\frac{m_n}{n}(1+u)^2|\Omega_n\right)\\
\leq&\sum_{|A|\leq K}\prob\left(\tilde{Z}_A(N)\geq 4|A|L_n^2\frac{m_n}{n}(1+u)^2|\Omega_n\right)\\
\leq&\sum_{k\leq K}\left(\frac{pe}{k}\right)^k\exp(-CKm_nu^2)\\
\leq&\sum_{k\leq K}\left(\frac{pe}{k}\right)^k\exp(-CKm_n\gamma_n\log p_n)\rightarrow 0
\end{align*}
Then we have
$$\prob\left(\sup_{|A|\leq K}\frac{1}{|A|}\tilde{Z}_A(N)\geq \gamma_n^2L_n^2\frac{m_n}{n}\log p_n\right)=o(1)+\prob(\Omega^c)\rightarrow 0$$
\end{proof}
\begin{lemma}
\label{prop2}
Under assumptions 1-3, we have
$$\sup_{|A|\leq K}\frac{1}{\sqrt{|A|}}\|\hat{\bbeta}^*(A)-\bbeta^*(A)\|_2=O_P\left(\gamma_nL_n\sqrt{\frac{m_n\log p_n}{n}}\right)$$
\end{lemma}
\begin{proof}
Define the convex combination of $\hat{\bbeta}^*(A)$ and $\bbeta^*(A)$ to be the same way as we did in proving theorem \ref{screening1} as $\hat{\bbeta}_u(A)$. Then is remains to show
$$\sup_{|A|\leq K}\frac{1}{\sqrt{|A|}}\|\hat{\bbeta}_u(A)-\bbeta^*(A)\|_2=O_P\left(\gamma_nL_n\sqrt{\frac{m_n\log p_n}{n}}\right)$$
By the definition of $\hat{\bbeta}^*(A)$ and the concavity of the likelihood function, we have
$$l_n(\hat{\bbeta}_u(A))-l_n(\bbeta^*(A))\geq 0$$
By the definition of $\bbeta^*(A)$, we have
$$E[l_n(\bbeta^*(A)-l_n(\hat{\bbeta}_u(A)]\geq 0$$
Combine the two inequalities above, we have
\begin{equation}
\label{lemma2nZbound}
0\leq E[l_n(\bbeta^*(A)-l_n(\hat{\bbeta}_u(A)]\leq l_n(\hat{\bbeta}_u(A))-l_n(\bbeta^*(A))-E[l_n(\hat{\bbeta}_u(A)-l_n(\bbeta^*(A)]\leq nZ_A(N)
\end{equation}
On the other hand, for any $\bbeta_A\in\mathbb{B}_A(N)$, we have
\begin{align*}
E[l_n(\bbeta_A)-l_n(\bbeta^*(A))]&=E[\boldy^T\Phi\bbeta_A-\boldone^Tb(\Phi\bbeta_A)-\boldy^T\Phi\bbeta^*(A)+\boldone^Tb(\Phi\bbeta^*(A))]\\
&=b'(\sum_{j=1}^{p_n}f_j)^T\Phi[\bbeta_A-\bbeta^*(A)]-\boldone^T[b(\Phi\bbeta_A)-b(\Phi\bbeta^*(A))]
\end{align*}
Observe that by the definition of $\bbeta^*(A)$, we have
$$\Phi[b'(\sum_{j=1}^{p_n}f_j)-b'(\Phi\bbeta^*(A))]=\boldzero$$
use Taylor expansion, we have
\begin{align*}
E[l_n(\bbeta_A)-l_n(\bbeta^*(A))]&=b'(\Phi\bbeta^*(A))^T\Phi[\bbeta_A-\bbeta^*(A)]-\boldone^T[b(\Phi\bbeta_A)-b(\Phi\bbeta^*(A))]\\
&=-\frac{1}{2}(\bbeta_A-\bbeta^*(A))^T\Phi_A^T\tilde{\bSigma}\Phi_A(\bbeta_A-\bbeta^*(A))\\
&\leq Cn\|\bbeta_A-\bbeta^*(A)\|_2^2
\end{align*}
where the last inequality is by assumptions 1 and 2. Then we have
$$\|\bbeta_A-\bbeta^*(A)\|_2^2\leq CZ_A(N)$$
Take $N=\gamma_nL_n\sqrt{|A|m_n\log p_n/n}$ and by lemma \ref{lemma2}, we have
$$\sup_{|A|\leq K}\frac{1}{\sqrt{|A|}}\|\hat{\bbeta}_u(A)-\bbeta^*(A)\|_2=O_P\left(\gamma_nL_n\sqrt{\frac{m_n\log p_n}{n}}\right)$$
Then lemma \ref{prop2} follows.
\end{proof}
\begin{lemma}
\label{prop4}
Under assumptions 1-3, we have
$$\sup_{|A|\leq K}\frac{1}{n|A|}\left(l_n(\hat{\bbeta}^*(A))-l_n(\bbeta^*(A))\right)\leq \frac{\gamma_n^2L_n^2m_n\log p_n}{n}$$
\end{lemma}
\begin{proof}
Define the event
$$\mathcal{E}=\left\{\sup_{|A|\leq K}\frac{1}{\sqrt{|A|}}\|\hat{\bbeta}^*(A)-\bbeta^*(A)\|_2=O_P\left(\gamma_nL_n\sqrt{\frac{m_n\log p_n}{n}}\right)\right\}$$
By lemma \ref{prop2}, we have $\prob(\mathbb{E})\rightarrow 1$. Using the same argument as in (\ref{lemma2nZbound}) in proving lemma \ref{prop2}, we have
$$0\leq l_n(\hat{\bbeta}^*(A))-l_n(\bbeta^*(A))\leq l_n(\hat{\bbeta}_u(A))-l_n(\bbeta^*(A))-E[l_n(\hat{\bbeta}_u(A)-l_n(\bbeta^*(A)]\leq nZ_A(N)$$
By lemma \ref{lemma2}, conditioning on $\mathcal{E}$, we have
$$l_n(\hat{\bbeta}^*(A))-l_n(\bbeta^*(A))\leq nO_P\left(\gamma_n^2L_n^2\frac{|A|m_n\log p_n}{n}\right)$$
Then the lemma follow from $\prob(A)\leq\prob(A|\mathcal{E})+\prob(\mathcal{E}^c)$.
\end{proof}
\begin{lemma}
\label{prop5}
Under assumption 1-3, we have
$$\sup_{|A|\leq K}\frac{1}{n|A|}\left|l_n(\bbeta^*(A))-E[l_n(\bbeta^*(A))]\right|=O_P\left(\sqrt{\frac{\gamma_nm_n\log p_n}{n}}\right)$$
where $\log p_n=o(n)$ for bounded response and $\gamma_nm_nK^2\log p_n=o(n)$ for unbounded response.
\end{lemma}
\begin{proof}
By the definition, we have $l_n(\bbeta^*(A))-E[l_n(\bbeta^*(A))]=\bepsilon^T\Phi\bbeta^*(A)$. For bounded response, by Hoeffding's inequality, we have
\begin{align*}
\prob(|\bepsilon^T\Phi\bbeta^*(A)|\geq t)&\leq C\exp\left(\frac{Ct^2}{\sum_{i=1}^n(\Phi_i^T\bbeta^*(A))^2}\right)\\
&\leq C\exp\left(-\frac{Ct^2}{n|A|m_n}\right)
\end{align*}
Take $t=|A|\sqrt{n\gamma_nm_n\log p_n}$, we have
$$\prob(|\bepsilon^T\Phi\bbeta^*(A)|\geq|A|\sqrt{n\gamma_nm_n\log p_n})\leq C\exp(-C|A|\gamma_n\log p_n)$$
Then we have
$$\sup_{|A|\leq K}\frac{1}{n|A|}\left|l_n(\bbeta^*(A))-E[l_n(\bbeta^*(A))]\right|=O_P\left(\sqrt{\frac{\gamma_nm_n\log p_n}{n}}\right)$$
If the responses are unbounded, we use Bernstein's inequality. First check the condition
\begin{align*}
E[|\Phi_i\bbeta^*(A)\epsilon_i|^m]&=m\int_0^{\infty}x^{m-1}\prob(|\Phi_i\bbeta^*(A)\epsilon_i\geq x)dx\\
&=m|\Phi_i^T\bbeta^*(A)|^m\int_0^{\infty}\left(\frac{x}{|\Phi_i^T\bbeta^*(A)|}\right)^{m-1}\prob\left(|\epsilon_i|\geq\frac{x}{|\Phi_i^T\bbeta^*(A)|}\right)d\frac{x}{|\Phi_i^T\bbeta^*(A)|}\\
&\leq m|\Phi_i^T\bbeta^*(A)|^m\int_0^{\infty}t^{m-1}C\exp(-Ct^2)dt)\\
&\leq m|\Phi_i^T\bbeta^*(A)|^m(\|\Phi\bbeta^*(A)\|_{\infty}C)^{m-2}\frac{m!}{2}
\end{align*}
Then by Bernstein's inequality, we have
$$\prob(|\bepsilon^T\Phi\bbeta^*(A)|\geq\sqrt{n}t)\leq 2\exp\left(-\frac{1}{2}\frac{nt^2}{C\|\Phi_A\bbeta^*(A)\|_2^2+C\sqrt{n}\|\Phi_A\bbeta^*(A)\|_{\infty}t}\right)$$
Taking $t=|A|\sqrt{\gamma_nm_n\log p_n}$, we have
\begin{align*}
&\prob(|\bepsilon^T\Phi\bbeta^*(A)|\geq\sqrt{n}|A|\sqrt{\gamma_nm_n\log p_n})\\
\leq &2\exp\left(-\frac{1}{2}\frac{n|A|^2\gamma_nm_n\log p_n}{C\|\Phi_A\bbeta^*(A)\|_2^2+C\sqrt{n}\|\Phi_A\bbeta^*(A)\|_{\infty}|A|\sqrt{\gamma_nm_n\log p_n}}\right)\\
\rightarrow& 0
\end{align*}
if $K^2\gamma_nm_n\log p_n/n\rightarrow 0$.
\end{proof}
\begin{lemma}
\label{lemma3}
Under assumptions 1-3, we have
$$\sup_{\substack{A\supset T\\|A|\leq K}}\frac{1}{|A|-|T|}(\boldy-\bmu_0)^T\bSigma_0^{-1/2}\boldB_A\bSigma_0^{-1/2}(\boldy-\bmu_0)=O_P(m_n(\gamma_n\log p_n)^{\xi})$$
where
$$\boldB_A=\bSigma_0^{1/2}\Phi_A(\Phi_A^T\bSigma_0\Phi_A)^{-1}\Phi_A^T\bSigma_0^{1/2}$$
and $\xi=1/2$ for bounded response and $\xi=1$ for unbounded response.
\end{lemma}
\begin{proof}
Let $k=|A|-|T|$ and $\boldP_A=\boldB_A-\boldB_T$. It's easy to verify that $\boldP_A$ is a projection matrix, thus we have $tr(\boldP)=km_n$, $\sum_{i=1}^nP_{ii}=km_n$ and $\sum_{i,j}P_{ij}=km_n$. Let
$$\tilde{\boldy}=\bSigma_0^{-1/2}(\boldy-\bmu_0)$$
We have the decomposition
$$\frac{1}{m_nk}\tilde{\boldy}^T\boldP_A\tilde{\boldy}=\frac{1}{m_nk}\sum_{i=1}^nP_{ii}\tilde{y}_i^2+\frac{1}{m_nk}\sum_{i\neq j}P_{ij}\tilde{Y}_i\tilde{Y}_j\equiv I_1(A)+I_2(A)$$
Let $\tilde{y}_i^*$ be independent copies of $\tilde{y}_i$, then by the decoupling inequality, there exists a constant $C>0$ such that
$$\prob\left(\frac{1}{m_nk}|\sum_{i\neq j}P_{ij}\tilde{y}_i\tilde{y}_j|\geq t\right)\leq C\prob\left(\frac{1}{m_nk}|\sum_{i\neq j}P_{ij}\tilde{Y}_i\tilde{Y}_j^*|\geq C^{-1}t\right)$$
For bounded response, apply Hoeffding's inequality, we have
$$\prob(I_1(A)\geq 1+x)\leq 2\exp\left(2\frac{Cx^2}{\sum_{i=1}^n(m_nk)^{-2}P_{ii}^2}\right)\leq 2\exp(-Cm_nkx^2)$$
Taking $x=\sqrt{\gamma\log p_n}$, use the inequality ${p\choose k}\leq (pe/k)^k$ and use the same technique as we used in proving lemma \ref{lemma2}, we have
$$\prob\left(\sup_{|A|\leq K}I_1(A)\geq 1+\sqrt{\gamma_n\log p_n}\right)\leq 2C\sum_{k=1}^K\left(\frac{(p_n-s_n)e}{k}\right)^k\exp(-Cm_nk\gamma_n\log p_n)\rightarrow 0$$
Then observe $\sum_{i\neq j}P_{ij}^2=\sum_i(P_{ii}-P_{ii}^2)\leq m_nk$, we have following the decoupling inequality that
\begin{align*}
\prob(|I_2(A)|\geq t)&\leq C\prob\left(\frac{1}{m_nk}|\sum_{i\neq j}P_{ij}\tilde{Y}_i\tilde{Y}_j^*|\geq C^{-1}t\right)\\
&\leq C\exp\left(-\frac{C^{-2}(m_nk)^2t^2}{\sum_{i\neq j}P_{ij}^2}\right)\\
&\leq C\exp(-Cm_nkt^2)
\end{align*}
Taking $t=\sqrt{\gamma_n\log p_n}$ and use the same technique as in the previous step, we have
$$\prob\left(\sup_{|A|\leq K}I_2(A)\geq 1+\sqrt{\gamma_n\log p_n}\right)\rightarrow 0$$
In the unbounded case, we apply the Bernstein's inequality. In the same way as we did in proving lemma \ref{prop5}, we check the condition
$$E|P_{ii}\tilde{Y}_i^2|^m\leq m!C^{m-2}\frac{P_{ii}^2}{2}$$
By Bernstein's inequality, we have
$$\prob(I_1(A)\geq x^2)\leq 2\exp(-Cm_nkx^2)$$
Taking $x=\sqrt{\gamma_n\log p_n}$, we have
$$\sup_{|A|\leq K}I_1(A)=O_P(\gamma_n\log p_n)$$
For $I_2(A)$, we have
$$\sum_{i\neq j}|P_{ij}|^mE[|\tilde{\boldy}_i\tilde{\boldy}_j^*|^m]\leq m!C^{m-2}\frac{P_{ij}^2}{2}$$
Then by Berstein's inequality and taking $x=\sqrt{\gamma_n\log p_n}$, we have
$$\prob(|I_2(A)|\geq \gamma_n\log p_n)\rightarrow 0$$
\end{proof}
\begin{lemma}
\label{prop3}
Under assumptions 1-3, for all $A\supset T$ and $|A|\leq K$, we have
\begin{align*}
l_n(\hat{\bbeta}^*(A))-l_n(\bbeta^(A))&=\frac{1}{2}(\boldy-\bmu_0)^T\bSigma_0^{-1/2}\boldB_A\bSigma_0^{-1/2}(\boldy-\bmu_0)+|A|^{5/2}O_P\left(m_n^{5/2}\gamma_n^{5/2}L_n^2\frac{(\log p_n)^{1+\xi/2}}{\sqrt{n}}\right)\\
&+|A|^4O_P\left(m_n^4\gamma_n^4L_n^4\frac{(\log p_n)^2}{n}\right)+|A|^3O_P\left(m_n^3\gamma_n^3L_n^3\frac{(\log p_n)^{3/2}}{\sqrt{n}}\right)
\end{align*}
\end{lemma}
\begin{proof}
Use Taylor's expansion, we have
\begin{align*}
&l_n(\hat{\bbeta}^*(A))-l_n(\bbeta^*(A))\\
=&(\hat{\bbeta}^*(A)-\bbeta^*(A))^T\Phi^T(\boldy-b'(\Phi\bbeta^*(A))-\frac{1}{2}(\hat{\bbeta}^*(A)-\bbeta^*(A))^T\Phi^T\bSigma_0\Phi(\hat{\bbeta}^*(A)-\bbeta^*(A))+\text{Remainder}\\
\equiv&I_1(A)+I_2(A)+I_3(A)
\end{align*}
First, by the definition of $\hat{\bbeta}^*(A)$, we have
$$\Phi_A^T[\boldy-b'(\Phi\hat{\bbeta}^*(A))]=0$$
Then by Taylor expansion, we have
\begin{align*}
\Phi_A^T\boldy&=\Phi_A^Tb'(\Phi\hat{\bbeta}^*(A))\\
&=\Phi_A^Tb'(\Phi\bbeta^*(A))+\Phi_A^T\bSigma_0\Phi(\hat{\bbeta}^*(A)-\bbeta^*(A))+\Phi_A^T\bnu_A
\end{align*}
where $\nu_{Ai}=b'''(\Phi_i^T\tilde{\bbeta}^*(A))(\Phi_i^T(\hat{\bbeta}^*(A)-\bbeta^*(A)))^2/2$ and $\tilde{\bbeta}^*(A)$ lies on the line segment joining $\hat{\bbeta}^*(A)$ and $\bbeta^*(A)$. By the definition of $\bbeta^*(A)$, we have
$$\Phi_A^T[b'(\sum_{j=1}^{p_n}f_j)-b'(\Phi_A\bbeta^*(A))]=0$$
we have
$$\hat{\bbeta}^*(A)-\bbeta^*(A)=(\Phi_A^T\bSigma_0\Phi_A)^{-1}\Phi_A^T(\boldy-b'(\sum_{j=1}^{p_n}-\bnu_A))$$
Therefore, we have
$$I_1(A)=(\boldy-\bmu_y)^T\bSigma_0^{-1/2}\boldB_A\bSigma_0^{-1/2}(\boldy-\bmu_y)+R_{1,A}$$
where $R_{1,A}=-\bmu_A^T\bSigma_0^{-1/2}\boldB_A\bSigma_0^{-1/2}\bepsilon$. By Cauchy-Schwartz inequality, we have
\begin{align*}
|R_{1,A}|&\leq\|\boldB_A\bSigma_0^{-1/2}\bepsilon\|_2\|\bSigma_0^{-1/2}\bnu_A\|_2\\
&\leq (\|\boldB_T\bSigma_0^{-1/2}\bepsilon\|_2+\|\tilde{R}_{1,A}\|_2)\|\bSigma_0^{-1/2}\bnu_A\|_2
\end{align*}
where $\tilde{R}_{1,A}=(\boldB_A-\boldB_0)\bSigma_0^{-1/2}\bepsilon$. Observe that $\bSigma_0=E[\epsilon\epsilon^T]$ and $tr(\boldB_T\boldB_T)=m_ns_n$, take $\gamma_n\rightarrow\infty$, by Markov's inequality, we have
\begin{align*}
\prob\left(\|\boldB_T\bSigma_0^{-1/2}\bepsilon\|_2\geq\sqrt{m_ns_n\gamma_n}\right)&\leq\frac{1}{m_ns_n\gamma_n}E[\|\boldB_T\bSigma_0^{-1/2}\bepsilon\|_2^2]\\
&=\frac{1}{m_ns_n\gamma_n}tr\{\boldB_T\bSigma_0^{-1/2}E[\bepsilon\bepsilon^T]\bSigma_0^{-1/2}\boldB_T\}\\
&=\frac{1}{\gamma_n}\rightarrow 0
\end{align*}
Then we have
\begin{equation}
\label{lastlemma1}
\|\boldB_T\bSigma_0^{-1/2}\bepsilon\|_2=O_P(\sqrt{m_ns_n\gamma_n})
\end{equation}
By lemma \ref{lemma3}, we have
\begin{equation}
\label{lastlemma2}
(|A|-|T|)^{-1/2}\|\tilde{R}_{1,A}\|_2=O_P(m_n^{1/2}(\gamma_n\log p_n)^{\xi})
\end{equation}
Finally, we have
\begin{align}
\label{lastlemma3}
\|\bSigma_0^{-1/2}\bnu_A\|_2&\leq C\|\bnu_A\|_2\nonumber\\
&\leq C\left(\sum_{i=1}^n|\Phi_i^T(\hat{\bbeta}^*(A)-\bbeta)^*(A))|^4\right)^{1/2}\nonumber\\
&\leq C\left(\sum_{i=1}^n\|\Phi_{iA}\|_2^4\|\hat{\bbeta}^*(A)-\bbeta)^*(A)\|_2^4\right)^{1/2}\nonumber\\
&\leq Cm_n|A|n^{1/2}\|\hat{\bbeta}^*(A)-\bbeta)^*(A)\|_2^2\nonumber\\
&=m_n^2|A|^2O_P\left(\gamma_n^2L_n^2\frac{\log p_n}{\sqrt{n}}\right)
\end{align}
Combining (\ref{lastlemma1}), (\ref{lastlemma2}) and (\ref{lastlemma3}), we have
$$I_1(A)=(\boldy-\bmu_y)^T\bSigma_0^{-1/2}\boldB_A\bSigma_0^{-1/2}(\boldy-\bmu_y)+O_P\left(|A|^{5/2}m_n^{5/2}\gamma_n^{5/2}L_n^2\frac{(\log p_n)^{1+\xi/2}}{\sqrt{n}}\right)$$
Then we look at $I_2(A)$. We have
\begin{align*}
I_2(A)&=\frac{1}{2}(\hat{\bbeta}^*(A)-\bbeta^*(A))^T\Phi^T\bSigma_0\Phi(\hat{\bbeta}^*(A)-\bbeta^*(A))\\
&=\frac{1}{2}(\boldy-\bmu_y)^T\bSigma_0^{-1/2}\boldB_A\bSigma_0^{-1/2}(\boldy-\bmu_y)+\frac{1}{2}R_{2,A}-R_{1,A}
\end{align*}
where
\begin{align*}
R_{2,A}&=\bnu_A\bSigma_0^{-1/2}\boldB_A\bSigma_0^{-1/2}\bmu_A\\
&\leq C\|\bnu_A\|_2^2\\
&\leq Cm_n^2|A|^2n\|\hat{\bbeta}^*(A)-\bbeta^*(A)\|_2^4\\
&=O\left(m_n^2|A|^4\gamma_n^4L_n^4\frac{m_n^2(\log p_n)^2}{n^2}n\right)\\
&=O\left(|A|^4m_n^4\gamma_n^4L_n^4\frac{(\log p_n)^2}{n}\right)
\end{align*}
Therefore,
\begin{align*}
I_2(A)=&\frac{1}{2}(\boldy-\bmu_y)^T\bSigma_0^{-1/2}\boldB_A\bSigma_0^{-1/2}(\boldy-\bmu_y)+O_P\left(|A|^{5/2}m_n^{5/2}\gamma_n^{5/2}L_n^2\frac{(\log p_n)^{1+\xi/2}}{\sqrt{n}}\right)\\
&+O\left(|A|^4m_n^4\gamma_n^4L_n^4\frac{(\log p_n)^2}{n}\right)
\end{align*}
Finally, we have for $I_3(A)$ that
\begin{align*}
|I_3(A)|&\leq Cn|A|^{3/2}m_n^{3/2}\|\hat{\bbeta}^*(A)-\bbeta^*(A)\|_2^3\\
&=O_P\left(|A|^3m_n^3\gamma_n^3L_n^3\frac{(\log p_n)^{3/2}}{\sqrt{n}}\right)
\end{align*}
Combining the three results for $I_1(A)$, $I_2(A)$ and $I_3(A)$, we get the desired result.
\end{proof}

\end{appendices}

\end{document}